%% file: main.tex
\documentclass[sigconf,nonacm]{acmart}
\usepackage{multicol}
\usepackage{stmaryrd}
\usepackage[capitalize]{cleveref}
\usepackage{amsmath}
\usepackage{graphicx}
\usepackage{booktabs}
\usepackage{xspace}
\usepackage[vlined,commentsnumbered,ruled,linesnumbered]{algorithm2e}
\usepackage{multirow}
\usepackage{subfig}
\usepackage[table,dvipsnames]{xcolor}
\usepackage{soul}
\usepackage{mathtools}
\usepackage{footmisc}  % Add this in the preamble
\usepackage{wrapfig}
\usepackage{amsthm}
\usepackage{thmtools}
\usepackage{longtable}
\usepackage[inline]{enumitem}
\usepackage{makecell}

\AtBeginDocument{%
  }

\setcopyright{acmlicensed}
\copyrightyear{2018}
\acmYear{2018}
\acmDOI{XXXXXXX.XXXXXXX}
\acmConference[Conference acronym 'XX]{Make sure to enter the correct
  conference title from your rights confirmation email}{June 03--05,
  2018}{Woodstock, NY}
\acmISBN{978-1-4503-XXXX-X/2018/06}

\newcommand{\benny}[1]{{\texttt{\color{purple} (Benny)~ [{#1}]}}}

\newcommand{\shunit}[1]{{\color{violet}\bf (Shunit)~[#1]}{\typeout{#1}}}

\newcommand{\revb}[1]{{\leavevmode\color{black}{#1}}}
\newcommand{\revc}[1]{{\leavevmode\color{black}{#1}}}
\newcommand{\revcommon}[1]{{\leavevmode\color{black}{#1}}}

\newcommand{\red}[1]{{\color{red} {#1}}}

\newcommand{\cut}[1]{}

\DeclareMathOperator*{\argmax}{arg\,max}
\crefname{algocf}{alg.}{algs.}
\Crefname{algocf}{Algorithm}{Algorithms}

\newtheorem{theorem}{Theorem}[section]

\begin{document}
% \acmReviewOff

\input{macros}

% \input{revision_letter}
% \acmReviewOn
\title{Analyzing Deviations from Monotonic Trends \\ through Database Repair}

\def\shorttitle{Analyzing Deviations from Monotonic Trends through Database Repair}

% \affiliation{%

\author{Shunit Agmon}
\orcid{0000-0001-9605-4131}
\affiliation{%
  % \institution{Technion -- Israel Institution of Technology}
  \institution{Technion}
  \streetaddress{}
  \country{Israel}
  % \city{Haifa}
  % \state{Israel}
  % \postcode{3200003}
}
\email{shunit.agmon@gmail.com}

\author{Jonathan Gal}
\orcid{0000-0002-3948-5266}
\affiliation{%
  \institution{Technion}
  % \city{Jerusalem}
  \country{Israel}
}
\email{jonathan.gal@campus.technion.ac.il}

\author{Amir Gilad}
\orcid{0000-0002-3764-1958}
\affiliation{%
  \institution{Hebrew University}
  % \city{Jerusalem}
  \country{Israel}
}
\email{amirg@cs.huji.ac.il}

\author{Ester Livshits}
\orcid{0000-0003-3485-9887}
\affiliation{%
  % \institution{Technion -- Israel Institution of Technology}
  \institution{Technion}
  % \city{Haifa}
  \country{Israel}
  % \postcode{3200003}
}
\email{esterlivshits@gmail.com}

\author{Or Mutay}
\orcid{0009-0006-2997-4775}
\affiliation{%
  % \institution{Technion -- Israel Institution of Technology}
  \institution{Technion}
  % \city{Haifa}
  \country{Israel}
  % \postcode{3200003}
}
\email{or.mutay@campus.technion.ac.il}

\author{Brit Youngmann}
\orcid{0000-0002-0031-5550}
\affiliation{%
  % \institution{Technion -- Israel Institution of Technology}
  \institution{Technion}
  % \city{Haifa}
  \country{Israel}
  % \postcode{3200003}
}
\email{brity@technion.ac.il}

\author{Benny Kimelfeld}
\orcid{0000-0002-7156-1572}
\affiliation{%
  % \institution{Technion -- Israel Institution of Technology}
  \institution{Technion}
  % \streetaddress{}
  % \city{Haifa}
  \country{Israel}
}
\email{bennyk@cs.technion.ac.il}
% \affiliation{%

% \affiliation{%

% \affiliation{%

% \affiliation{%

% \affiliation{%

% \affiliation{%

% \affiliation{%

\renewcommand{\shortauthors}{Anonymous authors}

\begin{abstract}
Datasets often exhibit violations of expected monotonic trends—for example, higher education level correlating with higher average salary, newer homes being more expensive, or diabetes prevalence increasing with age. We address the problem of quantifying how far a dataset deviates from such trends. To this end, we introduce Aggregate Order Dependencies (AODs), an aggregation-centric extension of the previously studied order dependencies. An AOD specifies that the aggregated value of a target attribute (e.g., mean salary) should monotonically increase or decrease with the grouping attribute (e.g., education level).

We formulate the AOD repair problem as finding the smallest set of tuples to delete from a table so that the given AOD is satisfied. We analyze the computational complexity of this problem and propose a general algorithmic template for solving it. We instantiate the template for common aggregation functions, introduce optimization techniques that substantially improve the runtime of the template instances, and develop efficient heuristic alternatives.
Our experimental study, carried out on both real-world and synthetic datasets, demonstrates the practical efficiency of the algorithms and provides insight into the performance of the heuristics. We also present case studies that uncover and explain unexpected AOD violations using our framework.
\end{abstract}


\begin{CCSXML}
<ccs2012>
   <concept>
       <concept_id>10002951.10002952</concept_id>
       <concept_desc>Information systems~Data management systems</concept_desc>
       <concept_significance>500</concept_significance>
       </concept>
 </ccs2012>
\end{CCSXML}

\ccsdesc[500]{Information systems~Data management systems}



\keywords{Database repair, aggregate order constraints, trend analysis}


\maketitle

\setcounter{page}{1}

% \section{Introduction}
% \input{macros}
\input{01-intro}
\input{06-related_work}

\input{02-formal}

\input{03-tuple_del}
\input{05-experiments-1}

\input{conclusions}

%\begin{acks}
%To Robert, for the bagels and explaining CMYK and color spaces.
%\end{acks}

\bibliographystyle{ACM-Reference-Format}
\bibliography{trends}


\newpage
\clearpage 
\appendix

% for sections 3-4
\input{10-app-proofs}

% section 5
\input{10-app-sum_opt_example}

% section 6
\input{10-app-greedy_optimizations}

% section 7.5

\input{10-app-outlier_removal}

\input{10-app-case_studies}

%section 7.4

\input{10-app-greedy_median}
% TODO mention it in 7.3

\input{10-graphs_with_CIs}

%\balance

%\newpage

\end{document}

%% file: macros.tex
\def\scs{\mathcal{S}}
\def\agg#1{\mathord{\mathsf{#1}}}
\def\maxagg{$\agg{max}$\xspace}
\def\minagg{$\agg{min}$\xspace}
\def\sumagg{$\agg{sum}$\xspace}
\def\avgagg{$\agg{avg}$\xspace}
\def\medianagg{$\agg{median}$\xspace}
\def\countagg{$\agg{count}$\xspace}
\def\countdagg{$\agg{countd}$\xspace}

\def\emaxagg{\agg{max}}
\def\eminagg{\agg{min}}
\def\esumagg{\agg{sum}}
\def\eavgagg{\agg{avg}}
\def\emedianagg{\agg{median}}
\def\ecountagg{\agg{count}}
\def\ecountdagg{\agg{countd}}

\def\atts{\mathit{Att}}
\def\set#1{\mathord{\{#1\}}}
\def\dbr#1{\mathord{\llbracket #1 \rrbracket}}
\def\bottom{\underline{b}}
\def\top{\overline{b}}
\def\procaggname{\mathsf{WholePack}}
\def\procagg#1{$\procaggname\langle{\mathord{\mathsf{#1}}}\rangle$}
\def\procaggalpha{$\procaggname\langle{\mathord{\alpha}\rangle}$\xspace}
\def\tupleprob{$\mathsf{TTD}\langle{\scs,\alpha,\delta}\rangle$\xspace}
\def\predprob{$\mathsf{TPD}\langle{\scs,\alpha,\delta,\kappa}\rangle$\xspace}
\def\aod{AOD\xspace}

\def\dpalg{$\mathsf{CardRepair}$\xspace}
\def\greedyalg{$\mathsf{HeurRepair}$\xspace}
\def\crepair{C-repair\xspace}
\def\crepairs{C-repairs\xspace}

\def\addtuples{noise tuples\xspace}
\newcommand{\dom}{{\tt dom}}
\newcommand{\bigo}{\ensuremath{\mathcal{O}}}
\newcommand{\pattern}{\ensuremath{\mathbb{P}}}

\newcommand{\hide}[1]{}

\def\ne{{\nearrow}}
\def\se{{\searrow}}

\theoremstyle{definition} 
\newtheorem{remark}{Remark}

%% file: 01-intro.tex
\begin{figure}
    \centering
    \includegraphics[width=0.9\linewidth]{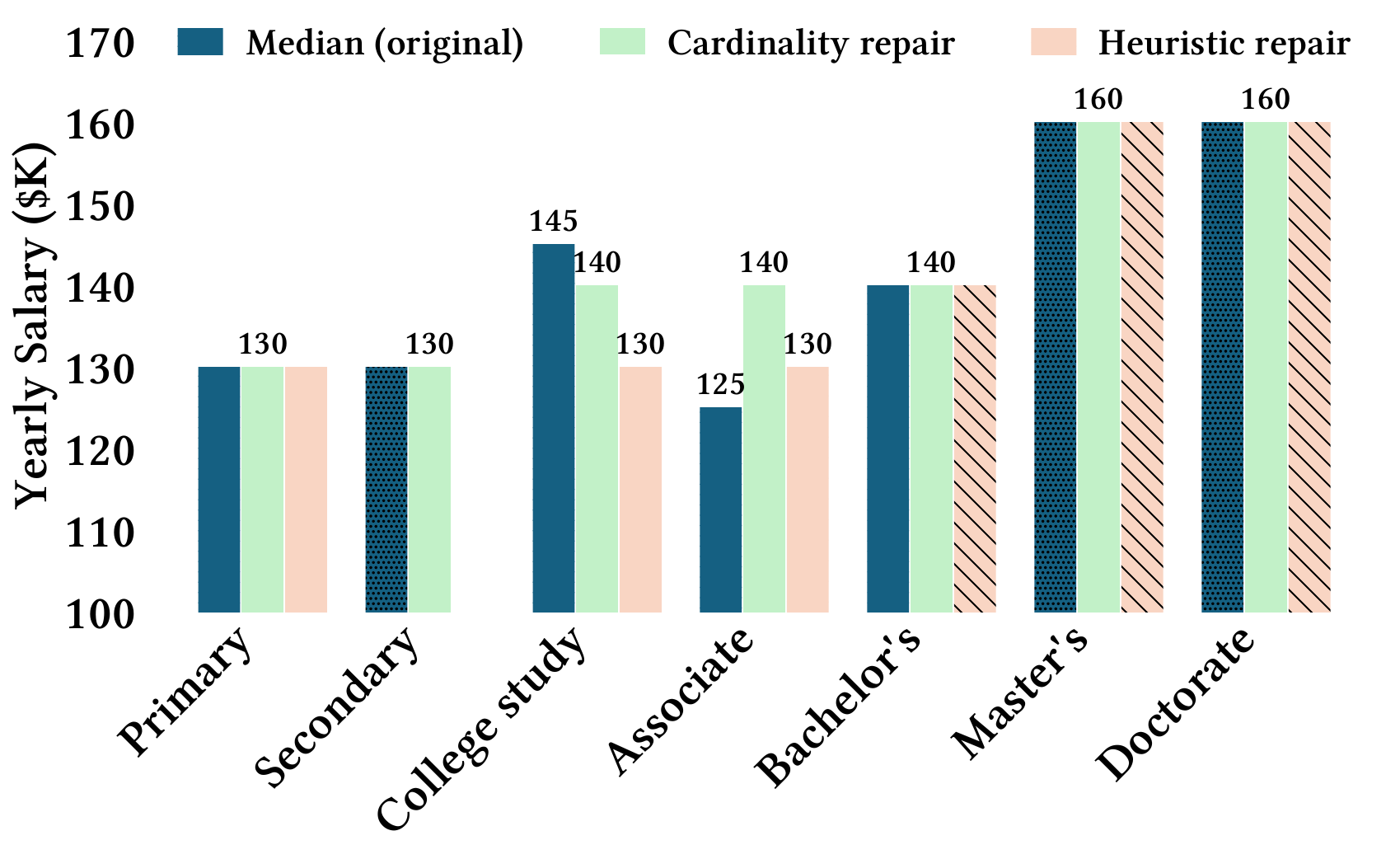}
    \vskip-1em
    \caption{Median salary versus education level in the USA, computed with the Stack Overflow dataset. ``College study'' refers to people who have studied in college or university without a degree. ``Doctoral'' includes professional doctoral degrees (e.g., MD, JD) and other doctoral degrees (e.g., PhD). %\shunit{TODO: tikz or matplotlib}
    \label{fig:SO_example}}
\end{figure}

\section{Introduction}\label{sec:intro}

Real-world datasets frequently violate anticipated monotonic relationships, such as the expectation that the median salary increases with the education level, that newer homes are generally more expensive, and that the prevalence of diabetes rises with age. 
\revb{The detection and analysis of monotonic trends in data is a fundamental problem in statistics~\cite{puri1990recursive,ramsay1998estimating,brunk1955maximum}, with wide-ranging applications in economics, medicine, and the social sciences. For example, monotonicity constraints incorporated into predictive models have been shown to increase predictive accuracy for housing prices, disease prediction, student success, and more~\cite{zhou2016,ovchinnik2019,gonzalez2024}.}
A central task in this domain is to determine whether an underlying relationship---often modeled through noise-aware regression---is monotonic, either increasing or decreasing, with respect to an explanatory variable. 
The inherent challenge is distinguishing monotonic relationships from random fluctuations, especially in high-noise or sparsely sampled data; notable contributions include the work of \citet{ghosal2000testing} and of \citet{hall2000testing}, who developed reliable tests for monotonicity.
In this work, we propose a complementary database-centric view on the problem of analyzing monotonic trends and explaining their violation, namely \emph{measuring the amount of intervention} required to exhibit monotonicity. 

\revb{
\vspace{-1em}
\begin{example}\label{ex:so_example}
The Stack Overflow dataset\footnote{\label{so-url}\url{https://survey.stackoverflow.co/2022} (accessed April 2025)} consists of responses from developers to questions about their jobs, including demographics, education level, role, and salary. We focus on entries from USA (8682 tuples).
One might expect that, as the education level grows, so does the median yearly salary.
This, however, is not the case, as \Cref{fig:SO_example} shows (left bars). 
%\benny{"Shoes" belongs to the other dataset, doesn't it?}
%\footnote{As a preprocessing step, we binned the salary attribute using equi-width bins of size \$1000, and truncated it to the 0.99 percentile.} 
%Precisely, the dataset violates the AOD $\mbox{education}\ne\emedianagg(\mbox{salary})$.
For example, the median salary of people with some college or university studies but without a degree (\$145K)
%\benny{The monetary amounts should be outside of math. \$145K and not $\$145K$.}
is higher than that of people with an Associate's degree (\$125K) or Bachelor's degree (\$140K). Additionally, the median salary of people with primary and secondary school education (\$130K each) is higher than that of people with an Associate degree (\$125K).
\qed
\end{example}
}

%To this aim,
\revb{To analyze the violation of monotonic trends,}
we propose the \emph{aggregate order dependency} (AOD), stating a monotonic trend exhibited by an aggregate query over an ordered grouping attribute. More precisely, an AOD has the form $G \nearrow \alpha(A)$ where $G$ is a grouping attribute (column name) with a linearly ordered domain, $A$ is an aggregated attribute, and $\alpha$ is an aggregate function. The AOD $G \nearrow \alpha(A)$ states that whenever one grouping variable is greater than another, its aggregated value cannot be lower. For example, 
``$\mbox{education}\ne\emedianagg(\mbox{salary})$'' states that the median salary grows with the level of education. Put differently, the AOD $G \nearrow \alpha(A)$ states that the SQL query
\begin{align*}
\textsf{SELECT $G$, $\alpha(A)$ \; FROM $R$ \; GROUP BY $G$}
\end{align*}
should return a table that defines a monotonically non-decreasing function. One can analogously define the AOD $G \searrow \alpha(A)$ that requires a non-increasing trend. 

\revb{The computational challenge we study in this paper is that of finding a \emph{cardinality-based repair} (\crepair for short), that is, 
%\benny{Saying every time "cardinality-based repair" is too mouthful. We should agree on some shorthand, like the conventional "c-repair," and stick to it later on. It is annoying to read the whole term each and every time (7 times in this section).}
the largest subset of the table that satisfies the AOD.
%a smallest set of tuples of which removal from the table leads to the satisfaction of the AOD. 
While \crepairs have been traditionally studied as a \emph{prescriptive} tool in data cleaning~\cite{DBLP:series/synthesis/2012Fan,DBLP:conf/icdt/KolahiL09,DBLP:journals/tods/LivshitsKR20}, our work is better motivated by the \emph{descriptive} aspect of measuring %\ag{Do we still want to use the term `measuring' here?}
%\shunit{I am keeping the term "measuring" here, since we mention the analysis part later in this paragraph.} 
how far the table is from satisfying the AOD and, consequently, the monotonic trend that it suggests.} From another angle, our work can be seen as applying the ``minimal repair'' measure of inconsistency~\cite{10.1145/3725397,DBLP:conf/lpnmr/Bertossi19,DBLP:conf/sigmod/LivshitsKTIKR21}; here, consistency is with respect to the AOD. \revb{Beyond their quantitative semantics, \crepairs can also serve as an explanatory tool of a qualitative nature:} analyzing the removed tuples can shed light on the nature of the trend violation, as we demonstrate through several use cases in \Cref{sec:experiments}.

%\shunit{prescriptive vs descriptive repair - we are focused on the latter. AODs are not standard integrity constraints, but a profiling method for a dataset, a measure of the distance between the given dataset and a dataset where the constraint holds.}

\revb{
\vspace{-1em}
\begin{example}\label{ex:so_example_cont}
%The Stack Overflow dataset\footnote{\label{so-url}\url{https://survey.stackoverflow.co/2022} (accessed April 2025)} consists of responses from developers to questions about their jobs, including demographics, education level, role, and salary. We focus on entries from USA (8682 tuples).
%One might expect that, as the education level grows, so does the median yearly salary.
%This, however, is not the case, as \Cref{fig:SO_example} shoes (left bars). %\footnote{As a preprocessing step, we binned the salary attribute using equi-width bins of size \$1000, and truncated it to the 0.99 percentile.} 
The dataset discussed in \Cref{ex:so_example}
%Precisely, the dataset 
violates the AOD $\mbox{education}\ne\emedianagg(\mbox{salary})$.
%For example, the median salary of people with some college or university studies but without a degree ($\$145K$) is higher than that of people with an Associate's ($\$125K$) or Bachelor's degree ($\$140K$). Additionally, the median salary of people with primary and secondary school education ($\$130K$) is higher than that of people with an Associate degree ($\$125K$).
A \crepair (\Cref{fig:SO_example}, central bars), removing the minimal number of tuples, requires deleting 77 tuples, constituting 0.88\% of the data (17 people with college study and no degree and 60 with Associate degrees). 
In contrast, repairing for a \emph{decreasing} trend 
$\mbox{education}\se\emedianagg(\mbox{salary})$
requires the removal of a minimum of 566 tuples, which constitute 6.52\% of the dataset. This suggests that the dataset is much closer to exhibiting an increasing trend than a decreasing one. 
We later discuss the \emph{heuristic} repair that is referred to in the right bars of the figure.
\qed
\end{example}
}

\revb{Computing a \crepair is challenging, both theoretically and pragmatically, since there can be exponentially many ways of restoring the AOD through tuple deletions (as is the case for functional dependencies~\cite{DBLP:journals/tods/LivshitsKR20,DBLP:conf/icdt/KolahiL09}). Hence, we begin with a complexity analysis of the problem, showing that while the general case is NP-hard, several practically relevant aggregates admit polynomial-time solutions. %\ag{Which says?}.
We present a general algorithmic template that %\revb{%R2.M5
uses dynamic programming to reduce %}
the problem to ones that involve a single group, which are arguably simpler. We show the instantiation of the template on a list of common aggregate functions $\alpha$, including max, min, count, count distinct, median, sum, and average. Hence, we establish a polynomial-time bound for each of these aggregates. However, in the case of sum and average, %\revb{%R2.M4
where the number of possible aggregation values can be exponential in the number of tuples, %}
the resulting algorithm is only weakly polynomial, that is, it runs in polynomial time if we assume that the input represents the numbers in a unary encoding; we show that this weakening is necessary, as the problem is NP-hard otherwise.}

\revb{%R2.M3
Although the algorithmic template offers a polynomial-time guaranty, 
in practice, this time can be very long. For example, over an hour for 10K tuples with \medianagg, or for 1K tuples with \avgagg.
%it is not practically scalable
%\ag{How do we know? Some detail about the results of our experiments}: 
%\benny{I agree, I'm also not happy about this statement. Does it imply that our optimized one is "practically scalable"? Many would say that it is still not. Say something more precise, like ``Although the algorithmic template offers a polynomial-time guarantee, in practice, this time can be very long. For example,..."}
%\shunit{Please review the next 2 sentences - is it clearer now?}
The algorithm is generic and applies to various aggregation functions under mild assumptions. While useful, this also means we are not using the specific properties of the aggregation functions, which could accelerate the computation.} 
%To be agnostic to the aggregation function, the algorithm necessarily includes some inefficiencies \ag{Again, what is `some inefficiencies'? Add some details}.}
%\benny{While I understand the statement of "to be agnostic to..." from knowing this research, I believe that very few others will.}
We address this limitation through two complementary approaches. First, 
we develop optimization techniques tailored to the specific aggregation functions, significantly reducing execution costs. \revb{Second, we introduce a heuristic method that computes a repair satisfying the AOD %\revb{%r2.M5
by repeatedly removing the tuple with the most impact on the violation, %}
albeit potentially with more deletions than the \crepair.} This heuristic serves two key purposes: %First, 
%\benny{Maybe we can avoid repeating the pattern of "first, second" in the same paragraph?"}
% it provides 
(1) providing a practical alternative when computational resources are limited, and (2) yielding an upper bound on the number of deletions required. We demonstrate how this bound can be leveraged to substantially prune the search space of the dynamic program of the exact computation.

\begin{example}
As we demonstrate in \Cref{sec:experiments}, the heuristic algorithm performs well for several aggregate functions. However, in certain scenarios, it may yield substantially inferior repairs. For the use case described in \Cref{ex:so_example}, the heuristic algorithm computed a repair for the non-decreasing trend $\mbox{education}\ne\emedianagg(\mbox{salary})$ in 0.52 seconds, removing 494 tuples---5.7\% of the data. While this is 30$\times$ faster than the exact algorithm (15.9 seconds), it removed 6.4 times more tuples than necessary. Notably, it eliminated the entire secondary-school group.

As for the opposite AOD, namely 
$\mbox{education}\se\emedianagg(\mbox{salary})$, the heuristic algorithm removed 8187 tuples (94.3\% of the data) in 16.7 seconds. This number provides an upper bound on the number of tuples needed to be removed, but it does not give any information about the \emph{minimum} amount of deletion needed to fulfill the AOD;
%\ester{did we expect it to be used as a lower bound?}, 
hence, it cannot support a claim as the one closing \Cref{ex:so_example_cont}. \qed
\end{example}

We then describe an extensive empirical study on our implementation of the algorithms. The datasets include the German credit~\cite{german_credit}, Stack Overflow,\footref{so-url} H\&M~\cite{relbench}, Zillow,\footnote{\url{https://www.zillow.com/research/data/}} and Diabetes,\footnote{\url{https://www.kaggle.com/datasets/iammustafatz/diabetes-prediction-dataset}} as well as synthetic data that we generated. %\ag{Some have citations, some have links, and some have none?}. 
The empirical study demonstrates the effectiveness of the algorithms and optimizations and shows the quality of the heuristic solution, which is often highly competitive, even though it takes a small fraction of the execution time of finding a \crepair.

% \ag{Add a paragraph title here: Our contributions}
%\paragraph{Our contributions}
In summary, our contributions are as follows:
\begin{enumerate*}%[leftmargin=*]
\item We introduce the notion of an AOD to express a monotonic trend in the result of an aggregate query over a table;

\item We provide a complexity analysis of \crepairs for AODs, and present a general algorithmic framework that reduces the global repair problem to local problems on individual groups; 
%For several common aggregate functions (e.g., min, max, count, median, sum, average), we establish polynomial-time algorithms, with matching complexity bounds distinguishing between strongly and weakly polynomial-time cases.

\item We develop multiple optimization techniques to improve the practical performance of the algorithms; 

\item We propose a heuristic algorithm that efficiently computes approximate repairs and supports the exact algorithm by providing an upper bound for pruning the search space;

\item We conduct an experimental study on real and synthetic datasets.
\end{enumerate*}
%
%Together, our results position AODs and their repairs as a useful framework for analyzing and interpreting data in scenarios where monotonic trends are expected yet violated.

% \section{Related Work}
% The AOD introduced here is reminiscent of, yet different from, the \emph{order dependency} (OD) in the unary case (where each attribute sequence includes a single attribute). An OD $X\rightarrow_{\succeq} Y$ states that an attribute $X$ functionally determines the attribute $Y$ and, moreover, the resulting function is monotonic with respect to the order $\succeq$ over the domain~\cite{langer2016order_deps, consonni2019order_deps,DBLP:journals/tcs/GinsburgH83}. In particular, an OD does not involve any aggregation, and an AOD does not entail any functional dependency.

%\shunit{at the end of the intro and maybe in other places, explain that avg remains an open problem, and why this is okay (because it's really a hard problem)}

%\subsubsection*{Organization}
The remainder of the paper is organized as follows.  After an overview of related work in \Cref{sec:related}, we present 
in \Cref{sec:formal}
the formal framework and define AODs. \Cref{sec:dp} provides a complexity analysis of the \crepair problem. In \Cref{sec:optimizations}, we describe our implementation and algorithmic optimizations, followed by heuristic repair strategies in \Cref{sec:greedy_tuple_del}. \Cref{sec:experiments} reports on our empirical evaluation, and we conclude in \Cref{sec:conclusions}.
%\ag{Remove some parentheses?}

%% file: 06-related_work.tex
\section{Related Work}
\label{sec:related}

\paragraph*{Database constraints}
% The AOD introduced here is reminiscent of, yet different from, the \emph{order dependency} (OD) in the unary case (where each attribute sequence includes a single attribute). An OD $X\rightarrow_{\succeq} Y$ states that an attribute $X$ functionally determines the attribute $Y$ and, moreover, the resulting function is monotonic with respect to the order $\succeq$ over the domain~\cite{langer2016order_deps, consonni2019order_deps,DBLP:journals/tcs/GinsburgH83,OD_repair2018}. In particular, an OD does not involve any aggregation, and an AOD does not entail any functional dependency. 

A large body of work in the data management community has studied enforcing constraints on relational data via data repair, traditionally focusing on integrity constraints~\cite{DBLP:series/synthesis/2011Bertossi,
DBLP:series/synthesis/2012Fan}. This includes functional dependencies~\cite{DBLP:journals/tods/LivshitsKR20,DBLP:conf/icdt/KolahiL09,DBLP:journals/vldb/MiaoZLWC23,DBLP:conf/icdt/GiladIK23,DBLP:conf/sigmod/BohannonFFR05, DBLP:conf/icdt/CarmeliGKLT21}, conditional functional dependencies~\cite{bohannon2006conditional, geerts2013llunatic}, multi-valued dependencies~\cite{salimi2019interventional}, denial constraints~\cite{chomicki2005minimal, chu2013holistic,holoclean}, and inclusion dependencies~\cite{DBLP:conf/sigmod/BohannonFFR05,DBLP:conf/foiks/MahmoodVBN24,DBLP:journals/iandc/ChomickiM05,DBLP:conf/icdt/KaminskyKLNW25}. The most common repair operations include tuple deletion~\cite{DBLP:journals/tods/LivshitsKR20,DBLP:conf/icdt/CarmeliGKLT21,DBLP:journals/vldb/MiaoZLWC23,chomicki2005minimal,DBLP:conf/foiks/MahmoodVBN24,DBLP:journals/iandc/ChomickiM05,DBLP:conf/sigmod/GiladDR20}, tuple insertion~\cite{salimi2019interventional}, and value updates~\cite{DBLP:conf/icdt/KolahiL09,DBLP:conf/icdt/KaminskyKLNW25,DBLP:conf/icdt/GiladIK23,DBLP:conf/sigmod/BohannonFFR05,geerts2013llunatic,holoclean,chu2013holistic}. 
%Our work proposes a new class of constraints over aggregate views rather than the input database and studies the complexity and algorithmic strategies for enforcing them through cardinality repair.
The AOD entails a fundamentally different repair problem, since it involves aggregation. Repairs of aggregation-centric constraints have been studied by
\citet{flesca2007preferred, flesca2010querying, flesca2011repairing}, but they have not considered monotonicity---a central aspect of our framework.

Closer to our work is the line of research on dependencies that involve ordering, specifically the \emph{order dependencies} (ODs)~\cite{dong1982,ginsburg1986,langer2016order_deps, consonni2019order_deps,DBLP:journals/tcs/GinsburgH83,OD_repair2018,jin2020}, also known as \emph{ordered functional dependencies}~\cite{ng1999,ng2001}. An OD of the form $X \rightarrow_{\succeq} Y$ asserts that $X$ functionally determines $Y$, and that the relationship is monotonic with respect to a specified order $\succeq$ over the domain. 
Of similar flavor are the
\emph{trend dependencies} (TD)~\cite{wijsen2001trends}, capturing changes over time (e.g., ``salaries do not decrease in time'').
These formalisms strengthen functional dependencies by capturing monotonic dependencies between attributes. In contrast, our AODs consider an aggregation level of the data and do not require or entail any underlying functional dependency.
In addition, past work in this category focused on logical analysis  (satisfiability and implication) and constraint mining, differently from the focus of this work---quantifying violation via repairs.

%Another related line of work focuses on aggregate constraints~\cite{flesca2007preferred, flesca2010querying, flesca2011repairing}, which enforce conditions on the result of an aggregate function over a group of tuples. These works typically assume fixed bounds and apply value updates to enforce the constraint. 

%\shunit{Other directions/citations reached via \cite{wijsen2001trends}: }
%\red{
%\begin{itemize}
    %\item Earlier works on order dependencies: \cite{dong1982,ginsburg1986}. And a later work on order dependency discovery \cite{jin2020}.

    %\item ordered functional dependencies - \cite{ng1999,ng2001} \benny{Related how?}

    %\brit{\cite{ng1999} - order FDs - no aggregation. Assuming a semantic order on one attribute imposes an order on another attribute (less education -> lower wedge). \cite{ng2001}: extend the relational data model to incorporate partial orderings into data domains, which we call the ordered relational model.}

    %\item sequential dependencies \cite{golab2009} - dependencies of the form “when tuples are sorted on X, the distance between the Y-values of any two consecutive tuples are within interval g.” \benny{Related how?}

    %\item Differential dependencies~\cite{song2011} - given distance functions for attributes X and Y, a differential dependency can require a connection between $f_x(t1, t2)$ and $f_y(t1,t2)$. E.g., the distance in x grows as the distance in y grows (or is reduced).
%\end{itemize}
%}

\paragraph*{Statistical perspectives on monotonicity}
Monotonic relationships between variables, such as the expectation that an outcome variable consistently increases or decreases with an explanatory factor, have been extensively studied in statistics and econometrics.
%~\cite{ghosal2000testing,patton2010monotonicity,lee2009testing}. 
Much of that work focuses on detecting violations of monotonicity using nonparametric or regression-based methods, often under linearity assumptions~\cite{patton2010monotonicity,hall2000testing,ghosal2000testing}. These approaches are common in applications such as environmental trend analysis~\cite{hussian2005monotonic}, economics~\cite{patton2010monotonicity,lee2009testing}, and financial forecasting~\cite{patton2010monotonicity}. Other research takes a more theoretical view, using property testing to assess monotonicity over general domains, including partial orders~\cite{bhattacharyya2011testing}. 
%Prior work has focused on testing whether a monotonic trend holds in the data. In contrast, we treat such trends as prior expectations and seek to explain violations by identifying a minimal set of tuples whose removal restores monotonicity. 
While this field focuses on identifying monotonic models with noise, our work offers a complementary perspective: we quantify the deviation from the trend by the amount of required database intervention that creates the expected behavior. Moreover, our framework accommodates general aggregation functions and is not bound to any specific one.
%, and provides diagnostic insights by pinpointing the specific data records that contribute to the deviation.

\paragraph*{Cherry-picking} Another relevant line of research focuses on assessing the robustness of claims based on the results of aggregate queries~\cite{wu2014toward,wu2017computational,asudeh2020detecting,lin2021detecting,jo2019aggchecker,DBLP:journals/pvldb/AgmonGYZK24}. \citet{wu2014toward,wu2017computational} defined parameterized queries and analyzed how perturbations in input parameters affect the resulting claims, using measures of relevance and naturalness. Closer to our work,~\citet{asudeh2020detecting,asudeh2021perturbation} investigate the robustness of trendlines to selective endpoints or item selection, aiming to expose cherry-picking. While these approaches highlight the vulnerability of trend-based claims to manipulation, our work focuses on quantifying and repairing violations of expected trends in the data. The challenge of identifying misleading trends is, again, different from the focus of this work on the repair challenge.

%% file: 02-formal.tex
\section{Formal Framework} \label{sec:formal}

We begin by presenting the formal framework of this work, starting with preliminary concepts on databases.

\subsubsection*{Database concepts}%\label{sec:db_concepts}
%\ag{We also need a preliminaries section with background on databases, aggregate queries, complexity classes, patterns...}

%\subsection{Complexity Classes}
A relation schema $\scs$ is a sequence $(A_1,\dots,A_k)$ of \emph{
attributes}. We denote by $\atts(S)$ the set $\set{A_1,\dots,A_k}$. 
%Each attribute $G\in\atts(\scs)$ is associated with a \emph{domain} $\dom(A)$ of values. Moreover, we assume a total order $\preceq_A$ over each domain $\dom(A)$.
We assume, without loss of generality, that all attributes are numeric.
In particular, a relation $r$ over $\scs=(A_1,\dots,A_k)$ is a finite set of tuples 
%$\dom(A_1)\times\dots\times\dom(A_k)$. 
in $\mathbb{R}^k$.
If $t\in r$ is a tuple and $G\in\atts(\scs)$, then we denote by $t[A]$ the value of $t$ for the attribute $A$. We denote by
$r[A]$ the projection of $r$ to $A$, that is, the set $\set{t[A]\mid t\in r}$. We also write $r\dbr{A}$ to denote the bag-semantics projection of $r$ to $A$, that is, the bag of values that occur under $A$ without removing duplicates.

In this work, an \emph{aggregate function} $\alpha$ maps a bag $V$ of numbers to a single value $\alpha(V)$. We focus on the common aggregate functions \avgagg, \sumagg, \countagg, \countdagg, \minagg, \maxagg, and \medianagg. 

Let $G$ and $A$ be two attributes of $\scs$, let $r$ be a relation, let $\alpha$ be an aggregate function, and let $g\in\dom(G)$ be a constant. We denote by $\alpha_r(A\mid G=g)$ 
the value $\alpha((\sigma_{G=g}r)\dbr{A})$, where
$\sigma_{G=g}r$ is the relation that consists of all tuples $t\in r$ with $t[G]=g$.

\subsubsection*{Aggregate Order Dependencies}
%\brit{constraint? those are not dependencies I think the name is confusing}
%\benny{I think it should be "aggregate order dependencies" and not "aggregated...". By "aggregated" it sounds like the order itself is aggregated, while our meaning is that the dependency is expressed using the aggregate function.}
%\brit{can we say that w.l.o.g we focus on increasing trends, but all results and algorithms can also be adjusted to support decreasing trends?}
Let $\scs$ be a relation schema. 
An \emph{aggregate order dependency} (\aod) is an expression of the form 
$G\nearrow \alpha(A)$ 
where $G$ and $A$ are attributes in $\atts(\scs)$.  We refer to $G$ as the \emph{grouping} attribute and to $A$ as the \emph{aggregate} attribute.
%\ag{So maybe denote the attributes by $G$ and $A$ is more intuitive?}. 
%
A relation $r$ \emph{satisfies} the \aod $G\nearrow \alpha(A)$, denoted $r\models G\nearrow \alpha(A)$, if the following holds: for all values $g$ and $g'$ in $r[G]$, 
% $$\forall a, a'\in r[A]:$$
if $g \leq g'$ %$a\preceq a'$
then $\alpha_{r}(A\mid G=g) \leq \alpha_{r}(A\mid G=g')$.
%\ag{We assume that there is always a total/partial order over the values in $G$?}\shunit{we said all attributes are numeric in \Cref{sec:db_concepts}}.
(We focus on \emph{increasing} trends, but all our results
%, algorithms, and implementations 
are applicable to \emph{decreasing} trends as well.)

\begin{example}\label{example:aod}
Consider the relation $r$ on the left of \Cref{tab:DB_example}. The tables on the right give the aggregate values 
$\alpha_r(\mbox{income}\mid \mbox{edu}=a)$ for $\alpha=\esumagg$ (top) and $\alpha=\eavgagg$ (bottom). We can see that $r$ satisfies the AOD
$\mbox{edu}\nearrow\esumagg(\mbox{income})$,
since the sum column is monotonically increasing. On the other hand, $r$ violates the AOD  $\mbox{edu}\nearrow\eavgagg(\mbox{income})$ since the average salary for education level 2 is higher than the average salary for education level 3.
\qed
\end{example}

\input{00-motivating_example}

\subsubsection*{\revb{Cardinality-based repair}}\label{sec:probdef}
\revb{The main problem we study is that of finding a \emph{cardinality-based repair} (\crepair for short) for an \aod. The formal definition is as follows.} %\shunit{add def for monotonic subset, use it to define a cardinality repair, refer to it from sec 5}

%\ag{This is the main definition in the paper. I suggest we place it in a definition environment:} Formally, let $\scs$ be a relation schema, let $G\nearrow \alpha(A)$ be an \aod over $\scs$, and let $r$ be a relation. %that violates $G\nearrow \alpha(A)$ %\ester{I think that we can remove "that violates"}. 
\begin{definition}
%A \emph{cardinality repair} of $r$ is a maximum-size subset $r'$ of $r$ such that $r'\models G\nearrow \alpha(A)$. 
A \emph{monotonic subset} of $r$ is a subset $r'$ of $r$ such that $r'\models G\nearrow \alpha(A)$.
\revb{A \emph{\crepair} of $r$ is a monotonic subset of maximum size.}
Hence, $r'$ is obtained from $r$ by deleting the minimum possible number of tuples needed to satisfy $G\nearrow \alpha(A)$. 
\end{definition}
We illustrate the definition with an example.
\begin{example}
Continuing \Cref{example:aod}, let us consider the AOD $\mbox{edu}\nearrow\eavgagg(\mbox{income})$. We can repair $r$ by removing the tuples of 
Larrie, Nathan, and Marie (highlighted in green), which would result in 
$\alpha_{r'}(\mbox{income}\mid \mbox{edu}=3)\;=\;4.25$
and, therefore, satisfy the AOD. 
Nevertheless, we can also repair $r$ by removing the tuples of Daniel and Emily (highlighted in pink), which would result in 
$\alpha_{r'}(\mbox{income}\mid \mbox{edu}=2)\;=\;3$.
In this example, removing any single tuple will not satisfy the AOD, so the repair requires removing at least two tuples. 
%\shunit{this next part isn't true anymore} Since we need to delete at least one tuple (as the AOD is violated), 
Hence, we conclude that removing the tuples of Daniel and Emily results in a \revb{\crepair} of size 12.%\ag{of size $2$.} 
%\ag{Color these rows in the figure so it's clearer}.
\qed 
%\brit{perhaps: we conclude that removing the tuple of Issac (or Joe) results in a cardinality repair}
\end{example}

%% file: 00-motivating_example.tex
\begin{figure}[t]
    \centering \small
{\begin{minipage}{1.5in}
\renewcommand{\arraystretch}{0.9}
    \begin{tabular}{lcc}
        \hline
        \rowcolor{lightgray}person & edu & income\\
        \hline
Ashley & 1 & 1 \\
Brandon & 1 & 2 \\
%Chloe & 1 & 8 \\
\hline
%\rowcolor{pink}
Chloe & 2 & 2 \\
\rowcolor{pink}Daniel & 2 & 5 \\
%\rowcolor{pink}
\rowcolor{pink}Emily & 2 & 6 \\
Faith & 2 & 5 \\
Gavin & 2 & 2 \\
%Hanna & 2 & 2 \\
\hline
Hanna & 3 & 8 \\
Isaac & 3 & 4 \\
Jerry & 3 & 3 \\
Katie & 3 & 2 \\
\rowcolor{green!15}Larry & 3 & 2 \\
\rowcolor{green!15}Marie & 3 & 1 \\
\rowcolor{green!15}Nathan & 3 & 1 \\
        \bottomrule
    \end{tabular}
\end{minipage}}
\quad
{
\begin{minipage}{1.2in}
%\vskip-10ex
  \begin{tabular}{cc}
        \hline
        \rowcolor{lightgray}edu & $\esumagg(\mbox{income})$ \\
        \hline
1 & 3 \\
2 & 20 \\
3 & 21 \\
        \bottomrule
    \end{tabular}
\vskip6.3ex
  \begin{tabular}{cc}
        \hline
        \rowcolor{lightgray}edu & $\eavgagg(\mbox{income})$ \\
        \hline
1 & 1.5  \\
2 & 4 \\
3 & 3 \\
        \bottomrule
    \end{tabular}   
\end{minipage}}
    \caption{Example relation $r$ (left) and statistics grouped by the education level (edu).}
    \label{tab:DB_example}
\end{figure}

%% file: 03-tuple_del.tex
\section{\revb{The Complexity of \crepairs} %Cardinality Repairs
}\label{sec:dp}
\revb{In this section, we investigate the computational complexity of finding a \crepair for an AOD. The main result is the following.

\begin{theorem}\label{thm:ptime}
    Let $G\nearrow \alpha(A)$ be an AOD. 
    \begin{enumerate}
        \item For $\alpha \in \set{\emaxagg, \eminagg, \ecountagg, \ecountdagg, \emedianagg}$, a \crepair can be found in polynomial time. 
        \item For $\alpha\in\set{\esumagg,\eavgagg}$, a \crepair can be found in pseudo-polynomial time.
    \end{enumerate}    
\end{theorem}
}
Recall that \emph{pseudo-polynomial time} means that the running time is polynomial if we assume that the numbers (here, in the column $A$) are integers given in a unary representation. In \Cref{sec:tuple_del_hardness}, we will show that this weakening is necessary, since the problem is NP-hard if we assume a binary representation.

\revb{To prove the theorem, we present algorithms for computing \crepairs for different aggregate functions $\alpha$}. We begin with a general procedure that applies to all aggregate functions (and requires some sub-procedures to be implemented) and then proceed to specific ones. 
Throughout this section, we assume the \aod is $G\nearrow \alpha(A)$ and refer to the input relation as $r$.

\subsection{General Procedure}\label{sec:general_dp}

\revb{We first show a dynamic programming algorithm for finding a \crepair (described in \Cref{alg:dp})}.
The algorithm requires specific procedures that depend on the  
aggregation function $\alpha$. In the next sections, we discuss the realization of these procedures for specific aggregation functions. The procedures are as follows.

\begin{enumerate}[leftmargin=*]
\item The first procedure takes a relation $r$ as input, and produces a set $V_\alpha(r)$ that contains all possible values $\alpha_{r'}(A)$ 
that any nonempty subset $r'\subseteq r$ 
can take.
%\ag{Explain why $|V_\alpha(r)|$ is polynomial or refer to 3.2}. 
For example, if $\alpha$ is \maxagg or \minagg, then $V_\alpha(r)$ can be $r[A]$. 

\item The second procedure 
takes as input a relation $r$ with a single group (hence, it ignores the grouping attribute $G$) and a value $x\in V_\alpha(r)$, and it computes an optimal (maximum-size) relation 
$r'\subseteq r$ such that $\alpha(r'\dbr{A})=x$, that is, the aggregate value is exactly $x$. 
We denote $r'$ by $O_\alpha(r,x)$. 
If no such $r'$ exists, then we define $O_\alpha(r,x)=\emptyset$. %\ag{Forward ref to the section describing the procedure}
(Later on, we will refer to this problem as \emph{aggregation packing}.)
\end{enumerate}

% \ag{Forward ref to the section describing the procedure}

\begin{example}
Consider again the relation $r$ of  \Cref{tab:DB_example}, and let $\alpha$ be \avgagg. Then
$V_\alpha(r)$ should include (at least) every 
mean of every possible bag of values from the income column. If $r$ consists of the tuples of Daniel, Emily, and Faith, and $\alpha$ is \avgagg, then $O_\alpha(r,5)$ consists of the tuples of Daniel and Faith. \qed
%\shunit{I don't understand this example - what is $x$ here? Also, does it need to change now that the toy example has changed?} \qed
\end{example}

%\brit{A small example here would help to understand the notations better}

%In our algorithm, both procedures are applied to each group $r_i$ separately.
%We later show algorithms to compute $V_\alpha(r)$ and $O_\alpha(r,x)$ for several aggregation functions $\alpha$. 
We describe how to compute $V_{\alpha}(r)$ and $O_\alpha(r,x)$ efficiently for different aggregation functions in \Cref{sec:aggpack_instances}.
\revb{Before we do so, we show how $V_\alpha(r)$ and $O_\alpha(r,x)$ can be used for computing a \crepair}.

%\red{For some aggregation functions, these are polynomial time algorithms; for others (\sumagg and \avgagg) the time complexity of the algorithm depends on the range of values in the input relation.}
%In the algorithm description that follows, we use the following notation.
%Let $r[A]=\set{a_1,\dots,a_m}$ where $a_i<a_j$ whenever $i<j$.
%We denote by $r_i$ the relation $\set{t\in r \mid t.A = a_i}$ \red{should we move this definition to the previous section? we use it in more than one place.}, which we refer to as a \emph{group}. For $i=1,\dots,m$, we denote by $r_{\leq i}$ the relation $r_1\cup\dots\cup r_i$.

For convenience, we 
use the following definitions.
Let $r[G]=\set{g_1,\dots,g_m}$ where $g_i<g_j$ whenever $i<j$.
We denote by $r_i$ the relation $\set{t\in r \mid t[G] = g_i}$, which we refer to as a \emph{group}. For $i\in \set{1,\dots,m}$, we denote by $r_{\leq i}$ the relation $r_1\cup\dots\cup r_i$. 
As an example, referring to the relation $r$ of \Cref{example:aod}, the relation $r_2$ consists of the tuples of Chloe, Daniel, Emily, Faith, and Gavin, %\ag{We should consider adding tuple identifiers. This would make the examples much more accurate.},
and the relation $r_{\leq 2}$ includes, in addition to the five, the tuples of Ashley and Brandon.

\begin{algorithm}[t]
\small
\DontPrintSemicolon
\SetKwInOut{Input}{Input}\SetKwInOut{Output}{Output}
\LinesNumbered
\newcommand\mycommfont[1]{\footnotesize\ttfamily\textcolor{blue}{#1}}
\SetCommentSty{mycommfont}
% \SetKwFunction{FindSol}{FindSol}
\Input{Relation $r$, \aod $G \nearrow \alpha(A)$} 
%\ag{Should there be $V_\alpha(r)$ here as well?}\shunit{No, it's computed from $\alpha,r$ using procedure (1)}}
\Output{Maximum-size %\ag{`Maximal' or `maximum sized'?} 
subset $r'\subseteq r$ such that $r'\models G \nearrow \alpha(A)$.}

%$\mathit{Combs}\gets $ all combinations~from $\splitset$ of size $\leq m$ \label{line:combs} \\
%$\mathit{CombsSorted}\gets \SortCombs(\mathit{Combs},k)$ \tcp*[l]{\Cref{sec:methods:prioritization}}\label{line:sort}

%\For {$i\in \set{1,\dots,m}$} {
% $\forall x:   F_i[x] \gets \emptyset$ \\
%    \For {$x \in V_\alpha(r_i)$} {
%        $F_i[x] \gets \mathsf{AggPack}\langle{\alpha}\rangle(r_i,x)$
%    }
%}
\For {$x \in V_\alpha(r)$} {
    $H_0[x]\gets \emptyset$
}
\For {$i\in \set{1,\dots,m}$} {
    %\For {$x \in V_\alpha(r_i)$} {
        %$F_i[x] \gets \mathtt{AggPack}\langle{\alpha}\rangle(r_i,x)$
    %}
    \For {$x \in V_\alpha(r_i)$} {
    \If{$O_\alpha(r_i,x)=\emptyset$}{
        $H_i[x] \gets H_{i-1}[x]$\\
    }
    %\tcp*[l]{Feasible values of $\alpha_{r_i}(A)$}
    \Else{
        $H_i[x] \gets O_\alpha(r_i,x)\cup
        H_{i-1}\left[
        \argmax_{y\leq x} |H_{i-1}[y]|
        \right].$
        %H_{i-1}^{\leq}[x] \cup O_\alpha(r_i,x)$
        }
    }
}
\Return $H_m[\max{(V_\alpha(r))}]$
%\benny{%For the "return" word, please use the same font as "if" and "for". Also,  What is $\max{V_\alpha(r)}$ max of what over what?}\shunit{It's the maximal value in the set $V_\alpha(r)$, is there a better way to denote it?}

%\label{line:for_loop}
%$P\gets\FindPredicates(\mathit{comb},Q,\Aagg,\Agb,\claim)$ \tcp*[l]{\Cref{sec:methods:find_preds}}\label{line:find_preds}
%$M \gets\ComputeNaturalnessMeasures(P)$ \label{line:compute_nat}
%\tcp*[l]{\Cref{sec:methods:optimizations}}\label{line:compute_nat}
%$\mathsf{Print}(M)$\label{line:print}
% \caption{Dynamic programming algorithm for tuple deletion}
\caption{\revb{Dynamic programming algorithm for \crepair: \dpalg. %cardinality-based repair.
}}
\label{alg:dp}
\end{algorithm}

\Cref{alg:dp} %\ag{`The algorithm, depicted as \Cref{alg:dp}' -> just `\Cref{alg:dp}'},
applies dynamic programming by increasing $i=1,\dots,m$. For each $i$ and $x\in V_\alpha(r_i)$, the algorithm computes $H_i[x]$, which is a largest subset $r'\subseteq r_{\leq i}$ such that
$r'$ satisfies the AOD and, in addition, the value of the rightmost bar is $x$; that is, if 
$g$ is the maximal value in $r'[G]$, then 
$\alpha_{r'}(A\mid G=g) =x$.
For each $i$ and $x\in V_\alpha(r_i)$, the algorithm considers (in lines~4--8) two cases. 
\begin{enumerate}
    \item 
In the first case, $O_\alpha(r_i,x)$ is empty; that is, no subset of $r_i$ has the $\alpha$ value $x$. In this case, $H_i[x]$ is simply $H_{i-1}[x]$ (ignoring the group $r_i$). 
\item Otherwise, if $O_\alpha(r_i,x)$ is nonempty, then $H_i[x]$ is the union of $O_\alpha(r_i,x)$ and the best $H_{i-1}[y]$ over all $y\leq x$.
\end{enumerate}

\subsection{Instantiations} 
Next, we present instantiations of \Cref{alg:dp} to specific aggregate functions $\alpha$ by showing how to compute $V_\alpha(r)$ and $O_\alpha(r,x)$ in polynomial time.
\label{sec:aggpack_instances}
%\ag{Refer back to Thm 3.1 and say that these instantiations will yield a PTime algo.}

\subsubsection{Polynomial-Time Instantiations}\label{sec:aggpack_instances_poly}
We begin with polynomial-time instantiations that give rise to the first part of \Cref{thm:ptime}.
\paragraph{Max and min}
%\ag{In all of the following, replace $V_\alpha$ with $V_{\emaxagg}$ or whatever the current aggregate function is. The same for $O_\alpha$.}\shunit{Benny to rewrite where appropriate}
When $\alpha$ is \maxagg,
the set $V_\alpha(r)$ of possible values over $r$ is simply the set of values in the column $A$.
%\brit{You refer to it as a column in some places and as an attribute in others — it would be better to use consistent terminology throughout.}. 
For $x\in V_\alpha(r)$, the set $O_\alpha(r,x)$ is simply the set $\set{t \in r \mid t[A] \leq x}$ of all tuples where the $B$ attribute is at most $x$. The case of \minagg is similar, except that  $O_\alpha(r,x)$ is the set $\set{t \in r \mid t[A] \geq x}$.

\paragraph{Count and count-distinct}
For \countagg, the set $V_\alpha(r)$ is 
$\set{0,\dots, |r|}$, and $O_\alpha(r,x)$ can be any subset of size $x$. For \countdagg, the set $V_\alpha(r)$ is $\set{0,\dots,|r[A]|}$; for $x\in V_\alpha(r)$, the set $O_\alpha(r,x)$ contains the tuples with the $x$ most frequent $A$ values in $r$. More precisely, for each $a\in r[A]$, denote $n_a=|\set{t\in r\mid t[A]=a}$. Let $a_1,\dots,a_m$ be the values of $r[A]$ sorted by decreasing $n_{a_i}$. Then
$$O_{\ecountdagg}(r,x)=\set{t\in r\mid t[A]\in\set{a_1,\dots,a_x}}\,.$$ 

%\benny{Got here }

\paragraph{Median}\label{sec:aggpack_instances_median}
%\ag{Why is Median separated from Max, min, and count?}
The set $V_\alpha(r)$ of possible median values over $r$ is 
\[\set{t[A]\mid t\in r}\cup \set{(t_1[A] + t_2[A])/2 \mid t_1, t_2 \in r}.\] For $x\in V_\alpha(r)$, we denote by $n^>_x$, $n^=_x$, and $n^<_x$ the number of tuples $t\in r$ with $t[A]>x$, $t[A]=x$, and $t[A]<x$, respectively. Let $t_1,\dots,t_n$ be the tuples of $r$ sorted by increasing $t_i[A]$. Then, the set $O_\alpha(r,x)$ is the largest of the two candidate sets $r_{\text{include}}$ and $r_{\text{exclude}}$. The set $r_{\text{include}}$ is the largest subset of $r$ where $x$ is the median by being the single middle element. Hence, the number of values higher than and lower than $x$ should be equal (and the size of the set is odd):
%\benny{Is the lack of an ending period intentional?}
%\jonny{should this be odd? or does this reference the set of tuples which are not x?}.
\[r_{\text{include}}=\begin{cases}
        \emptyset & \text{if }n^=_x=0,\\
        r & \text{if }n^=_x\neq 0\text{ and }n^>_x=n^<_x,\\
        \set{t_1,\dots,t_{2n^<_x+n^=_x}} & \text{if }n^=_x\neq 0\text{ and }n^>_x>n^<_x,\\
        \set{t_{n^<_x-n^>_x+1},\dots,t_n} & \text{otherwise}.
    \end{cases}\]

The set $r_{\text{exclude}}$ is the largest subset where $x$ is the median because it is the average of two middle elements, which is only possible for an even-sized set. Let $S=\set{(i,j)\mid t_i[A]+t_j[A]=2x,i<j}$ and let $(i^*,j^*)=\argmax_{(i,j)\in S}\{\min\{i-1,n-j\}\}$. Then,

\[ r_{\text{exclude}}=
\begin{cases}
    \emptyset & \text{if } S=\emptyset;\\[0.4em]
    \begin{aligned}
        &\{t_1,\dots,t_{i^*-1}\} \cup \{t_{i^*},t_{j^*}\}\\
        &\cup \{t_{j^*+1},\dots,t_{j^*+i^*-1}\}
    \end{aligned} & \text{if } i^*-1\leq n-j^*;\\[1em]
    \begin{aligned}
        &\{t_{i^*-n+j^*},\dots,t_{i^*-1}\} \cup \{t_{i^*},t_{j^*}\} \\
        &\cup \{t_{j^*+1},\dots,t_n\}
    \end{aligned} & \text{otherwise.}
\end{cases}
\]
Note that the goal here is to construct the largest possible balanced subset around a pair $(t_i,t_j)$. The number of tuples that we can place on either side of this pair is at most $\min\{i-1,n-j\}$. The total size of the subset is $2 + 2 \cdot \min\{i-1, n-j\}$; hence, to maximize this size, we must choose the pair maximizing the value $\min(i-1, n-j)\}$.

%\ag{It would beneficial to add a small example here.}

\subsubsection{Weakly polynomial-time instantiations}\label{sec:aggpack_instances_sum_avg}
%Assuming only \textbf{positive integer} \benny{???} values in $r_i[A]$, 
For sum and average,  we assume that values in $r[A]$ are integers given in a unary representation (hence, establish pseudo-polynomial time as stated in \Cref{thm:ptime}).
Throughout this section, we denote $r=\set{t_1,\dots, t_n}$.

\paragraph{Sum}
We begin with $\alpha=\esumagg$.
For the set $V_\alpha(r)$ we can use  
$V_{\esumagg}(r) = \set{\esumagg(r^{-}\dbr{A}),\dots,\esumagg(r^{+}\dbr{A})}$
where $r^{-}$ (resp.,~$r^{+}$) is the subset of $r$ consisting of the tuples with a negative (resp.,~positive) value in the $B$ attribute.  

For $x\in V_\alpha(r)$, the computation of $O_\alpha(r,x)$ can be done via a 
standard dynamic program for the knapsack problem, 
%\ag{cite the relevant paper/book}, 
as follows. 
The program computes the value $M[j,s]$, for $j=0,\dots,n$ and $s\in V_\alpha(r)$,
that holds the maximal size of any subset of $\set{t_1,\dots, t_j}$ with the sum $s$ of $A$ values. 
%\shunit{Unclear example - bug? Should it be $M[2,8]=3$? Why is it 2? 2 means we are only looking at the two first tuples.}
As an example, considering the relation $r$ of \Cref{tab:DB_example}, the value $M[5,8]=3$, as witnessed by the subset that corresponds to Ashley, Brandon, and Daniel.

With the $M[j,s]$ computed, $O_\alpha(r,x)$ is simply 
$M[n,x]$. 
Beginning with $j=0$, we set $$M[j,s]= 0 \mbox{ if $s=0$, and }M[j,s]= -\infty\mbox{ otherwise}\,.$$ For $j>0$ and any $s\in V_\alpha(r)$, we check whether $s-t_j[A]\in V_\alpha(r)$. If so, we set
$$M[j,s]=\max\left(M[j-1,s],1+M[j-1,s-t_j[A]]\right)\,.$$  
If $s-t_j[A]\notin V_\alpha(r)$, then we set $M[j,s]= M[j-1,s]$. %\shunit{Not sure about the last sentence, why is it needed? If $s-t_j[A]\notin V_\alpha(r)$ then $M[j-1,s-t_j[A]]$ will be $-\infty$, and the result of $M[j,s]$ will be $M[j-1,s]$ anyway.}

\paragraph{Average}
Consider now $\alpha=\eavgagg$.
For the set $V_\alpha(r)$ we can choose  
\[V_{\eavgagg}(r) = \set{
 s/\ell \mid s\in V_{\esumagg}(r)\land 
1\leq\ell\leq r
}\;.\]
For $O_\alpha(r,x)$, we construct the data structure $M'$ that has the three arguments $j$, $\ell$, and $s$ for $j=0,\dots,n$, $\ell=0,\dots,j$, and $s\in V_{\esumagg}(r)$. The value $M'[j,\ell,s]$ is Boolean, with $M'[j,\ell,s]$ being true if and only if $\set{t_1,\dots, t_j}$ has a subset with precisely $\ell$ tuples, where the $A$ values sum up to $s$. 
As an example, in \Cref{tab:DB_example}, both $M'[5,8,2]$ is true (due to Brandon and Emily) and $M'[5,8,3]$ is true (due to Ashley, Brandon, and Daniel).
%\shunit{again the first index should be at least the maximal tuple that is participating in creating the sum}
To set $O_\alpha(r,x)$, we take the maximal value $\ell$ such that 
$M'[j,\ell,s]$ is true and $s=\ell\cdot x$. (If no such $s$ and $\ell$ exist, then we set  $O_\alpha(r,x)=\emptyset$ as per our convention.)
We build $M'$ as follows.
\begin{itemize}[leftmargin=*]
\item For $j=0$, we set $M'[j,\ell,s]$ to true if $\ell=0$ and $s=0$, and false otherwise.
\item For $\ell=0$ and every $j$, we set 
$M'[j,\ell,s]$ to true if $s=0$, and false otherwise.
\item For $j>0$ and $\ell>0$, we check whether $s-t_j[A]\in V_\alpha(r)$; if so:
\[
M'[j,\ell,s]= 
M'[j-1,\ell,s] \lor
M'[j-1,\ell-1,s-t_j[A]]\,,
\]
otherwise, if $s-t_j[A]\notin V_\alpha(r)$, then
%\[
$M'[j,\ell,s]=
M'[j-1,\ell,s]$.
%\]
\end{itemize}

\subsection{Hardness of Sum and Average} \label{sec:tuple_del_hardness}
\revb{The following theorem implies that the pseudo-polynomial weakening of \Cref{thm:ptime} is required, since finding a \crepair %cardinality repair
is intractable for $\alpha=\esumagg$ and $\alpha=\eavgagg$ under the conventional binary representation of numbers. The proof is available in \Cref{app:proofs}.}

\revb{
\begin{restatable}{theorem}{thmnphardness}
Let $G\nearrow \alpha(A)$ be an AOD where $\alpha$ is either \sumagg or \avgagg. It is NP-hard to find a \crepair %cardinality repair 
(under the conventional binary representation of numbers).
\end{restatable}
}

\section{Implementation and Optimizations}\label{sec:optimizations}
%\shunit{Everyone - please review this section again, it has changed a lot.}
%While the algorithm and the instantiations described in \Cref{sec:dp} have a polynomial-time guarantee, they can be optimized for practical efficiency. 
\revb{Next, we discuss the implementation of \dpalg
(\Cref{alg:dp}), along with its instantiations.} %We refer to this algorithm as the \emph{dynamic programming} or \emph{DP} algorithm. 
We also discuss practical optimizations of the implementation. 

The implementation uses a single data structure (dictionary) $H$ instead of the $m$ different $H_i$s of \Cref{alg:dp}. To do so, in line~4 we traverse the $x$s in decreasing order, thereby assuring that we do not run over values from the previous iteration (i.e., those belonging to $H_{i-1}$) that are needed for computing $H[x]$ at iteration $i$; this holds because the computation of $H_i[x]$ requires only values $H_i[y]$ such that $y\leq x$. By doing so, we can restrict line~4 to iterate over only \emph{feasible} values $x$ of $r_i$ (i.e., the ones satisfying $O(r_i,x)\neq\emptyset$), which we compute in advance. The rest of the values are set in a previous iteration (hence, due to line~6, need not be changed). We use lazy initialization of $H$, so we avoid lines~1--2 altogether.

\subsection{\revb{Pruning based on Upper Bounds}}\label{sec:opt_greedy_aggpack}
%\ag{This is a very short subsection. Replace with a paragraph?}

\def\bound{h}

\revb{%R2.W2
\dpalg computes a solution (subset) for every feasible aggregation value and subset size. We can prune the search space to ignore subsets that are too small. 
For example, consider an \aod over a database with 1K tuples, and assume we know of a monotonic subset satisfying the \aod by removing 100 tuples. \dpalg will compute a solution for every subset size, but solutions removing more than 100 tuples will never be used, as they cannot be completed into a \crepair, and therefore, can be avoided altogether.

In this work, we propose two methods of achieving a bound on the number of tuples to remove. The first is to compute an \emph{approximate} repair using a fast algorithm that makes the best effort to build a large (but possibly not the largest) monotonic subset. Such a heuristic algorithm is described in \Cref{sec:greedy_tuple_del}.
}

%\shunit{TODO rewrite: start with a numeric example - bound of 100 tuples, of 1K tuples in the DB. the alg will compute a solution for every subset size, but solutions removing more than 100 will never be used, and can be avoided altogether. Specifically, for this optimization, we need to compute an approx. repair to get an upper bound. OR: in this work we have two ways of achieving bounds: using the heur, and using a domination relation between solutions (described in appendix)}
%To do so,}
%As we show in \Cref{sec:greedy_tuple_del}, one
%we can compute an \emph{approximate} repair using a fast algorithm that makes the best effort to build a large (but possibly not the largest) monotonic subset (as shown in \Cref{sec:greedy_tuple_del}). Such a heuristic gives us an upper bound $\bound$ on the number of tuples that we need to delete for a \revb{\crepair}. 

%The bound $\bound$ can actually be used to prune the computation of an \red{\emph{optimal}} repair. 
We utilize this pruning for the computation of $O(r_i,x)$ in the cases of \sumagg and \avgagg, which are the most computationally intensive. For that, we replace the computation of $M[j,s]$ (or, $M'[j,\ell,s]$ for \avgagg) with a new (yet conceptually equivalent) data structure $P[j, s]$ ($P'[j,\ell,s]$ for \avgagg), which holds the minimal number of tuples needed to be deleted from $\set{t_1,\dots,t_j}$ in order to satisfy the AOD, starting with $j=0$ (i.e., no tuples are removed from $r$) and $s=\alpha_r(A)$, where $P[j,s]=0$. \revb{This allows us to ignore the entries $P[j,s]$ and avoid their storage in the case where $P[j,s]>\bound$, where $\bound$ is the number of tuples removed in a given approximate repair. 
We elaborate on this optimization in \Cref{sec:heuristic_pruning}.}

\revb{
The second method exploits a dominance relation over partial solutions saved in $H$: a solution is said to dominate another if it retains more tuples and induces weaker constraints on the remaining groups. Dominated solutions can be pruned from the search space without affecting optimality. We elaborate more on this in \Cref{sec:dp_pruning}. In general, while the pruning strategies for $H$ we describe there substantially reduce $H$, their current overhead diminishes a runtime improvement.
}

%\benny{Need to add this to the appendix.}\shunit{I will do this later today (wed).}

%\red{A similar modification can be done for \avgagg.} \benny{But we said that it applies to both sum and avg. How can we say in the end that the same applies to avg? Where it is said that the previous applies only to sum?} 

%In the case of \medianagg, the bound can be used as follows. %\benny{Please complete this part???} 
%\red{For the odd-sized set case, if $x$ is too far from the median, then making it the new median would require removing more than $\bound$ tuples. Specifically, if $n$ is the size of the sorted array, then if $n^x_< + n^=_x < n/2-\bound/2$ or $n^x_= + n^>_x > n/2+\bound/2$ - we can return $r_{\text{include}}=\emptyset$. For the even-sized set case, where we consider all pairs $(i,j)$ such that their $A$ values sum to $2x$, we can focus on pairs where $i$ and $j$ are not too far from the middle of the sorted array, i.e.: $n/2-\bound/2 \leq i<j \leq n/2+\bound/2$. Otherwise, selecting $(t_i,t_j)$ to compose the median would require removing more than $\bound$ tuples.}

 In the next section, we show how to compute $O(r_i,x)$ in a holistic way for all relevant $x$, and there we also show how the bound $\bound$ can be incorporated for \avgagg, \sumagg, and \medianagg.

%In \Cref{sec:dp_pruning}, we describe further pruning strategies for $H$. In general, while these substantially reduce $H$, their current overhead diminishes a runtime improvement. %leaving challenges for future work.

%We first run the heuristic algorithm and reach a monotonic subset that satisfies the \aod and requires removing $m$ tuples. Then, the number of tuples removed in a cardinality repair is at most $m$. Therefore, we prune the search to exclude any repair that removes more tuples than the heuristic solution.

\hide{
\subsubsection{Pruning the intermediate dictionary.}\label{sec:dp_pruning_1}
% \subsubsection{Pruning $H_i$}
A notable drawback of \Cref{alg:dp} is that the size of the intermediate dictionary, $H_i$, can grow to be very large.

Let $F_i = \set{x \mid O_{\alpha}(r_i,x)\neq \emptyset}$ denote the set of feasible aggregate values of $r_i$. Each feasible value $x\in F_i$ is added to $H_i$ in line~8 of \Cref{alg:dp}. The size of $F_i$ may be exponential in the size of $r_i$ since, potentially, every subset of $r_i$ may lead to a different aggregate value. However, after computing $H_i$ (at the end of the loop of line~\red{3}), we can prune $H_i$, so that the number of keys saved in this dictionary is linear in the size of the dataset. 

\shunit{repair-> monotonic subset}
We say that a monotonic subset $r'\subseteq r$ %\ester{what is the meaning of repair here? Is it just a subset that satisfies the AOD? I don't think that we have a definition of repair in the paper, only that of a cardinality repair.} \shunit{yes, a subset that satisfies the AOD. we used to have such a definition. Should it make a come back?}
\emph{uses} a value $x$ for group $i$ %an entry $x$ of $H_i$ 
if it holds that $\alpha(r'_i\dbr{A})=x$.
An entry $x_1$ of $H_i$ is \emph{dominated} by entry $x_2$ of $H_i$ if $x_1 > x_2$ and for every monotonic subset $r'$ that uses $x_1$ for $r_i$, there exists another monotonic subset $r''$ such that $|r''|\geq|r'|$, 
%\ag{The word `dominated' suggests that there is a total order, i.e., $|r''| > |r'|$. This terminology makes it a bit unclear. In other words, if $|r''| = |r'|$ can it be that $x_1$ and $x_2$ dominate each other?}\jonny{I am not sure 'dominate' implies a total order (see game theory, where strategies can dominate each other and the ordering is partial)},
$r''$ uses $x_2$ for $r_i$, %\ester{for the same group $r_i$ or can it be a different group?}, 
and $\forall j\neq i: \alpha(r'_j\dbr{A})=\alpha(r'_j\dbr{A})$. That is, it is possible to use $x_2$ instead of $x_1$ as the aggregate value for group $i$, without compromising the optimality of the solution. 
%For each entry $x$ of $H_i$, if any repair that uses it can use another entry $x'$ of $H_i$ instead, and have the same number of tuples or more - then we can prune $x$ from $H_i$.
%\ag{It's worth recalling here that $H[x]$ is a set.}
Recall that $H_i[x]$ holds the largest subset $r'$ of $r_1\cup \dots \cup r_i$ such that $\alpha(r'_i\dbr{A})=x$.

\begin{restatable}{lemma}{lemhprune}
\label{lem:hprune}
Let $x_1$ and $x_2$ be two values in $H_i$ such that $x_1 > x_2$ and $|H_{i-1}[x_1]| \leq |H_{i-1}[x_2]|$. Then, $x_1$ is dominated by $x_2$.%\jonny{in the last paragraph $x_1$ is dominated by $x_2$ and in the lemma it is flipped.}
\end{restatable}

Given \Cref{lem:hprune}, we can prune $H_i$ in the following manner: we iterate over the keys $x$ in ascending order and keep track of the largest solution size $k$ seen so far. When considering a value of $x$, keep it only if $|H[x]|>k$.
Since the solution sizes $|H_i[x]|$ must be strictly increasing, this means that the size of $H$ after this pruning is bounded by the size of the dataset.
%\ag{It would be nice to demonstrate the pruning process if you can do it succinctly.}
Consider $r_1=\set{t\in r | t.\mbox{edu} = 1}$ from the example in \Cref{tab:DB_example}, with $r_1\dbr{A}=[1,2]$. Then $|H_1[1]| = 1, |H_1[2]| = 1, |H_1[3]| = 2$. In this case, 1 dominates 2, and we prune 2 from $H_1$.
}

\hide{
\subsubsection{Pruning the DP using a heuristic algorithm}\label{sec:dp_pruning_greedy}
%\shunit{move this to before greedy and use a general heuristic algorithm}
To accelerate the search, we can use a lower %\brit{upper? The heuristic gives a solution of size k, so the size of the optimal solution is at most k}\shunit{No, lower. note the size of the repair is the number of tuples we keep.} 
bound on the size of the repair (i.e., an upper bound on the number of removed tuples). Such a bound can come from a heuristic %an inexact \ag{Replace with `heuristic'} 
algorithm that finds a repair, and we show an example of such an algorithm in \Cref{sec:greedy_tuple_del}. 
%Let $m$ denote the number of tuples removed in the given (non-optimal) repair. 
We first run the heuristic algorithm and reach a (possibly non-optimal) repair, that requires removing $m$ tuples. The number of \revb{tuple removals required to get a \crepair} is at most $m$.
%\benny{The greedy algorithm will come only later. We should discuss about a general inexact algorithm that gives a one-sided error (more than the minimum). }
Therefore, we prune the search so as not to consider repairs that remove more tuples than the heuristic repair. %given solution \ag{What is the `given solution'? The output of the heuristic algorithm? If so, do you need to compute it first before running the DP algo.? This needs to be clarified.}.

In the intermediate result dictionary $H$, we only keep solutions that remove up to $m$ tuples from the entire relation $r$. Specifically, in line \red{8} of \Cref{alg:dp}, if $|H_{i-1}[x]\cup F_i[x]|\leq |r|-m$, then $H_{i-1}[x]\cup F_i[x]$ cannot be completed into an optimal solution, and therefore we can avoid saving $x$ in $H_i$.
}

\subsection{Holistic Packing}\label{sec:holistic}
\subsubsection{\revb{%R2.M9
Naive \procaggalpha}}\label{sec:holistic_naive}
\Cref{alg:dp} requires a packing procedure to compute the set $O_\alpha(r_i,x)$ for each value $x\in V_\alpha(r_i)$. This computation is costly, especially in the cases of \sumagg, \avgagg and \medianagg. 
%\shunit{I updated the text to include here the difference between a naive version of \procagg{\alpha} and the optimized versions in the following sections, so we can refer to them in experiments. But it's not as elegant as the previous text here was.}
\revb{For \sumagg and \avgagg, recall that we populate data structures containing the optimal subset size for each aggregation value (\Cref{sec:aggpack_instances_sum_avg}). In \Cref{alg:dp}, these data structures need to be populated from scratch for every $x\in V_\alpha(r_i)$. We can optimize the algorithm by populating these data structures once, before the iteration in line~3 and without repeating their shared computation. For \medianagg, instead of searching for a tuple (or pair) that yields a specific median $x$ for each $x\in V_\alpha(r_i)$,
we can iterate once over all tuples and tuple pairs and store the resulting medians.} %We later refer to these methods as naive \procagg{\alpha}. 
%These methods alone are more efficient than populating the data structures repeatedly for every value $x$. 

Computing the $O_\alpha(r_i,x)$ in a holistic manner for all $x\in V_\alpha(r_i)$ enables further optimizations for specific aggregation functions. 
Next, we devise such optimizations for each of the three aggregates $\alpha$; we refer to the optimized procedures uniformly as $\procaggname\langle\alpha\rangle$.
%\red{In the following algorithms detailed in this section, we construct mappings from aggregation values $x$ to the \emph{size} of $O_\alpha(r_i,x)$, and we also propose data structures for reconstruction of the maximal subset $O_\alpha(r_i,x)$.}
As we show in \Cref{sec:experiments}, the holistic computation gives considerable acceleration of the computation of a repair.

\subsubsection{\revb{%R2.M9
Optimizing \procagg{median}. %holistic packing for median
}}\label{sec:opt_histogram_aggpack}
%Recall that in \Cref{alg:dp}, for each group $r_i$, we iterate over all values $x\in V_\alpha(r_i)$, and for each one, we search for the largest subset of $r_i$ where $x$ is the median of $A$ values. As explained in \Cref{sec:agnostic_optimizations}, we can optimize this process by computing $O_{\alpha}(r_i,x)$ for all feasible values $x$ (i.e.,~values for which $O_{\alpha}(r_i,x)\neq \emptyset$) before the loop of line~4. Next, we explain how this can be done efficiently for \medianagg. This approach can also be combined with the pruning optimization from \Cref{sec:opt_greedy_aggpack}.

\revb{%R2.W1
We next describe an algorithm to compute $O_{\emedianagg}(r_i,x)$ holistically for all feasible medians $x$. This allows us to avoid the search for the specific tuples that yield each median. Furthermore, we can also iterate over unique values in $r_i$ instead of a costly iteration over tuples.
Due to its length, the pseudocode is omitted and can be found in \Cref{sec:app:holistic-median}. % \ag{Are we submitting the paper with an appendix in the final version of should we cite the full version here?}. \shunit{we submit an appendix in a separate file.}
}

As described in \Cref{sec:aggpack_instances_median}, a median $x$ is either the $A$ value of a single middle tuple (if the number of tuples if odd) or the average of the $A$ values of two middle tuples (if it is even). In what follows, we refer to the middle tuple(s) as pivot(s).

%Next, we describe an algorithm to compute $O(r_i,x)$ for every feasible value $x$, based on iterating over unique values instead of a costly iteration over tuples. 
%Due to its length, the pseudocode is omitted  and can be found in \Cref{sec:app:holistic-median}. 
The algorithm uses a histogram to keep track of the number of tuples on each side of the pivot(s). 
It maintains two data structures: a mapping $M[x]$, representing the size of the largest subset of $r_i$ with median $x$, and a mapping $D[x]$ for reconstructing this subset, containing the (one or two) pivots and the number of tuples on each side to include in the subset. 
It addresses three cases:

\begin{enumerate*}%[leftmargin=*]
    \item A single middle tuple: the algorithm iterates over the tuples $t_i$ sorted by $t_i[A]$ within a pruned range $i\in [n/2-\bound/2, n/2+\bound/2]$ determined by the repair size bound $\bound$. For each tuple, it calculates the size of the largest balanced subset that can be formed around it as the single median and updates $M$ and $D$ if a new maximum is found.
    \item Two adjacent middle tuples: A similar computation is performed for medians formed by pairs of consecutive tuples.
    \item Two nonadjacent middle tuples: the algorithm iterates over pairs of unique values from the histogram that stores, for each $A$ value, its number of appearances in $r_i$. For each pair $(a_1,a_2)$, the histogram is used to compute $k_{\text{left}}$ %$\mathsf{remaining\_left}$ 
    (the number of tuples with $A$ value at most $a_1$) and 
    $k_{\text{right}}$
    %$\mathsf{remaining\_right}$ 
    (the number of tuples with $A$ value at least $a_2$). 
    %\benny{Please change these variable names to something more mathematical, less coding.}
    The maximum size of a balanced subset around this pair is $2 \cdot \min\set{k_{\text{left}}, k_{\text{right}}}$. The structures $M$ and $D$ are updated if a new maximum is found. This method avoids the costly process of checking all pairs of nonadjacent tuples.
\end{enumerate*}

\subsubsection{\revb{%R2.M9
Optimizing \procagg{sum}%holistic packing for sum
}}\label{sec:opt_knapsack_aggpack}
\revb{%R2.W1 
We next describe an algorithm to compute $O_{\esumagg}(r_i,x)$ holistically for all feasible sums $x$, to avoid the repeated population of $M[j,s]$ (described in \Cref{sec:aggpack_instances_sum_avg}). We additionally apply a strategy of iterating over unique values instead of individual tuples (as in \Cref{sec:opt_histogram_aggpack}). We next describe an additional optimization, based on existing optimized solutions for the knapsack problem~\cite{martello1984mixture,poirriez2009}. For the latter optimization, we assume that all $A$ values are positive.
%A similar strategy of iterating over unique values using a histogram can be applied to holistic packing for \sumagg. For the additional optimization we next describe, which is based on existing optimized solutions for the knapsack problem, we assume that all $A$ values are positive. 
}

%This method uses a similar approach to optimized solutions for the knapsack problem and assumes that all $A$ values are positive. \benny{Please change the terminology... ``holistic.'' Also, it seems like we make a weakening assumption, while we handled the negative numbers in the original case. Is it applicable only to positive numbers, as opposed to the general case?} As in the case of \medianagg, instead of considering individual tuples, the algorithm operates on a histogram of unique values.

Let $v_1,\dots,v_n$ denote the unique $A$ values in ascending order and stored in the histogram.
\Cref{alg:aggpack_sum_opt} maintains two data structures: \emph{(1)} an array $M$ where an index $s$ represent a possible sum, and $M[s]$ is the maximum size of a subset with sum $s$, and \emph{(2)} a matrix $D$, where $D[j,s]$ represents the number of instances of the unique value $v_j$, used in the maximum subset for sum $s$. $D$ can also be used to reconstruct the subset for a given feasible sum $s$.
 %\benny{Why is it in parentheses? Was it said already? Where do we say it previously? Otherwise, this is a crucial point to put in parentheses.}

The algorithm builds these data structures by iteratively processing each unique value $v_j$ from the histogram in ascending order. 
The algorithm considers two distinct cases for each sum $s$:
\begin{enumerate}[leftmargin=*]
    \item No instances of $v_j$ left (lines~12-17): This is the case if the solution for the sum $s-v_j$ has already used every available copy of $v_j$, i.e., $D[j, s-v_j]=\mathsf{Hist}[v_j]$. The algorithm must therefore reevaluate how to form the sum $s$ by checking different combinations involving $v_j$ and the previous unique values.
    \item There are available instances of $v_j$ left (lines~19-21): This is the case if the solution for the sum $s-v_j$ did not require all copies of $v_j$. Meaning, $D[j,s-v_j]<\mathsf{Hist}[v_j]$. Here, we rely on an insight (stated in \Cref{lemma:dmatrix}), that the maximum subset for a sum $s$ can be found by comparing just two scenarios: the best solution for $s$ without using $v_j$, and the best solution formed by adding $v_j$ to the best subset for sum $s - v_j$.
\end{enumerate}

\begin{restatable}{lemma}{lemmaDmatrixons}\label{lemma:dmatrix}
If $D[j, s-v_j] \neq \mathsf{Hist}[V_j]$, then either $D[j,s] = 0$ or $D[j,s]=D[j,s-v_j]+1$. 
\end{restatable}
%\benny{Where is $D_j$ defined? I see only $D$.}

The proof of the lemma is in \Cref{app:proofs}, and a detailed running example can be found in \Cref{sec:sum_opt_example_full}. %\benny{Please make sure that we refer to the missing proof in every missing proof, or omit it consistently.}
As the algorithm builds the set of feasible sums from smallest to largest, it cannot be combined with the dictionary pruning optimization from \Cref{sec:opt_greedy_aggpack}. %\benny{I cannot find "top-down pruning optimization"... what is it?}

%Instead of the dynamic programming table $M[j,s]$ that holds the maximum size of a subset of $\set{t_1,\dots,t_j}$ with sum $s$ (\Cref{sec:aggpack_instances_sum_avg}), we maintain an array $M$ of the the possible sums $1,\dots \esumagg(r_i\dbr{A})$, where $M[s]$ represents the maximum size of a subset of $r_i$ with sum $s$. Additionally, we maintain an array $D[j,s]$ that holds the number of usages of unique value $v_j$ for constructing the sum $s$.

%This solution is similar to an optimized pseudo-polynomial solution to knapsack, as both algorithms maintain an array with pseudo-polynomial size, and is based on the assumption that all values in $r_i[A]$ are positive integers. All tuples $t$ with value $t.A=0$ can be added to each found solution without changing the sum. Note that due to building the set of feasible sums from small to large, this optimization cannot be combined with the pruning optimization from \Cref{sec:opt_greedy_aggpack}.
%\shunit{TODO: add end-to-end example for the algorithm.}

\begin{algorithm}[t]
\small
\DontPrintSemicolon
\SetKwInOut{Input}{Input}\SetKwInOut{Output}{Output}
\LinesNumbered
\newcommand\mycommfont[1]{\footnotesize\ttfamily\textcolor{blue}{#1}}
\SetCommentSty{mycommfont}
% \SetKwFunction{FindSol}{FindSol}
% \Input{\emph{(1)} A histogram $\mathsf{Hist}$ of ordered unique $A$ values $v_j$, where $\mathsf{Hist}[v_j]$ is the number of appearances of $v_j$ in the input $r_i$, \emph{(2)} $n=|\mathsf{Hist}|$, \emph{(3)} $T = \esumagg(r_i\dbr{A})$.}
\Input{A bag of values $r_i\dbr{A}$.}
\Output{Array $M$ mapping sums to maximum subset sizes, and an array $D$ mapping unique value ids $j$ and sums $s$ to the number of usages of $v_j$ in the construction of $s$.}

\renewcommand{\baselinestretch}{0.85}\selectfont

$\mathsf{Hist} \gets$ a histogram of $r_i\dbr{A}$, in ascending order by values $v_j$ \\
$T \gets \esumagg(r_i\dbr{A})$
$M\gets$ Array of size $T+1$ containing $-1$ in each cell and $0$ in $M[0]$ \\ 
$D \gets$ Empty array of size $|\mathsf{Hist}|\times (T+1)$\\
$Z\gets 0$ \tcp*[l]{Maximal sum reached so far.}
\For {$v_j \in \mathsf{Hist}$}{
    %$D_j\gets$ array of size $T+1$ with $0$ in each cell. \\
    $Z \gets Z + v_j \cdot \mathsf{Hist}[v_j]$ \\
    $M' \gets M$ \tcp*[l]{Avoid changing $M$ directly.}
    \For{$s \gets v_j$ \KwTo $Z$}{
        $u \gets D[j,s-v_j]$ \\ %num_used
        $s_{\text{prev}} \gets s-(u+1)\cdot v_j$\\
        \tcp{All instances of $v_j$ are used.}
        \If{$u = \mathsf{Hist}[v_j]$} {
            \For{$u' \gets 1$ \KwTo $\mathsf{Hist}[v_j]$} {
                $\mathsf{size} \gets M[s - u' \cdot v_j] + u'$ \\
                \If{$M[s - u' \cdot v_j]\geq 0$ and $\mathsf{size} > M'[s]$}{
                    $M'[s] \gets \mathsf{size}$\\
                    $D[j,s] \gets u'$ \\%\tcp*[l]{Track number of $v_j$ uses for reconstruction.}
                    }
                }
            }
            % check the option of adding u+1 vj items. (the only possibility that is not 0).
            \tcp{Additional instances of $v_j$ are available.}
            \If{$u < \mathsf{Hist}[v_j]$ \textbf{and} $M[s_{\text{prev}}]\geq 0$ \textbf{and} $M[s_{\text{prev}}] + u + 1 > M[s]$}{
                 $M'[s] \gets M[s_{\text{prev}}] + u + 1$ \\
                 $D[j,s] \gets u + 1$
            }
        }
    $M \gets \mathsf{M'}$
    %Append $D_j$ to $D$ %\benny{ending period or not?}
    }
\Return{$(M, D)$}
\caption{\procagg{sum}}
\label{alg:aggpack_sum_opt}
\end{algorithm}
% \vspace{-1em}

%\Cref{alg:sum_knapsack} maintains two main data structures. The first is an array $M$ %\ag{We already use $G$ for the grouping attribute. How about $P$ for `possible sums'?} 
%where each index represents a possible sum of a subset of $r_i$, and the value $M[s]$ represents the maximal subset size of $r_i$ with sum $s$. The second data structure is a matrix $D$, with a row for each unique value $v_j$ in $r_i[A]$ and a column for each possible sum $s$. $D[j, s]$ represents the number of instances of unique value $v_j$ used to compose a maximal subset of $r_i$ with sum $s$.
\hide{
\begin{figure*}[t]
\definecolor{mygreen}{RGB}{209, 238, 205}
\definecolor{myblue}{RGB}{203, 233, 247}
\definecolor{mypeach}{RGB}{251, 222, 201}

\begin{center}
\scriptsize
\begin{tabular}{|c|*{21}{c|}}
\hline
$s$ & \textbf{0} & \textbf{1} & \textbf{2} & \textbf{3} & \textbf{4} & \textbf{5} & \textbf{6} & \textbf{7} & \textbf{8} & \textbf{9} & \textbf{10} & \textbf{11} & \textbf{12} & \textbf{13} & \textbf{14} & \textbf{15} & \textbf{16} & \textbf{17} & \textbf{18} & \textbf{19} & \textbf{20} \\

\hline\hline
$M$ &
\cellcolor{mygreen}0 & -1 & \cellcolor{mygreen}1 & -1 & \cellcolor{mygreen}2 &
\cellcolor{mypeach}1 & -1 & \cellcolor{mypeach}2 & -1 &
\cellcolor{mypeach}3 & \cellcolor{mypeach}2 & -1 &
\cellcolor{mypeach}3 & -1 &
\cellcolor{mypeach}4 & -1 &
-1 & -1 & -1 & -1 & -1 \\
\hline
% \begin{tabular}[c]{@{}c@{}}temp $M$\end{tabular} &
% \cellcolor{mygreen}0 & -1 & \cellcolor{mygreen}1 & -1 & \cellcolor{mygreen}2 &
% \cellcolor{mypeach}1 & \cellcolor{myblue}1 & \cellcolor{mypeach}2 & \cellcolor{myblue}2 &
% \cellcolor{mypeach}3 & \cellcolor{myblue}2 & \cellcolor{myblue}2 &
% \cellcolor{mypeach}3 & \cellcolor{myblue}3 &
% \cellcolor{mypeach}4 & \cellcolor{myblue}4 &
% \cellcolor{myblue}3 & -1 &
% \cellcolor{myblue}4 & -1 & \cellcolor{myblue}5 \\
\begin{tabular}[c]{@{}c@{}}temp $M$\end{tabular} &
0 & -1 & 1 & -1 & 2 &
1 & \cellcolor{myblue}1 & 2 & \cellcolor{myblue}2 &
3 & \cellcolor{myblue}\st{2} 3 & \cellcolor{myblue}2 &
3 & \cellcolor{myblue}3 &
4 & \cellcolor{myblue}4 &
\cellcolor{myblue}3 & -1 &
\cellcolor{myblue}4 & -1 & \cellcolor{myblue}5 \\
\hline \hline
\multirow{1}{*}{$D_j\!: 2$} &  &  & \cellcolor{mygreen}1 &  & \cellcolor{mygreen}2 &  &  &  &  &  &  &  &  &  &  &  &  &  &  &  &  \\
\hline
\multirow{1}{*}{$D_j\!: 5$} &  &  &  &  &  & \cellcolor{mypeach}1 &  & \cellcolor{mypeach}1 &  & \cellcolor{mypeach}1 &   \cellcolor{mypeach}2 &  & \cellcolor{mypeach}2 &  & \cellcolor{mypeach}2 &  &  &  &  &  &\\
\hline
\multirow{1}{*}{$D_j\!: 6$} &  &  &  &  &  &  & \cellcolor{myblue}1 &  & \cellcolor{myblue}1 &  & \cellcolor{myblue}1 &  \cellcolor{myblue}1 &  & \cellcolor{myblue}1 &  & \cellcolor{myblue}1 &   \cellcolor{myblue}1 &  & \cellcolor{myblue}1 &  & \cellcolor{myblue}1\\
\hline
\end{tabular}
\end{center}

 \caption{Example for \Cref{alg:aggpack_sum_opt}. Empty cells in $D_j$ contain the default value $0$. 
 %\ag{This is taking up a lot of space. We need to think of a way to compact it. Also, it needs to be in the same page as the example referring to it (or before) to avoid confusion.}
 }
    \label{fig:sum_opt_example}
\end{figure*}
}

\hide{
\begin{example}
Consider the group $r_2$ defined by $\mbox{edu}{=}2$ from \Cref{tab:DB_example}. The values of the aggregate attribute are $2,2,5,5,6$, creating the input histogram $\mathsf{Hist}=[(2,2), (5,2), (6,1)]$ (sorted in ascending $\mbox{edu}$ order). The sum of all $r_2[A]$ values is $T=20$; hence, the size of $M, tempM$, and each $D_j$ is $22$. The cell $M[0]$ is initialized to $0$, and all other cells to $-1$.

\Cref{fig:sum_opt_example} shows the point in the algorithm where we have already performed the first two iterations of the loop in line~4 of \Cref{alg:aggpack_sum_opt}. That is, we have processed $v_j=2$ (green cells) and $v_j=5$ (orange cells). 
For example, using two instances of $v_j=5$ we reached $s=10$, so $M[10]=2$ and $D_j[5][10]=2$. Using one instance of $5$ and two instances of $2$ we have reached $s=9$, so $M[9]=3$, $D_j[5][9]=1$, and $D_j[2][9-5]$ has already been set to $2$ in the first iteration.
At this point, $Z=2\cdot 2+5\cdot 2=14$.
The last pair in the histogram is $(v_j,c_j)=(6,1)$. We increase $Z$ to $20$ (line 6). Then, $M$ is copied into $\mathsf{tempM}$, which is shown in \Cref{fig:sum_opt_example} in the white cells of $\mathsf{tempM}$. 
We iterate over $s$ values from $v_j=6$ to $Z=20$ (line~8). For example, consider $s=10$. In the previous iteration, we had $M[10]=2$ using two instances of $5$, and no instances of $6$, so $u=D_j[10-6]=0$. We now consider $M[\mathsf{noVj}]=M[s-(u+1)\cdot v_j]=M[10-6]=M[4]=2$, update $\mathsf{tempM}[10]=2+0+1$ (line~18), and $D_j[10]=1$ (line~19). 
An additional example for the use of lines~11-16 of the algorithm can be found in \Cref{sec:sum_opt_example_full}.
% As another example, consider $s=16$. We have $u=D_j[s-v_j]=D_j[10]=1$ since the instance of $6$ was already used to construct a maximum-sized subset with sum 10. The condition in line 11 holds, and we cannot use $u+1$ instances of $6$.
% Therefore we iterate over possible usages of 6 (combined with the previous values) for a way to construct $16$. In line 13 we have $M[s-u'\cdot v_j]=M[16-6]=M[10]=2$. Note that $M[10]$ still holds the solution made up of two instances of 5 (the best solution without using the current $v_j$), since the $M$ data structure only changes in the end of the iteration on $v_j$. In lines 15-16, we update $\mathsf{tempM}[16]=2+1=3$ and $D_j[16]=1$. \ag{What does this second part show that the first part doesn't? If there is no added value, we can remove it (or move to the appendix).}\shunit{It's a different case of the algorithm, demonstrating lines 11-16 whereas the previous part demonstrated lines 17-19. If you think there is no added value, we can remove.}
\qed
\end{example}
}

%We next prove the following lemma, which we used in the second case of the algorithm (lines 17-19). It means that when it is possible to add another $v_j$ element ($D_j[s-v_j] \neq c_j$), the optimal solution for $s$ will either include it ($D_j[s]=D_j[s-v_j]+1$) or not include any $v_j$ elements ($D_j[s] = 0$).

%\begin{restatable}{lemma}{lemmaDmatrixons}
%If $D_j[s-v_j] \neq c_j$, then either $D_j[s] = 0$ or $D_j[s]=D_j[s-v_j]+1$. 
%\end{restatable}
%Before we prove this lemma, we give some intuition as to its meaning. When it is possible to add another $v_j$ element ($D_j[s-v_j] \neq c_j$), the optimal solution for $s$ will either include it ($D_j[s]=D_j[s-v_j]+1$) or not include any $v_j$ elements ($D_j[s] = 0$) \ag{Just move this paragraph to be combined with the text above the lemma.}.

\subsubsection{\revb{%R2.M9
Optimizing \procagg{avg}%holistic packing for average
}}\label{sec:opt_histogram_aggpack_avg}

\def\mytable{P}

\revb{%R2.W1
As in the previous section, a holistic computation of $O_{\eavgagg}(r_i,x)$ for all feasible averages $x$ enables skipping repeated computations. In the following, we describe several additional optimizations.

First, we modify the data structure to a value-based semantic (instead of tuple-based), which greatly reduces the size of the data structure.}
Recall that for \avgagg, we used a dynamic programming table  $M'[j,\ell,s]$ that stores true if there is a subset of $\set{t_1,\dots,t_j}$ of size $\ell$ and sum $s$ of $A$ values.
We modify the algorithm to maintain a table $\mytable[s,\ell]$ of sums and sizes. %This greatly reduces the size of the data structure.

To enable the pruning optimization (described in \Cref{sec:opt_greedy_aggpack}), we use the semantics that $\mytable[s,\ell]$ stores true if there is a subset $r'\subseteq r$ of size $\ell$ such that the sum of $A$ values of $r\setminus r'$ is $s$.

Furthermore, we apply a similar approach to \procagg{median} and \procagg{sum}, and avoid iterating on individual tuples. Instead, we iterate on unique $r_i[A]$ values using a histogram, to find possible sums and sizes (thereby, possible averages). We store backtracking information in a separate data structure $D[s,\ell]$, to reconstruct the removed subset of tuples.
%To enable pruning by a given bound on the number of removed tuples as described in \Cref{sec:opt_greedy_aggpack}, we use the semantics that $\ell$ is the number of \emph{removed} tuples.

\begin{algorithm}[tbp]
\small
\DontPrintSemicolon
\SetKwInOut{Input}{Input}\SetKwInOut{Output}{Output}
\LinesNumbered
\newcommand\mycommfont[1]{\footnotesize\ttfamily\textcolor{blue}{#1}}
\SetCommentSty{mycommfont}
% \SetAlgoLined
% \DontPrintSemicolon
%\KwIn{\emph{(1)} A histogram $\mathsf{Hist}$ of ordered unique $r_i\dbr{A}$ values $v_j$, where $Hist[v_j]$ is the number of appearances of $v_j$ in the input $r_i$, \emph{(2)} $T = \esumagg(r_i\dbr{A})$, and \emph{(3)} an upper bound $\bound$ on the number of tuples to remove.}
\KwIn{A bag of values $r_i\dbr{A}$, and an upper bound $\bound$ on the number of tuples to remove.}
% A histogram \texttt{hist} of $(v_j,c_j)$ pairs, \texttt{total\_sum}, optional \texttt{max\_removed}}

%\KwOut{An array $\mytable$ mapping pairs of sums and sizes to a binary value representing the feasibility of this combination.}
\KwOut{An array $F$ mapping average values to maximum subset sizes.
}

\renewcommand{\baselinestretch}{0.85}\selectfont

$\mathsf{Hist} \gets$ a histogram of $r_i\dbr{A}$, in ascending order by values $v_j$ \;
$T \gets \esumagg(r_i\dbr{A})$ \;
$\mytable \gets$ an empty mapping with default false\;

$\mytable[T][0] \gets$ true \tcp{Removing 0 tuples yields the sum $T$.}
$D \gets$ an empty mapping \;
\For{$v_j \in \mathsf{Hist}$}{
    $\mytable' \gets$ an empty mapping with default false\;
    $D' \gets D$\;

    \For{$s \in \mytable$}{
        \For{$\ell \in \mytable[s]$}{
            \For{$k \gets 0$ \KwTo $\mathsf{Hist}[v_j]$}{
                \tcp*[l]{Add $v_j$ to each previous (removed) subset.}
                $s' \gets s - k \cdot v_j$ \;
                $\ell' \gets \ell + k$ \;

                \If{$\ell' > \bound$}{
                    \textbf{break}\tcp*[l]{Prune using bound $\bound$.}
                }
                $\mytable'[s'][\ell'] \gets$ true\;

                \If{k > 0 and ($s' \notin D'$ \textbf{or} $\ell' \notin D'[s']$)}{
                    $D'[s'][\ell'] \gets (s, \ell, v_j, k)$
                }
            }
        }
    }

    $\mytable \gets \mytable' $\; %\benny{??? db as in our DP?}
    \For{$s \in D'$}{
        \For{$\ell \in D'[s]$}{
            $D[s][\ell] \gets D'[s][\ell]$ \;
        }
    }
}
\tcp*[l]{$F$ Maps \avgagg values to maximal subset sizes.}
$F \gets$ an empty mapping \;
\For{$s \in \mytable$}{
    \For{$\ell \in \mytable[s]$}{
        $c \gets |r_i|-\ell$ \;
        $x \gets s/c$ \;
        \If{$x\notin F$ \textbf{or} $F[x] < c$}{
            $F[x] \gets c$
        }
    }
}
\Return{F} %\benny{Why is the return value different from the others?}
\caption{\procagg{avg}}
\label{alg:aggpack_avg_opt}
\end{algorithm}

%The optimization is shown in \Cref{alg:aggpack_avg_opt}. The algorithm maintains two data structures: a $\mathsf{dp}[s,\ell]$ dictionary, %that is set to true if there exists a subset of the relation of size $\ell$ that when removed, the sum of the remaining tuples is $s$, 
%and a backtracking data structure $\mathsf{bt}$, used to reconstruct the removed subset.

In \Cref{alg:aggpack_avg_opt}, 
%\benny{What is it? What is the optimized algorithm?}
we iterate over $r_i[A]$ values $v_j$ stored in the histogram (line~6). For each previously obtained sum $s$ and removed subset size $\ell$, and each possible number $k$ of instances of $v_j$, we try to remove $k$ instances of $v_j$ from $s$ (lines~13-14), to obtain a new sum $s'$ with a new removed subset size $\ell' = \ell+k$. If a bound $\bound$ is given on the number of removed tuples and  $\ell'>\bound$, we prune this option (lines~15-16).
If this is indeed a new combination of sum and subset size, we update the data structures (lines~17-19). %\benny{Are we handling the backtracking information in the other algorithms? As far as I recall, we do not, so why do we do it here?}
We then update the original data structures with the sums and sizes  achieved using $v_j$ (lines~20-23).
Note that unlike the algorithms for \procagg{median} and \procagg{sum}, the data structure used in the algorithm ($\mytable$) does not map from aggregation values (averages) to maximal subset sizes. Therefore a new mapping is built (lines~25-31), based on the sums and sizes of removed subsets.

\section{Heuristic Algorithm}\label{sec:greedy_tuple_del}
% adapted from Or's report
%\shunit{can handle the hard cases, can accelerate the other agg functions as well}
\revb{\dpalg (\Cref{alg:dp}) guarantees a \crepair, %\brit{minimal? optimal?},
but its runtime can be prohibitively long for some aggregation functions. This is the case especially for \sumagg and \avgagg that take pseudo-polynomial time and depend on the size of $V_{\alpha}(r)$ (the set of possible values that $\alpha_{r'}(A)$ can take for any $r'\subseteq r$). In this section, we present a polynomial-time greedy algorithm that serves two purposes. First, it provides a fast, inexact alternative to \dpalg. Second, it accelerates \dpalg via pruning, by providing an upper bound on the number of tuples to remove
%size of the optimal solution, 
in both the generic \dpalg algorithm (as described in \Cref{sec:dp_pruning_greedy}) and the implementation of \procaggalpha (described in \Cref{sec:opt_greedy_aggpack}), without sacrificing the optimality guarantee. Using this bound, we can avoid the computations in \dpalg searching for subsets that are too small.}

%As an alternative, we have developed a polynomial greedy algorithm, which may remove more tuples than the optimal solution, but its runtime is significantly faster. %\ag{Add complexity analysis for both the optimal and greedy algorithms to show this?}. 

%The greedy algorithm serves two purposes: (1) to provide a fast but inaccurate alternative to the DP algorithm. (2) to accelerate the DP algorithm using a bound on the size of the optimal solution.

%\subsection{General Procedure}
%\ag{This section has a single subsection. This is a weird structure...}
\revb{%R2.W5 intuition/motivation
\revb{\dpalg} iterates over each feasible aggregation value, which may correspond to every possible subset of tuples. \revb{\greedyalg} avoids this heavy iteration. Instead, it greedily selects the tuple(s) with the largest impact on the violations of the \aod. This process is repeated until the violations are eliminated, and a monotonic subset is reached.}
%The greedy algorithm removes a single tuple at a time. In each iteration, the tuple is chosen based on its impact on the violation of the \aod.
We next define the notions of violation and impact.

%The greedy algorithm (presented in \Cref{alg:greedy}) is based on the notion of violations: 
Denote the \emph{measure of violation} of the groups $r_i$ and $r_{i+1}$ by:
\begin{equation}\label{eq:mvi}
\mathsf{MVI}(r_i,r_{i+1}) = \max(0, \alpha(r_i\dbr{A})-\alpha(r_{i+1}\dbr{A}))\,.  
\end{equation}
%$$\mathsf{MVI}(r_i,r_{i+1}) = \max(0, \alpha(r_i[A])-\alpha(r_{i+1}[A])).$$
The \emph{sum of violations} over all groups $r_i$ is defined as $S_{\mathsf{MVI}}(r) = \sum_{i=1}^{|r[G]|-1}\mathsf{MVI}(r_i,r_{i+1})$.
The \emph{impact} of a tuple $t$ on the sum of violations is defined as $S_{\mathsf{MVI}}(r) - S_{\mathsf{MVI}}(r\setminus\set{t})$; this is called the ``leave-one-out'' influence and is inspired by causal counterfactuals~\cite{pearl2009causality}. %\brit{add cites}. 
That is, a tuple $t$ has a positive impact if removing it results in a decrease in the sum of violations.
%\red{If $t$ is the only tuple in $r_i$ and it is removed, then $\mathsf{MVI}(r_i,r_{i+1})=0$ and $\mathsf{MVI}(r_{i-1},r_{i})=0$.} \benny{Where is this coming from? Removed from where? Is it a definition? An observation? Who cares?}

% \begin{algorithm}[t]
% \small
% \DontPrintSemicolon
% \SetKwInOut{Input}{Input}\SetKwInOut{Output}{Output}
% \LinesNumbered
% \newcommand\mycommfont[1]{\footnotesize\ttfamily\textcolor{blue}{#1}}
% \SetCommentSty{mycommfont}
% \SetKwFunction{ComputeViolations}{ComputeViolations}
% \SetKwFunction{ComputeImpact}{ComputeImpact}
% \SetKwFunction{ComputeMVI}{ComputeMVI}
% \Input{Relation $r$, \aod $G \nearrow \alpha(A)$}
% \Output{Subset $r'\subseteq r$ such that $r'\models G \nearrow \alpha(A)$.}
% $\mathsf{MVI}\gets [\ComputeMVI(r_i, r_{i+1}) \text{ for } i\in\{1,\dots,|r[G]|-1\}]$ \\
% \While{$\esumagg(\mathsf{MVI}) > 0$}{
%      $V\gets \set{r_i | \mathsf{MVI}[i] > 0 \vee \mathsf{MVI}[i-1] > 0}$\\
%     \For {$r_i\in V$}{
%         \For {$t \in r_i$}{
%             $\mathsf{impacts}[t]\gets \ComputeImpact(t,r)$
%         }
%     }
%      $t^*\gets \argmax_{t} \set{\mathsf{impacts}[t]}$ \\
%      $r'\gets r\setminus \set{t^*}$ \\
%      %$\mathsf{alphas}\gets [\alpha(r_i[A]) \text{ for } a_i \in r[A]]$ \\
%      $\mathsf{MVI}\gets [\ComputeMVI(r'_{i},r'_{i+1}) \text{ for } i=\set{0,\dots, |r[G]|-1}]$ \\
% }

% $\text{return } r'$
% \caption{Greedy repair heuristic.}
% \label{alg:greedy}
% \end{algorithm}

\begin{algorithm}[t]
\small
\DontPrintSemicolon
\SetKwInOut{Input}{Input}\SetKwInOut{Output}{Output}
\LinesNumbered
\newcommand\mycommfont[1]{\footnotesize\ttfamily\textcolor{blue}{#1}}
\SetCommentSty{mycommfont}
% \SetKwFunction{ComputeViolations}{ComputeViolations}
% \SetKwFunction{ComputeImpact}{ComputeImpact}
% \SetKwFunction{ComputeMVI}{ComputeMVI}
% \SetKwProg{ComputeMVI}{ComputeMVI}{}{end}
\SetKwProg{Fn}{Function}{}{end}
\Input{Relation $r$, \aod $G \nearrow \alpha(A)$.}
\Output{Subset $r'\subseteq r$ such that $r'\models G \nearrow \alpha(A)$.}

\renewcommand{\baselinestretch}{0.9}\selectfont

\Fn{ComputeMVI($r$, $A$, $G$, $\alpha$)}{
    \Return $[\max\left(0,\alpha(r_i\dbr{A})-\alpha(r_{i+1}\dbr{A})\right) \text{ for } i\in\{1,\dots,|r[G]|-1\}]$\;
}
$\mathsf{MVI} \gets \texttt{ComputeMVI}(r, A, G, \alpha)$\;
%$\mathsf{MVI}\gets [\max\left(0,\alpha(r_i\dbr{A})-\alpha(r_{i+1}\dbr{A})\right) \text{ for } i\in\{1,\dots,|r[G]|-1\}]$ \\
\While{$\esumagg(\mathsf{MVI}) > 0$}{
     % $V\gets \set{r_i \mid \mathsf{MVI}[i] > 0 \text{ or } \mathsf{MVI}[i-1] > 0}$\\
    %\For {$r_i\in V$}{
        % \For {$t \in r_i$}{
        \For {$t \in r$}{
            $r'\gets r\setminus\{t\}$ \\
            $\mathsf{MVI}'\gets
            \texttt{ComputeMVI}(r', A, G, \alpha)$\;
            %[\max\left(0,\alpha(r'_i\dbr{A})-\alpha(r'_{i+1}\dbr{A})\right) \text{ for } i\in\{1,\dots,|r[G]|-1\}]$ \\

            $\mathsf{impacts}[t]\gets \esumagg(\mathsf{MVI})-\esumagg(\mathsf{MVI}')$
            %$\mathsf{MVI}'[i-1]=\max\left(0,\alpha(r_{i-1}'\dbr{A})-\alpha(r_{i}'\dbr{A})\right)$\\
            %$\mathsf{MVI}'[i]=\max\left(0,\alpha(r_{i}'\dbr{A})-\alpha(r_{i+1}'\dbr{A})\right)$\\
            %$\mathsf{impacts}[t]\gets \left(\mathsf{MVI}[i-1]+\mathsf{MVI}[i]\right) - \left(\mathsf{MVI}'[i-1]+\mathsf{MVI}'[i]\right)$
        }
    %}
     $t^*\gets \argmax_{t} \set{\mathsf{impacts}[t]}$ \\
     $r'\gets r\setminus \set{t^*}$ \\
     %$\mathsf{alphas}\gets [\alpha(r_i[A]) \text{ for } a_i \in r[A]]$ \\
     $\mathsf{MVI}\gets \texttt{ComputeMVI}(r', A, G, \alpha)$\;
     %[\max\left(0,\alpha(r'_i\dbr{A})-\alpha(r'_{i+1}\dbr{A})\right) \text{ for } i=\set{1,\dots, |r'[G]|-1}]$ \\
}

\Return{$r'$}
%\caption{Greedy repair heuristic.}
\caption{\revb{Heuristic algorithm for repair: \greedyalg.}}
\label{alg:greedy}
\end{algorithm}

% \begin{algorithm}[t]
% \small
% \DontPrintSemicolon
% \SetKwInOut{Input}{Input}\SetKwInOut{Output}{Output}
% \LinesNumbered
% \newcommand\mycommfont[1]{\footnotesize\ttfamily\textcolor{blue}{#1}}
% \SetCommentSty{mycommfont}
% \SetKwFunction{ComputeViolations}{ComputeViolations}
% \SetKwFunction{ComputeImpact}{ComputeImpact}
% \Input{Relation $r$, \aod $G \nearrow \alpha(A)$}
% \Output{Subset $r'\subseteq r$ such that $r'\models G \nearrow \alpha(A)$.}
% \tcp*[l]{Array of aggregation value for each $r_i$, defined by $A=a_i$.}
% $\mathsf{alphas}\gets [\alpha(r_i\dbr{A}) \text{ for } i\in\{1,\dots,m\}]$ \\
% $\mathsf{MVI}\gets [\max(0, \mathsf{alphas}[i]-\mathsf{alphas}[i+1]) \text{ for } i\in\{1,\dots,m-1\}]$ \\
% %\ComputeViolations(\mathsf{alphas})$\\ 
% \While{$\esumagg(\mathsf{MVI}) > 0$}{
%     %$V\gets \set{r_i | \exists a_i\in r[A]: r_i=\set{t\in r|t.A=a_i} \wedge \text{MVI}[a_i]>0}$\\  %\tcp*[l]{Violating groups}\\
%     $V\gets \set{r_i | \mathsf{MVI}[i] > 0 \vee \mathsf{MVI}[i-1] > 0}$\\
%     %$V\gets \set{r_i | \mathsf{MVI}[i]>0}$\\
%     \For {$r_i\in V$}{
%         \For {$t \in r_i$}{
%             $\mathsf{impacts}[t]\gets \ComputeImpact(t,r_i)$
%         }
%     }
%     $t^*\gets \argmax_{t} \set{\mathsf{impacts}[t]}$ \\
%     $r'\gets r\setminus \set{t^*}$ \\
%     $\mathsf{alphas}\gets [\alpha(r_i[A]) \text{ for } a_i \in r[A]]$ \\
%     $\mathsf{MVI}\gets [\max(0, \mathsf{alphas}[i]-\mathsf{alphas}[i+1]) \text{ for } i=\set{0,\dots, |\mathsf{alphas}|}]$ \\
% }
% $\text{return } r'$
% \caption{Greedy repair heuristic.}
% \label{alg:greedy}
% \end{algorithm}

\revb{\greedyalg (\Cref{alg:greedy}) iteratively removes the tuple with the highest impact, 
until all violations are resolved  (\(S_{\mathsf{MVI}} = 0\)).}
The algorithm begins by computing the initial $\mathsf{MVI}$ for every pair of consecutive groups (line~3), based on \Cref{eq:mvi}. 
%aggregation value for each group $r_i$ (line 2). It then computes the measures of violation based on the aggregation values (line 3).
As long as the sum of violations remains positive (line~4), the algorithm calculates the impact of each tuple (lines~5-8).
%identifies the set of violating groups, i.e.,~groups for which \red{the measure of violations is positive with the previous or next group} (line 3). Then, it calculates the impact of every tuple within these violating groups (lines~4-9). Observe that a tuple $t\in r_i$ only affects $\mathsf{MVI}[i-1]$ and $\mathsf{MVI}[i]$; hence, to compute its impact it is enough to recalculate these two values. 
The tuple with the maximum impact is selected for removal (line~9). Once this tuple is removed, the dataset is updated accordingly (line~10), and the measures of violations are recomputed (line~11). %The process continues iteratively until \red{$\esumagg(\mathsf{MVI})=0$, meaning that no violations remain and the relation composed of the remaining tuples satisfies the \aod.
%Since a tuple is removed in each iteration, the loop is guaranteed to terminate. In the worst case, all tuples are removed, and the \aod trivially holds.
Finally, the algorithm returns the modified relation $r'$, which satisfies the \aod (line~12). %It is also useful to maintain and return the ids of the removed tuples. \brit{the last sentence is redundant}
%\ag{I think we need a proposition and a proof for the last sentence. This is basically the correctness of the algorithm.} \shunit{do we? the sum of violations is 0, so there are no violations.. what more would the proof say?} \ag{The proof would show that the While loop is not endless.} \shunit{I added a sentence about it in red in the previous paragraph - let me know if you think a proper proof is needed.}

%Note that impact ties (multiple tuples with the same impact) %\ag{What are `impact ties'? Explain} 
%can be broken randomly, but for the sake of reproducibility, the algorithm selects the tuple from the group with the lowest index, and within that group, the tuple with the highest value.
%\benny{I don't think it's interesting... we can avoid the discussion on ties.}
%\brit{why? why not just pick a random best tuple?}
%\brit{or showed us with an example that this algo can get stuck, so we sometimes remove tuples with negative impact as well. this should be explained. }

\begin{remark}
The algorithm may choose to remove a tuple with a non-positive impact. 
In such cases, it still selects the tuple with the highest impact, which would be the least negative one.  %\brit{in this case, it selects the tuple with the smallest impact? What is the selection rule here?} 
Restricting the algorithm to remove only tuples with a positive impact might prevent it from guaranteeing a valid solution, as it could get stuck in a situation where no tuples with a positive impact exist, yet violations remain. This can happen if removing a tuple from $r_i$ reduces the violation between $r_i,r_{i+1}$ but causes another violation to increase. We next show an example of this phenomenon.
%\benny{Please explain how this is possible.}
\qed
\end{remark}

%\brit{put in Example environment. The figure is confusing. perhaps it will be easier if you put the tuples in a standard table, and specifically use the terms SMVI and compute the impact of each tuple. we need to see the formula for the impact score computation to understand that}\shunit{clearer now?} \brit{can you put the tuples in a table - then add another column for the order of removal.}

\begin{example}\label{example:non-pos-impact}
Consider the relation $r(G,A)$ composed of the following tuples: $(1,3),(1,4),(2,2),(2,3),(2,4),(3,1),(3,2)$,
%shown in \Cref{tab:negative_impact_example}, 
and the \aod $G\nearrow \emaxagg(A)$. 

% \begin{table}[htb]
% \centering
% \caption{Example of non-positive impact tuple removal. The rightmost column shows the order of tuple removal in the example. \benny{This figure is too wasteful. Please write it in text instead of a table.}}
% \begin{tabular}{|c|c|c|}
% \hline
% \textbf{G} & \textbf{A} & \textbf{Removal order} \\
% \hline
% 1 & 3 & 3\\
% 1 & \textbf{4} & 1\\
% \hline
% 2 & 2 & -\\
% 2 & 3 & 4\\
% 2 & \textbf{4} & 2\\
% \hline
% 3 & 1 & -\\
% 3 & \textbf{2} & -\\
% \hline
% \end{tabular}

% \label{tab:negative_impact_example}
% \end{table}

% \begin{align*}
% &r_1: (1,3)^{(3)},\ (1,4)^{(1)},\\ 
% &r_2: (2,2), (2,3)^{(4)},\ (2,4)^{(2)},\\
% &r_3: (3,1),\ (3,2)
% \end{align*}
%$(1,3), (1,4), (2,2), (2,3), (2,4), (3,1), (3,2)$,
%depicted in \Cref{fig:negative_impact_example},
%In the equation above, the numbers in parenthesis are the order of removal described in the example.

The relation consists of three groups $r_1, r_2$, and $r_3$, where $r_i=\set{t\in r\mid t[G] = i}$.
Any tuple $t\in r_i$ with $t[A]<\emaxagg(r_i\dbr{A})$ %the maximum in its group 
has impact $0$, as removing it does not change the aggregation value $\emaxagg(r_i\dbr{A})$.
As for the tuples with the maximum $A$ value in each group: 
in $r_1$, the tuple $(1,4)$ has impact $0$ since $\emaxagg(r_1\dbr{A}) = \emaxagg(r_2\dbr{A})$ (hence, $\textsf{MVI}(r_1,r_2)=0$), and removing this tuple only decreases the value $\emaxagg(r_1\dbr{A})$ (hence, $\textsf{MVI}(r_1,r_2)$ remains $0$). 
The tuple $(2,4)$ in $r_2$ has impact $0$ since removing it will decrease $\textsf{MVI}(r_2,r_3)$ by 1 but increase $\textsf{MVI}(r_1,r_2)$ by 1.

Finally, the tuple $(3,2)$ in $r_3$ has impact $-1$, as removing it increases $\textsf{MVI}(r_2,r_3)$ by $1$. If the algorithm was restricted to removing only tuples with a positive impact, it would stop at this point and fail to find a solution. In contrast, the algorithm will remove one of the tuples with a maximum impact ($0$). Following the tie-breaking rule (lowest group index, then highest value), the algorithm removes $(1,4)$. Now the impact of $(1,3)$ is $0$, the tuple $(2,4)$ has impact $1$, and the tuple $(3,2)$ has impact $-1$. The algorithm removes the tuple $(2,4)$. The process continues with the removal of $(1,3)$, followed by $(2,3)$, and at this point, no violations remain. \qed
%\ag{I don't follow this example. Which tuple do you end up removing? What is the final obtained repair?}.
%: $(1,1), (1,2), (2,3), (2,4), (3,2), (3,3), (3,4), (4,1), (4,1), (4,2)$

% \begin{figure}
%     \centering
%     \includegraphics[width=0.5\linewidth]{figures/greedy_stuck_example.pdf}
%     \caption{Non-positive tuple impact example.}
%     \label{fig:negative_impact_example}
% \end{figure}

\end{example}

\hide{
\paragraph{Complexity} %\ag{Refer to the lines of the algorithm when analyzing the complexity of the different steps:} 
The greedy algorithm runs in polynomial time in the size of the relation. Computing the impact of each tuple (lines~6-8) is done by comparing the outcomes of two group-by aggregation queries (with and without the tuple). \red{Let $T(\alpha)$ be the time complexity of computing the group-by-aggregation query for $\alpha$. }
%,which can be done in $O(n)$.\jonny{not sure if this is nitpicky, but in the analysis the complexity relies on the fact that the aggregation queries take O(n), which is true on functions we analyze this paper, but not true for all functions.}  
In the worst case, this calculation is performed for every tuple in each iteration. There are at most $n$ iterations of the while loop of line~2, where $n$ is the number of tuples, since a tuple is removed in each iteration. \red{Therefore, the greedy algorithm's runtime is at most $O(T(\alpha)\cdot n^2)$. If $T(\alpha)$ is linear in $n$, which is the case for the aggregation functions considered in this paper, then the greedy algorithm's runtime is $O(n^3)$.}
\benny{I'm not sure we need that... this is a trivial complexity calculation.}

%\ag{This paragraph could be a good segue to the next subsection as motivation for efficiently computing the influence for every function.}

However, for specific aggregation functions, we designed optimizations to compute the impact of a tuple more efficiently. These optimizations avoid the need to compute two group-by queries for each tuple impact calculation. We describe these next. \benny{Something here is inconsistent. You say that the next section about aggregation-specific optimizations, but it starts with non-specific optimizations. }
}

For every aggregation function $\alpha$, the impact of a tuple $t$ (computed in lines~6-8 of \Cref{alg:greedy}) can be calculated by running the query 
``\textsf{SELECT $\alpha($A$)$ GROUP BY G}'' twice: once with $t$ and once without it. Then, $S_{\textsf{MVI}}$ is calculated for each of the queries, and the impact is the difference between them. This process can be optimized to avoid the need to run these two queries for every tuple. We describe such optimizations in \Cref{sec:greedy_optimize}.

\revcommon{
%We now analyze the quality of \revb{\greedyalg} compared to the \crepair.
The following proposition states that the approximation ratio of \greedyalg can be arbitrarily high. 

\begin{restatable}{proposition}{thmGreedyRatio}\label{thm:greedy_ratio}
For $\alpha \in \set{\eavgagg,\esumagg,\emedianagg}$, there are classes of instances where the approximation ratio of \greedyalg is $\Omega(n)$.
\end{restatable}

%Note that a naive solution of removing all tuples and trivially satisfying the \aod results in $O(n)$ removals; Meaning that the bound on the approximation ratio from \Cref{thm:greedy_ratio} is tight.

%\benny{We should not call the algorithm in the theorem  ``the heuristic algorithm'' in a generic way. We should refer to its name, \greedyalg. In  general, the name \greedyalg appears in the algorithm, but never in the text of this section. This is weird -- if we choose to name it in the algorithm environment, then we should use this name throughout the section.}

%\ag{This theorem essentially means that the w.c. is very bad and this fact may detract from the achievements of this algorithm in practice. In light of this, I have two suggestions: (1) write this in text and not emphasize this as a theorem, and (2) add more details about the good performance of this algorithm in practice, e.g., it is 10x better than our DP approach...}

%\benny{Amir - this is a reaction to one of the comments. It takes an analysis to show it, and one can think of text as being a "soft" statment. I suggest dowgraing the "theorem" to "proposition."}

The proof is in \Cref{app:proofs}.
%To prove \Cref{thm:greedy_ratio}, 
For each of the three aggregations, we show a database with $n$ tuples where \greedyalg removes $\Omega(n)$ tuples while 
$O(1)$ tuples suffice for a \crepair.

Despite this large worst-case approximation ratio, the use of \greedyalg to prune the search of \dpalg can still lead to a reduction in the runtime of finding a \crepair, as we show in \Cref{sec:exp_optimizations}. For example, pruning has improved the runtime of \dpalg by up to 426 times for \avgagg over 1K tuples and by up to 15 times for \medianagg over 10K tuples.} 

%% file: 05-experiments-1.tex
\section{Experimental Evaluation}\label{sec:experiments}
%\ag{Add experimental questions}
%\shunit{What is the correct order for the research questions and the subsections?}
In our experiments, we aim to answer the following questions:
%\shunit{add reference to subsections, change terminology (optimal sol=DP=alg.1)}
\begin{enumerate*}
    \item How do the repairs perform, and what do they delete in case studies? (\Cref{sec:exp_case_studies})
    \item How do the optimizations of the DP algorithm (\Cref{sec:optimizations}) affect the runtime across various aggregation functions? (\Cref{sec:exp_optimizations})%How do the optimizations of the DP algorithm for \aod repair (\Cref{sec:optimizations}) perform in terms of run time for different aggregation functions? 
    \item How do the algorithms perform in terms of runtime, with varying data sizes and varying amounts of noise? (\Cref{sec:exp_agg_compare})
    %does the size of the data and the fraction of violation tuples affect the run time of the proposed algorithms? \ag{Maybe less specific: How do the algorithms perform in terms of runtime?}
    
    \item What is the quality of the heuristic algorithm (\Cref{sec:greedy_tuple_del}), compared to the \crepair? (\Cref{sec:exp_greedy_quality}) %\ag{Refer to the relevant algorithms/sections. For example in this question - Where are the heuristic and optimal solutions presented?}
    
    \item Can outlier removal methods be used to find an \aod repair or to accelerate \aod repairing? (\Cref{sec:outlier_exp}) %Can the problem be solved by outlier removal methods or by combining outlier removal methods with the algorithms for \aod repair?\jonny{is there anyway to rephrase this? Kinda hard to understand}
    %\item How do the proposed optimizations affect the run time of the DP algorithm?
\end{enumerate*}
%\brit{we need quantitive and qualitative experiments - namely basically quality compare to the baselines, and scalability analysis. these sets of experiments should appear in different subsections}

\paragraph{Result summary:}%\shunit{TODO}
We found that aggregation-specific optimizations of \procaggalpha in combination with pruning by a heuristic algorithm whenever possible consistently provided the largest runtime improvement. Over all real world datasets, the largest improvement was recorded for the stack overflow dataset with \avgagg, with a speedup of over 3 orders of magnitude. 
\revb{When comparing the algorithm runtimes for various aggregations, we found that in line with \Cref{thm:ptime}, the fastest aggregations were \maxagg, \countagg and \countdagg, and that the runtime advantage of the heuristic algorithm depends on the number of tuple removals necessary for a \crepair (an optimal repair). %\brit{remind the reader what is a C-Repair here}. %size of the cardinality repair. 
In terms of quality, the heuristic solution found \crepairs in all tested datasets and AODs except for Zillow (\medianagg aggregation) and Diabetes (\avgagg aggregation).} In these specific case studies, we found that the diabetes dataset deviated from the \aod $\mbox{age}{\nearrow}\eavgagg(\mbox{diabetes})$ by 0.84\% of tuples, and the Zillow dataset deviated from the \aod $\mbox{construction\ year}{\nearrow}\emedianagg(\mbox{listing\ price})$ by 2.6\% of the tuples.

\revb{While outlier removal as a preliminary step before \revb{\dpalg} was sometimes useful in speeding up the computation, it also resulted in many unnecessary tuple removals (over 10 times more than the \crepair).}
This suggests that repairing for an \aod typically requires more than just applying an outlier removal algorithm.

\subsection{Experimental Setup}
%\brit{Here is the place to talk about which machine we used to run the experiments, specific Python packages and Python version we used, and add a link to the git repository}
The implementation is in Python 3.11 and publicly available.\footnote{\url{https://anonymous.4open.science/r/aod-C28C/}}
The experiments were done on a $2.50$GHz CPU PC with $512$GB memory. 

\paragraph{Datasets and AODs}
%\brit{rename to datasets and expected trends or something similar. queries sound out of context here}
%\ag{Suggestion: group together datasets and algorithms. Basically, one subsection for the experimental background - datasets, queries, algorithms, baselines, evaluation measures, how is sampling performed...}
We analyze several publicly available datasets that contain intuitive near-trends.
%, making them suitable for investigation and repair.
%\ag{This should be an itemized list:}

\begin{itemize}[leftmargin=*] 
\item
\emph{German credit~\cite{german_credit}} (1K tuples): This dataset classifies people as low or high credit risks. We focused on the probability of being classified as low-risk over the time a person has been in their current job. We expect to observe an increasing trend, as longer job tenure may reflect greater financial stability, which is typically associated with a higher credit score.

\item \emph{Stack Overflow\footnote{\url{https://survey.stackoverflow.co/2022} (accessed April 2025)}} (31,301 tuples): This dataset consists of responses from developers worldwide to questions about their jobs, including demographics, education level, role, and salary. We focused on the average salary, grouped by the education level (from elementary school to PhD). 
We expect to see an increasing trend, suggesting that with higher education, the average salary increases.
%To make the DP algorithm feasible to run with $\alpha=\eavgagg$, 
We binned the salary attribute using equi-width bins of size \$1000, and truncated it at \$2,250,000, which is larger than 99\% of the salary values in the data.

\item  \emph{H\&M~\cite{relbench}} (1,877,674 tuples): This dataset contains transactions of customers in H\&M in 2019-2020, with features like the purchased item, the price, the age of the customer, etc. When analyzing the purchase prices by customer age, we found that overall, there was a clear trend where customers aged 25 spent the most money in H\&M, with a clear decrease until the age of 40, where the sum of transaction prices begins rising again. 
%We searched by month and found that 
The trend is violated from May to July 2020. We focused on this time period. We experimented with the sum and max aggregations.
%for this setting. 
%We also performed a preprocessing step on the prices, as most of them were between 0 and 1, we multiplied the price column by 1000 and rounded it to integers. 
\item \emph{Zillow\footnote{\url{https://www.zillow.com/research/data/}}} (629926 tuples): This dataset contains information on real estate sales in three California counties in 2016. We focus on listing prices, grouped by the year of construction, and expect listing prices to rise with construction year. %\brit{What is the expected trend here?}
%We focus on the trend of house prices by the year they were built. The years were binned to 5-year terms: 1985-1989 (inclusive) until 2010-2014 (inclusive). We observed (\Cref{fig:zillow_median}) that the median of house prices, grouped by these 5-year terms, is increasing but with an outlier in the years 2005-2009. Meaning, the \aod: $\mathsf{year}\nearrow \emedianagg(\mathsf{price})$ almost holds.

% \begin{figure}
%     \centering
%     \includegraphics[width=0.8\linewidth]{figures/zillow_value_median_over_year_range.png}
%     \caption{Property value (median) vs. year built, based on Zillow dataset.}
%     \label{fig:zillow_median}
% \end{figure}

\item \emph{Diabetes}\footnote{\url{https://www.kaggle.com/datasets/iammustafatz/diabetes-prediction-dataset}} (100K tuples): This dataset contains information about diabetes diagnosis, risk factors, and complications of individuals. We examine the trend of diabetes prevalence across age groups, anticipating higher rates among older individuals. To that end, we divided the age attribute into equal-width bins, 5 years each.
\end{itemize}
% \noindent
% \red{
% \emph{Sampling for scale experiments.} For the scale experiments shown later, we created samples of the Stack Overflow, H\&M and Zillow datasets. For Zillow and H\&M we used sample sizes of different orders of magnitude. For SO, we used smaller samples sizes (up to 1K) since the \avgagg aggregation function used with this dataset is the hardest variation of the problem. For each sample size, we performed the sampling 3 times. } \brit{this should be moved to the relevant subsection and not here}

\paragraph{Synthetic datasets.}\label{sec:synthetic_data_gen}  The synthetic data generation process is controlled by the following parameters: %aggregation function, 
number $n$ of tuples, number of groups, fraction $f$ of \addtuples, and the number of violating groups. In our experiments, we used 10 groups in total and 4 violating groups. %\brit{so the number of groups and number of violating groups is not a parameter? Are they fixed? You mentioned two sentences earlier that these are parameters.}
We first generate $n\cdot(1-f)$ tuples, by uniformly sampling values for the aggregated attribute %\brit{what do you mean by aggregate value? How is a single number an aggregate value?} 
between 1 and 100, and uniformly dividing them into groups. 
%Group aggregations are computed and the groups are sorted such that $\alpha(G_1)\leq \dots \leq \alpha(G_k)$. \shunit{TODO possibly remove the last sentence as this method was not used in our experiments. (this is done for a specific agg func and we used a generic version were we can't perform this step, the sort by one agg function might be different from other ones.)} 
Next, $f\cdot n$ \addtuples are generated by uniformly sampling aggregate values from (100, 120) and uniformly dividing them into the specified number of violating groups.  This process means that the number of tuples that need to be removed to find a repair will be correlated with $n\cdot f$. % It is not an upper bound because there is no guarantee that all "normal" groups satisfy the AOD. We used to have a step of sorting them before we add the upper range tuples, but we can't do it without knowing the aggregation function in advance, and the whole point here is to create a unified dataset to compare all aggregation functions.

%This process means $n\cdot f$ is an upper bound on the number of tuples that need to be removed to find a repair.

\paragraph{Algorithms}
%\ag{Revise this paragraph to cohesive text.}
%\ag{Are all of these valid algorithms that we actually examine or are some just baselines? Anyway, we need to specify the baselines in a separate list as well.}

\revcommon{% we changed here the names and terms, which led to rephrasing.
We evaluate two algorithms: \dpalg (\Cref{sec:dp}), which is guaranteed to return a \crepair, and \greedyalg (\Cref{sec:greedy_tuple_del}), a heuristic algorithm that returns a monotonic subset. We also use a third alternative, that combines \dpalg with outlier removal methods. We bring this alternative to understand the role of outliers in the violation of AODs, and in particular, to understand if they capture a significant part of the violation. Specifically, since outlier removal is not guaranteed to return a monotonic subset, we first apply outlier removal methods, followed by \dpalg. We measure their effect on the size of the repair and the runtime. For the combined algorithm (\emph{Outlier removal + \dpalg}), we experimented with three methods of outlier removal, of different types: Z-score~\cite{kaliyaperumal2015outlier}, Local outlier factor~\cite{breunig2000lof} (a K-Nearest-Neighbors based method), and Isolation forest~\cite{liu2008isolation} (based on decision trees).}

\subsection{\revb{Evaluation of \dpalg Optimizations}}\label{sec:exp_optimizations}
Next, we compare the \procagg{\alpha} variations with regard to the effect of optimizations on the runtime. \revb{%R2.M9
We evaluate the following variations: \emph{(1)} Naive \procaggalpha---computing $O_{\alpha}(r_i,x)$ for all $x\in V_{\alpha}(r_i)$ holistically, %, as described in the beginning of \Cref{sec:holistic}, 
without optimizations (\Cref{sec:holistic_naive}), 
\emph{(2)} aggregation-specific optimization of \procaggalpha (as described in \Cref{sec:opt_histogram_aggpack,sec:opt_knapsack_aggpack,sec:opt_histogram_aggpack_avg}), %only without pruning,
\emph{(3)} pruning \procaggalpha using a bound from a heuristic algorithm (\Cref{sec:opt_greedy_aggpack}), 
and \emph{(4)} a combination of pruning and aggregation-specific optimizations.}
%We evaluate the optimizations and their combinations. 
Experiments with dictionary pruning (\Cref{sec:opt_greedy_aggpack}) can be found in \Cref{sec:dp_pruning}.

For each aggregation, we compare the optimizations over increasing size samples from the datasets, as well as on synthetic datasets we generated.
For the scale experiments shown later, we created samples of the Stack Overflow, H\&M and Zillow datasets. For Zillow and H\&M we used sample sizes of different orders of magnitude. For SO, we used smaller samples (up to 1K) since the \avgagg aggregation function used with this dataset is the hardest variation of the problem. For each sample size (and each synthetic dataset size), we performed the sampling (or generation) three times.

The most effective optimization for \medianagg was the combination of the value-based iteration and dictionary pruning (\Cref{sec:opt_histogram_aggpack}): for real world data, the optimized algorithm was up to 860$\times$ faster than the naive algorithm (which timed out at 50K tuples). 
This combination of optimizations, adjusted for \avgagg (\Cref{sec:opt_histogram_aggpack_avg}) was 2282 times faster than the naive algorithm for 1K rows of real world data.
For \sumagg, the knapsack-based optimization of \Cref{sec:opt_knapsack_aggpack} led to the greatest improvement on real world data - up to 92$\times$ faster than the naive algorithm (which timed out at 100K tuples).
%\brit{add a short summary here: something like: the most effective optimization for avg is the knap-sack optimization (section xx), which improves runtime by yy compared to the naive implementation. for median.... }\shunit{regarding this part: "which improved the runtime by yy compared to the naive implementation." - it's not easy to find a number to write here, since the multiplicative improvement is very sensitive to the size of the data. It can be from 1.2 to 1000 times faster (not including the timeouts).}

\subsubsection{\procagg{sum} optimizations}
For \sumagg we compare variations (1)-(3), since a combination of the two optimizations is not possible. %three variations: the naive \procagg{sum}, pruning \procagg{sum} by a bound from the heuristic algorithm (\Cref{sec:opt_greedy_aggpack}), or optimizing it using the knapsack-based algorithm (\Cref{sec:opt_knapsack_aggpack}). 
The three variations are compared over increasing sample sizes from the H\&M dataset (\Cref{fig:sum_opt_hm_rows_runtime}), and increasing number of generated synthetic rows (\Cref{fig:sum_opt_synth_rows_runtime}).

\begin{figure}[tbp]
    \centering
    \subfloat[Samples of the H\&M dataset.\label{fig:sum_opt_hm_rows_runtime}]{
        \includegraphics[width=0.48\linewidth]{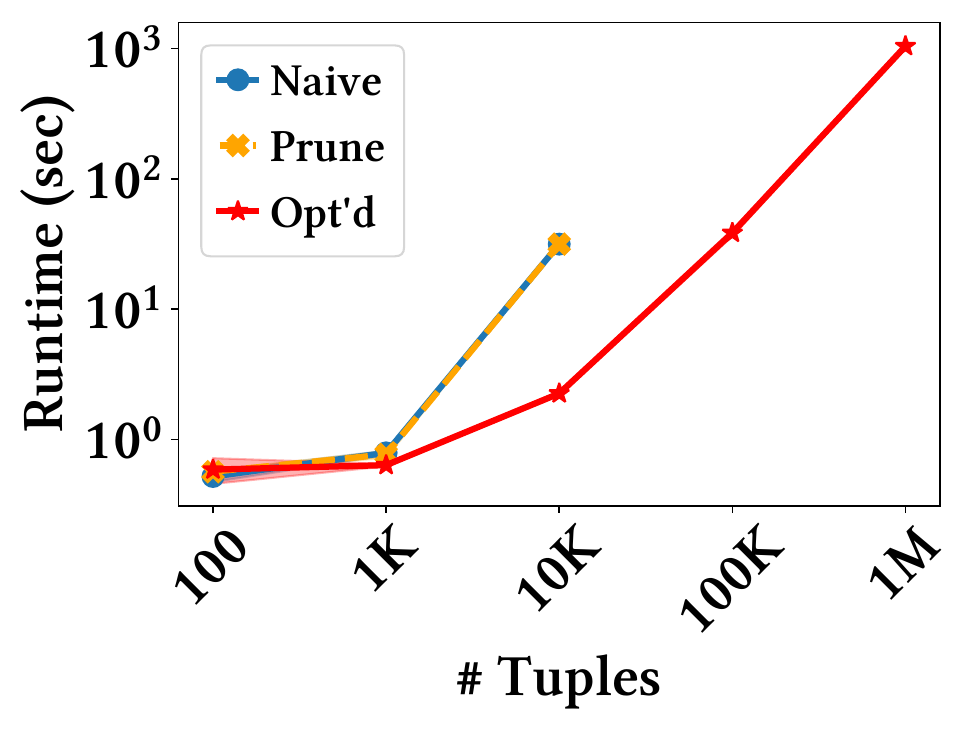}
    }
    \hfill
    \subfloat[Synthetic data.\label{fig:sum_opt_synth_rows_runtime}]{
        \includegraphics[width=0.48\linewidth]{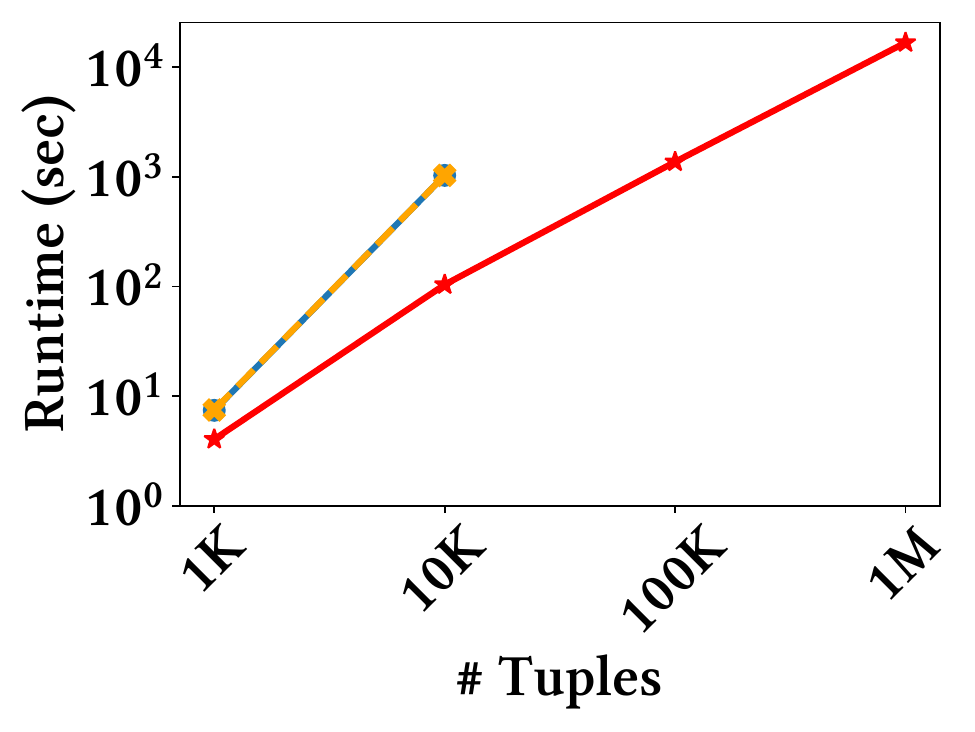}
    }
    \vskip-1em
    \caption{Comparison of runtime performance for \procagg{sum} optimizations over increasing number of tuples. In \Cref{fig:sum_opt_synth_rows_runtime}, the timeout was 20 hours. }
    \label{fig:opt_sum_rows_runtime}
\end{figure}

The most effective optimization was the knapsack-based algorithm (red in \Cref{fig:opt_sum_rows_runtime}). Pruning the search by a bound on the number of removed tuples had almost no effect on the search time. This is likely because the main bottleneck of the \procaggalpha procedure for \sumagg is in combining each tuple with all of the previously reached subsets. \revb{Additionally, the \crepairs in the H\&M samples required removing 30 tuples on average. Therefore, the bounds yielded by \greedyalg %the heuristic algorithm 
were not small enough to make a significant reduction to the runtime. Yet, for \avgagg and \medianagg, we observed a runtime decrease when pruning by the \greedyalg %the heuristic algorithm 
result.}
%\brit{Provide a sentence to explain this. Also, is this optimization beneficial for other aggregation functions? If not, we should exclude it from the paper. If it is, make sure to mention that here.}

\subsubsection{\procagg{median} optimizations}
\revb{For \medianagg, we compare four variations: the naive \procagg{median}, pruning by a bound from \greedyalg %the heuristic algorithm
(\Cref{sec:opt_greedy_aggpack}), optimizing it using value-based iteration (\Cref{sec:opt_histogram_aggpack}), or combining the two optimizations.}
We compared them over increasing samples from the Zillow dataset (\Cref{fig:median_opt_zillow_rows_runtime}), and increasing number of rows generated by the synthetic dataset generator (\Cref{fig:median_opt_synth_rows_runtime}).

In both real and synthetic data, the most effective optimization was the combination of pruning and value-based iteration. Iterating on unique values instead of tuples reduced the number of iterations. 
%(The exact effect depends on the number of unique values in the dataset.) 
Additionally, pruning the search allowed us to consider fewer options per value. 
\revb{For the heuristic pruning in the synthetic data experiment, we saw a high variance in the runtime measurements. The runtime of this optimization greatly depends on the number of tuples removed by \greedyalg,% the heuristic algorithm,
which is farthest from optimal for the \medianagg aggregation (more on this in \Cref{sec:greedy_optimize}).} 

\subsubsection{\procagg{avg} optimizations}
\revb{For \avgagg (as for \medianagg), we compare four variations: the naive \procagg{avg}, pruning the procedure by a bound from \greedyalg %the heuristic algorithm 
(\Cref{sec:opt_greedy_aggpack}), optimizing it using value-based iteration (\Cref{sec:opt_histogram_aggpack_avg}), or combining the two optimizations.}
We compared them over increasing samples from the SO dataset (\Cref{fig:avg_opt_so_rows_runtime}), and increasing number of rows generated by the synthetic dataset generator (\Cref{fig:avg_opt_synth_rows_runtime}). In both cases we sampled up to 1K rows, since \procagg{avg} still poses a computational challenge, even with the optimizations.

\begin{figure}[t]
    \centering
    \subfloat[Samples of the Zillow dataset.\label{fig:median_opt_zillow_rows_runtime}]{
        \includegraphics[width=0.48\linewidth]{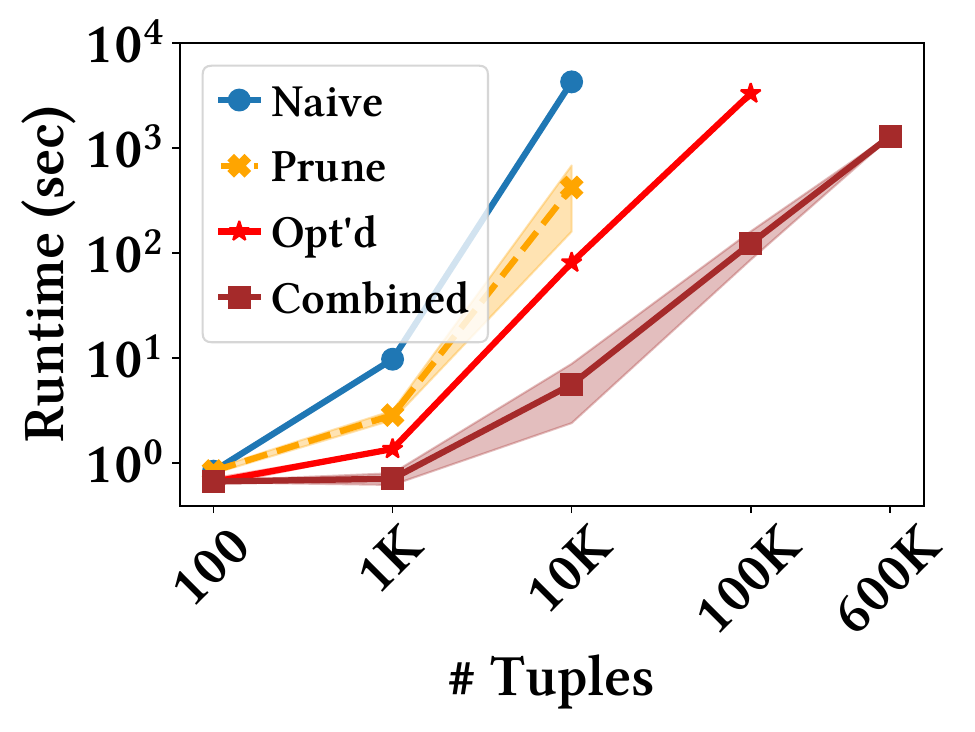}
    }
    \hfill
    \subfloat[Synthetic data.\label{fig:median_opt_synth_rows_runtime}]{
        \includegraphics[width=0.48\linewidth]{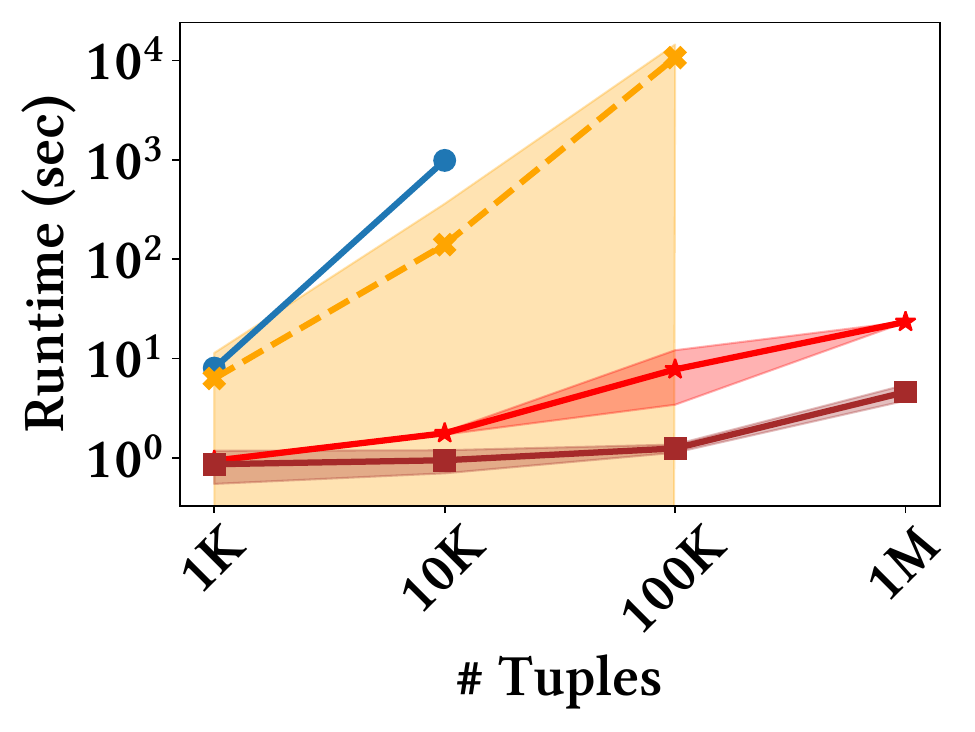}
    }
    \vskip-1em
    \caption{Runtime performance for \procagg{median} optimizations for different numbers of tuples.}
    \label{fig:opt_median_rows_runtime}
\end{figure}

\begin{figure}[b]
    \centering
    \subfloat[Samples of the SO dataset.\label{fig:avg_opt_so_rows_runtime}]{
        \includegraphics[width=0.48\linewidth]{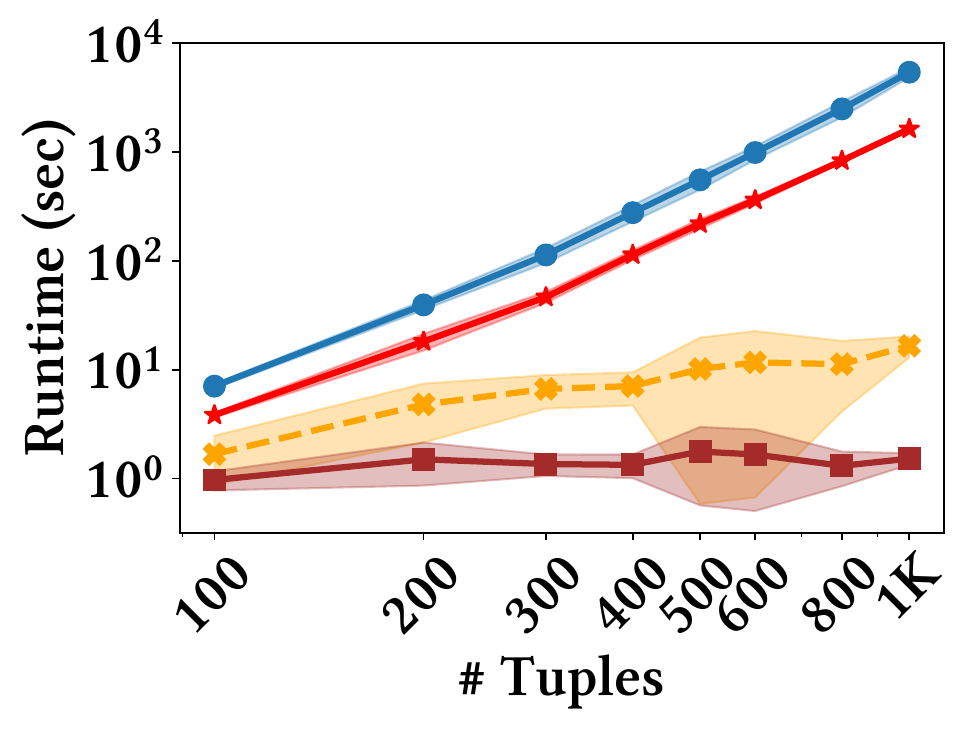}
    }
    \hfill
    \subfloat[Synthetic data.\label{fig:avg_opt_synth_rows_runtime}]{
        \includegraphics[width=0.48\linewidth]{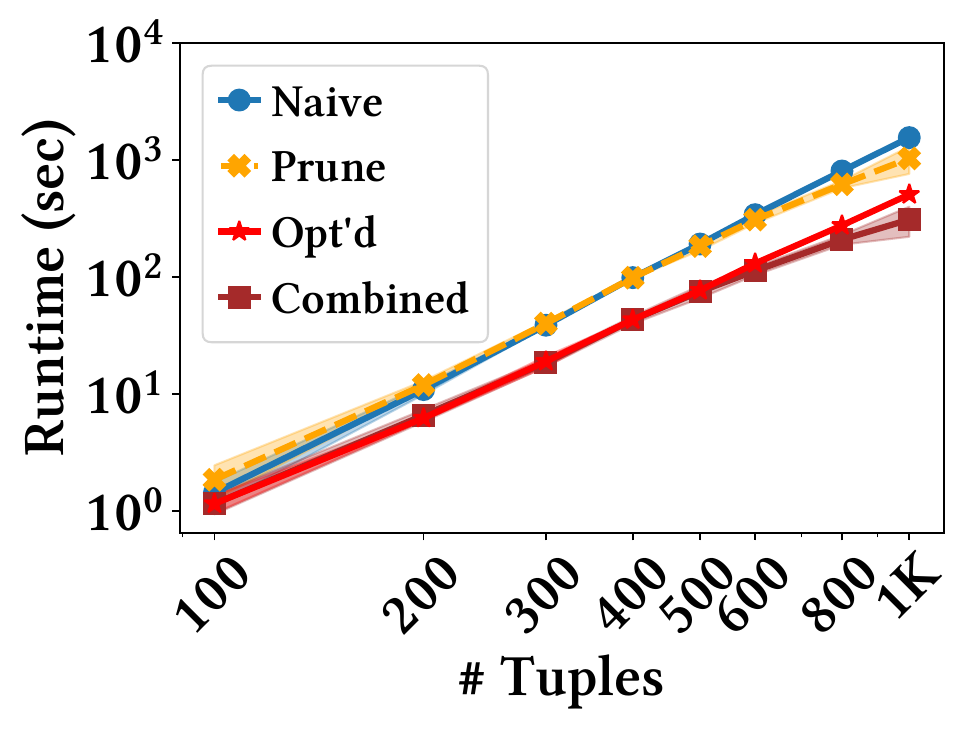}
        % synthetic was run with hprune - does that matter?
    }
    \vskip-1em
    \caption{Runtime performance for \procagg{avg} optimizations over increasing number of tuples. Both axes in log-scale.}
    \label{fig:opt_avg_rows_runtime}
\end{figure}

For \avgagg, the combination of both optimizations worked best. Pruning by the heuristic bound had a greater effect on real data than it did on the generated synthetic data, due to the number of tuples that need to be removed (and therefore the size of the bound): in the real dataset, removing a small number of tuples (0.1\%) yields a repair. In the synthetic dataset, the repair requires removing 10\% of the data (10\% of the generated tuples were \addtuples). %\brit{So you are implying the upper bound was tighter in the real dataset? I don't see your point here. be more specific} \shunit{No, the bounds were similarly tight. But the size of the optimal solution was different. I rephrased, please see if it's clearer.}

\subsection{Comparing Between Aggregations}\label{sec:exp_agg_compare}
%\jonny{why is this part interesting? The aggregations and optimizations are fundamentally different. It feels a bit like comparing apples to oranges.}
%\brit{I disagree and I thing the reviewer would want to see this. also move t before the outlier detection algorithm - this is more related to the runtime analysis}
\revb{We compare the runtimes \revb{\dpalg and \greedyalg} 
%\brit{it seems weird that only one word is colored. you can color the entire sentence.}% DP and heuristic algorithms 
on different aggregation functions. 
\dpalg %The DP algorithm 
runs substantially longer for \medianagg, \sumagg, and \avgagg.
For each aggregation function, we chose the version of \dpalg that was fastest according to the experiments in \Cref{sec:exp_optimizations}.} 
To compare all aggregations on the same dataset, we generated synthetic datasets (as described in \Cref{sec:synthetic_data_gen}) ranging from 10K to 1M tuples, with 10\% being \addtuples. 

\revb{To evaluate the effect of the amount of violation on the algorithm runtimes, we generated 100K-tuple datasets with varying numbers of \addtuples (1K to 15K).
In both experiments, for \avgagg we included only \revb{\greedyalg}, %the heuristic algorithm, 
since \revb{\dpalg}
%the DP algorithm 
with \avgagg remains a scalability challenge. For \countagg and \countdagg, we included only \revb{\dpalg} %the DP algorithm 
results; since \revb{\dpalg} %the DP algorithm 
is fast for these aggregations, we did not implement an instantiation of \revb{\greedyalg} %the heuristic algorithm 
for them.}

\begin{figure}[t]
    \centering
    \subfloat[Increasing \# tuples.\label{fig:scale_rows}]{
        \includegraphics[width=0.54\linewidth]{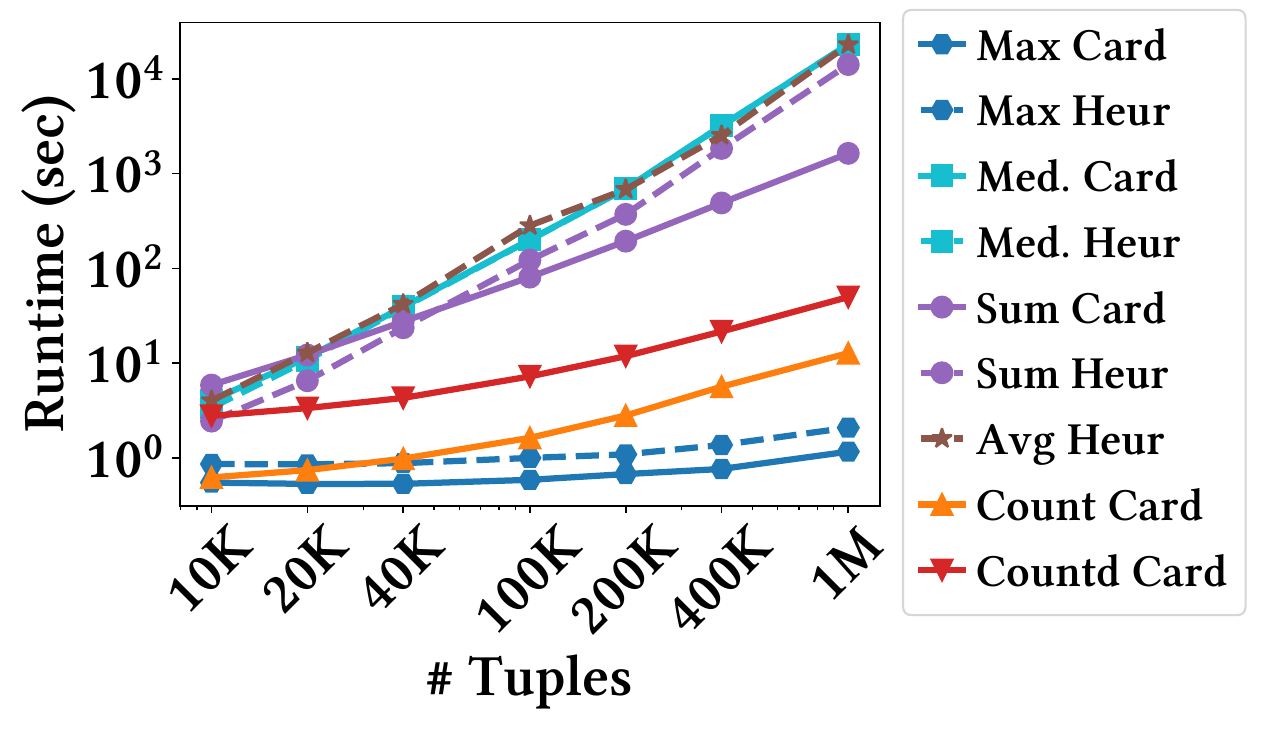}
    }
    % \hfill
    \subfloat[Increasing \# \addtuples.\label{fig:scale_violations}]{
        \includegraphics[width=0.42\linewidth]{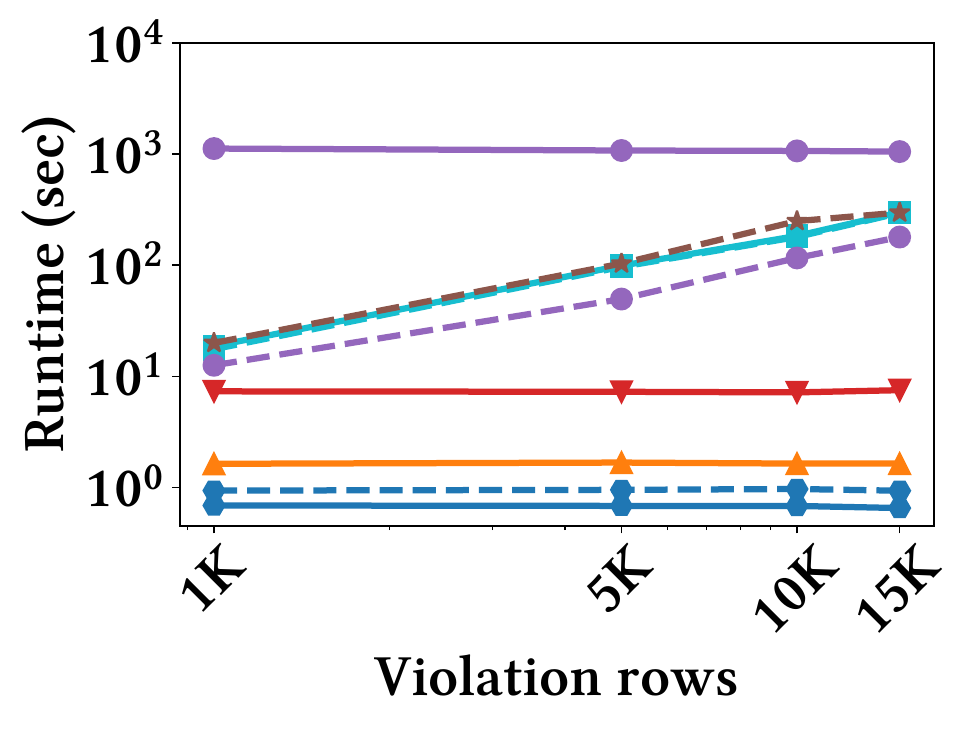}
    }
    \vskip-1em
    \caption{\revb{Run times of \dpalg and \greedyalg with various aggregations over synthetic data.}}
\label{fig:runtime_per_aggregation_synth}
\end{figure}

\revb{The runtimes versus the increasing number of tuples are shown in \Cref{fig:scale_rows} (A version of \Cref{fig:runtime_per_aggregation_synth} with confidence intervals is available in \Cref{sec:graphs_with_CIs}.)
As expected, \revb{\dpalg} %the DP algorithm 
was fastest for \maxagg, \countagg, and \countdagg, where the computation of $O_{\alpha}(r_i,x)$ is linear in the number of tuples in $r_i$. %$|r_i\dbr{A}|$ \brit{remode the readers what you mean here}. 
\revb{\dpalg} %The DP algorithm 
was slower for the non-linear aggregations: \medianagg, \sumagg and \avgagg. %\ag{This sentence should be in the first paragraph of the subsection.}.

Interestingly, for \sumagg, \revb{\greedyalg} %the heuristic algorithm 
was faster than \revb{\dpalg}
%the DP algorithm 
only for tables with fewer than 80K tuples, since \revb{\greedyalg}
%the heuristic algorithm 
removes one tuple at a time. As the number of tuples to remove increases with the size of the dataset, \revb{\greedyalg} takes longer. This can also be observed in \Cref{fig:scale_violations}, showing the runtimes versus the number of \addtuples: the time of \revb{\greedyalg} %the heuristic algorithm, 
as well as \revb{\dpalg} %the DP algorithm 
for \medianagg that uses greedy-based pruning, increase with the number of \addtuples, while other times are not affected.}

\subsection{Evaluation of the Heuristic Algorithm}\label{sec:exp_greedy_quality}
\revcommon{We next evaluate \greedyalg compared to \dpalg %the heuristic algorithm compared to DP 
in terms of quality (number of tuples removed) and their runtime. 

\underline{Result Summary:}
\greedyalg outperformed \dpalg
%The heuristic algorithm outperformed the DP algorithm
in terms of runtime across all tested cases. For \avgagg over Stack Overflow, it was more than 2000 times faster, while for \maxagg over H\&M, the speedup was more modest at 2.5 times. In terms of solution quality, \greedyalg %the heuristic algorithm
reached a \crepair for most cases, except \medianagg for Zillow and \avgagg for the Diabetes dataset. For these two cases, we include an analysis of the differences in the repairs found by the two algorithms in \Cref{sec:exp_case_studies}.}

%\brit{add a short summary here: something like: The heuristic algorithm outperformed the DP algorithm in terms of runtime across all tested cases. For the X aggregation, it was xx times faster, while for XXX, the speedup was more modest at yy times. In terms of solution quality, ... }

\revb{
We ran the best variation of \revb{\dpalg} %DP
per each aggregation function as found in the previous experiments (\Cref{sec:exp_optimizations}). For \revb{\greedyalg}, %the heuristic algorithm, 
we applied the relevant optimizations described in \Cref{sec:greedy_optimize}.
The results can be seen in \Cref{tab:tuple_del_real_datasets} and \Cref{fig:quality_full_data_rel}.}

\begin{table}[b]
\small
\centering
\caption{\revb{\dpalg and \greedyalg %DP and heuristic algorithms 
over real datasets.} 
}
\vskip-1em
\label{tab:tuple_del_real_datasets}
\renewcommand{\arraystretch}{0.95}
\begin{tabular}{|l|c|c|c|c|c|}
\hline

\textbf{Dataset} & $\alpha$ & \multicolumn{2}{c|}{\textbf{\revb{\greedyalg}}} & \multicolumn{2}{c|}{\textbf{\revb{\dpalg}}} \\
\cline{3-6}
& & \textbf{time (s)} & \shortstack{\\[0ex]\textbf{\#tuples}\\[-0.5ex]\textbf{removed}} & \textbf{time (s)} & \shortstack{\\[0ex]\textbf{\#tuples}\\[-0.5ex]\textbf{removed}}  \\

% \textbf{Dataset} & $\alpha$ & \shortstack{\textbf{Greedy}\\\textbf{time}} & \shortstack{\textbf{heuristic rows}\\\textbf{removed}} & \shortstack{\textbf{DP}\\\textbf{time}} & \shortstack{\textbf{DP rows}\\\textbf{removed}} \\
\hline
Zillow & \medianagg & 4912  & 22481 & 59904 & 16361 \\
H\&M & \sumagg & 57.53    & 1464  & 4441 & 1464  \\
H\&M & \maxagg & 0.7522   & 43    & 1.91     & 43    \\
German & \avgagg & 0.0126   & 16    & 0.19     & 16    \\
SO & \avgagg & 0.61     & 37    & 1082  & 37    \\
Diabetes & \avgagg  & 4.95	& 834	& 306	& 518 \\

\hline
\end{tabular}
\end{table}

\revb{\greedyalg %The heuristic algorithm 
was faster than \dpalg %the DP 
in all tested cases. The largest difference was for $\mbox{edu}{\nearrow}\eavgagg(\mbox{salary})$ in Stack Overflow---over 2000$\times$ faster; the smallest was for $\mbox{age}{\searrow}\emaxagg(\mbox{price})$ in H\&M (2.5$\times$ faster). For most of the tested AODs, \revb{\greedyalg} %the heuristic algorithm 
found a \revb{\crepair}. For $\mbox{age}{\nearrow}\eavgagg(\mbox{diabetes})$ and $\mbox{year}{\nearrow} \emedianagg(\mbox{value})$, it removed more tuples than necessary (1.37$\times$ and 1.61$\times$ more tuples than \revb{\dpalg},
%the cardinality repair, 
respectively).}
%
%In addition to the pitfalls of the heuristic algorithm for \medianagg described in \Cref{sec:median_difficulty}, another 
A possible reason could be that for \avgagg and \medianagg we lack monotonicity---the removal of a tuple can either increase or decrease the aggregate value of the group.
%(not only decrease it).

Another issue is that for \medianagg, the impact of the removed tuple on the aggregate value of the group (and on the sum of violations) depends on the other tuples in the group (specifically those in proximity to the current median). \revb{In these cases, \greedyalg %the heuristic algorithm 
that only sees a single step ahead is likely to make suboptimal choices.} An example of this issue is provided in \Cref{sec:median_difficulty}.

\begin{figure}[t]
    \centering
    \includegraphics[width=0.8\linewidth]{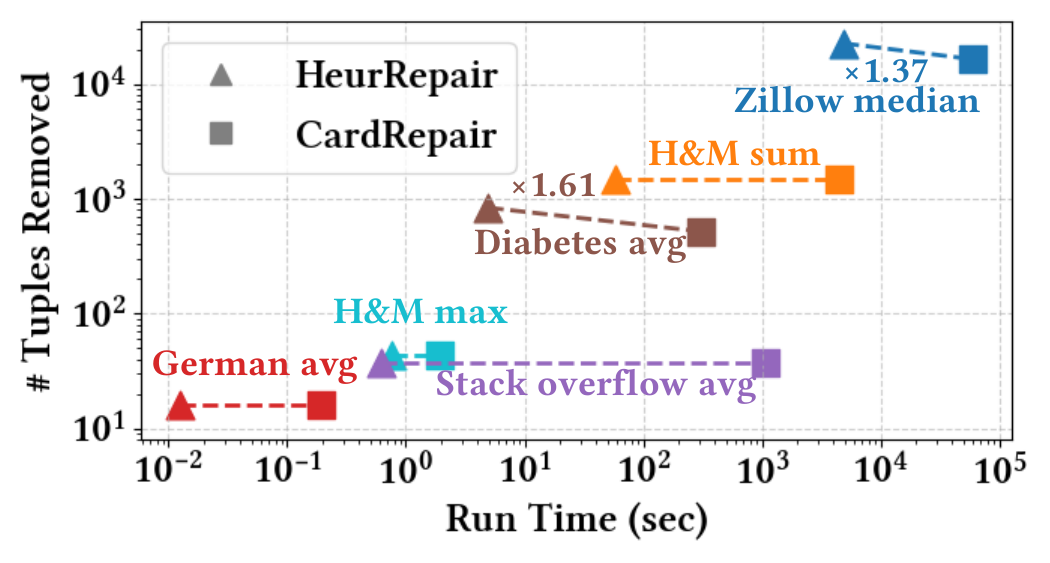}
     \vskip-1em
    \caption{\revb{Runtime and number of tuples removed (log scale) by \revb{\greedyalg and \dpalg}. % the heuristic and DP algorithms.
    The number near the dashed line represents the removal overhead of \revb{\greedyalg} %the heuristic algorithm
    (triangles) relative to \revb{\dpalg} (squares).} 
    %Since DP  finds a cardinality repair, the heuristic repair  will always be at least as high as the square. %\ag{state that the square should always be lower or equal height to the triangle and why.} 
    %\benny{Make this chart shorter please... too much white space for our space limitations. Also, please increase the font size of the chart... no reason for such small font.}
    }
    \label{fig:quality_full_data_rel}
\end{figure}

\subsubsection{\revcommon{%R1.O4
Drill Down on Removed Tuples}}\label{sec:exp_case_studies}
%Begin with an explanation of why we show these case study: 
\revcommon{We next present two in-depth analyses, for the datasets where the \greedyalg %heuristic algorithm 
reached a different solution size than \dpalg %the optimal 
(analyses for additional datasets are included in \Cref{sec:case_studies_plus}). Our goal is to highlight the differences between
%\revb{%R2.M7
\dpalg %DP 
and 
\greedyalg %heuristic repair %solutions 
in terms of the tuples they remove, and to demonstrate how our approach can help explain why a trend does not fully hold.} %\brit{what about german, SO and H\&M? } \shunit{We said to include 2 case studies here, and SO is in the introduction. What do you mean about german and H\&M? I can add their before and after figures, if we want.}

\paragraph{Diabetes}
%\shunit{I'm not sure about the terms "diabetes patients" and "healthy people" used throughout here. Any other idea?}
%\brit{I suggest to begin with a short summary, something like that:}\\
\underline{Result Summary:} 
\revb{\greedyalg %The heuristic algorithm 
removed 834 (0.83\%) tuples in 4.95 seconds, while \dpalg %the DP algorithm
removed 316 fewer tuples (0.52\% removed) but took 306.18 seconds - over 61 times more than \greedyalg, %R2.M8
as shown in \Cref{tab:tuple_del_real_datasets}}. %\brit{This is weird that only the last sentence is colored in pink. If there are comments that are common across reviewers, you should color-code them in a different color - the entire paragraph.}
%The heuristic repair completed in 4.95 seconds and removed 834 tuples involving individuals both with and without diabetes. The DP algorithm, on the other hand, took 306.18 seconds to run, and performed a cardinality repair by removing only 518 tuples, all corresponding to diabetes patients. \jonny{we repeat the number of removed tuples in 2 paragraphs which is a bit jarring. Maybe this should more in the style of conclusions? (i.e. we remove almost twice as many tuples in the greedy...)} \brit{Agree on that. A one sentence like this: Algorithm 1 removed x tuples in xx seconds, while Algorithm Y removed yy fewer tuples but took yyy seconds - an increase of zz\% in runtime compared to Algorithm 1.}

%\brit{There is no description of this dataset. We need to say that we bin the age attribute into equal-width bins}
%The diabetes dataset contains information about diabetes, risk factors, and complications, for $100K$ individuals. 
\underline{Initial \aod violations.}
Consider the average diabetes prevalence over age as depicted in \Cref{fig:diabetes_age_query_res} (blue bars). There is a visible trend: as age increases, the prevalence of diabetes is higher. However, there are some violations to this \aod: $\eavgagg_r(\mbox{diabetes}| \mbox{age}{=}15{-}19)=0.009$ while $\eavgagg_r(\mbox{diabetes}| \mbox{age}=20{-}24)=0.007$, %\ag{Has this notation been introduced before?},\shunit{yes in Section 3}
and there is a decrease among the groups $70-74, 75-79$ and $80+$. 

\begin{figure}[t]
    \centering
    \includegraphics[width=\linewidth]{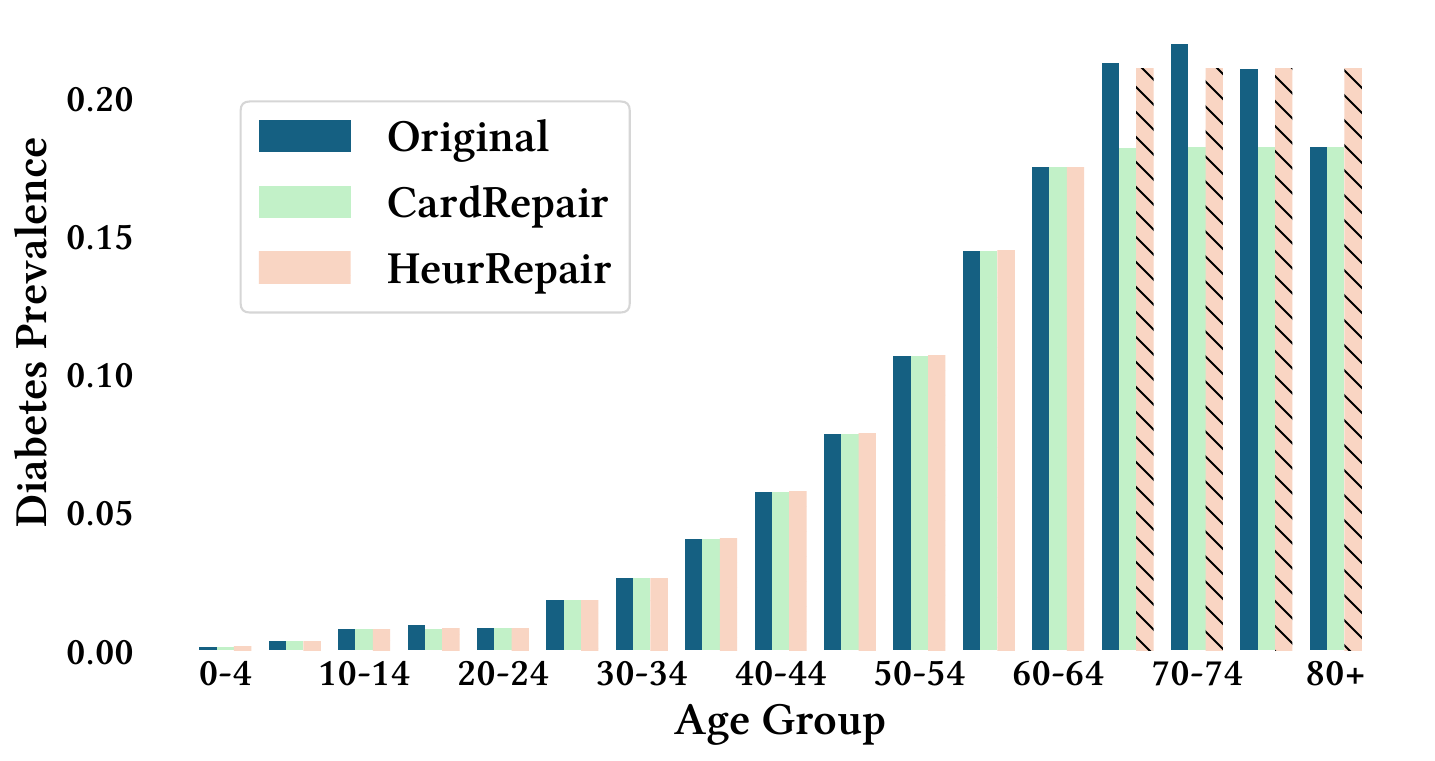}
       % \vskip-1em
    \caption{Diabetes prevalence by age. %\brit{the zillow figure looks much better. can you make it look the same?}
    }
    \label{fig:diabetes_age_query_res}
\end{figure}

\underline{Comparison of removed tuples.}
The data is close to satisfying the AOD. 
\revb{\greedyalg %The heuristic repair 
(\Cref{fig:diabetes_age_query_res}, orange bars) executed in 4.95 seconds, removing 834 tuples (0.84\%):} %\ag{out of how many?}: 
8 diabetes patients aged 15-19, to fix the first violation, as well as 60 diabetes patients aged 65-75, and 766 healthy individuals aged over 75, to fix the second violation. 
\revb{\dpalg %The DP algorithm 
(green bars) executed in %returned a cardinality repair in 
306.18 seconds and removed 518 diabetes patients (0.52\%). %\brit{also here, how much in percentage of the data}. 
Of these, it removed 8 patients aged 15-19, as \greedyalg %the heuristic algorithm 
did, and 510 patients aged 65-79.

The greedy nature of \greedyalg, %the heuristic algorithm, 
which performs the most impactful single-tuple removal at each iteration, resulted in a sub-optimal repair; it removed 1.61 more tuples. However, it had a less drastic impact on the aggregation query result, as can be seen in \Cref{fig:diabetes_age_query_res}. The aggregation values of the heuristic repair remained closer to the original values compared to those of \dpalg}. %the DP algorithm's cardinality repair.

% \hide{
% When analyzing the tuples excluded from the cardinality repair (achieved by the DP algorithm), %\brit{by heuristic or by DP?}, 
% we witnessed larger BMI scores (removed average BMI 30.7, remaining average BMI 27.3, p-value<0.001). The removed tuples also included a higher rate of former smokers (23.4\% of the removed tuples, vs. 9.3\% of the remaining tuples \red{TODO p-value}). This is consistent with smoking and large BMI scores being risk factors of diabetes \cite{diabetes_risk}. We also observed a higher rate of heart disease patients (18.5\% of the removed tuples vs. 3.9\% of the remaining \red{TODO p-value}). This aligns with heart disease being a common comorbidity for diabetes \cite{wilson1998diabetes}. \brit{I'm not sure I understand the takeaway from the last paragraph. Are you referring to the tuples removed by the DP algorithm? If so, you could clarify that these removed tuples help explain why the original trend wasn’t perfect: the individuals they represent had higher risk factors for diabetes.} \shunit{yes, these are the 518 diabetes patients removed by the DP. They do have higher risk factors, but that's not why they were removed. The only thing that matters is that they have diabetes=1 and are from the correct age groups.}
% }

\paragraph{Zillow}
%\brit{again, begin with a short summary of the results, and refer to the bars colors when you talk about the original trend, the heuristic repair and the dp repair}
\underline{Result Summary:} 
\revb{\greedyalg %The heuristic algorithm 
removed 22,481 (3.57\%) tuples in 39.8 minutes. In comparison, \dpalg %the DP algorithm 
removed 6120 fewer tuples (2.6\% ), but took 16.64 hours to run (12.2 times more than \greedyalg), % the heuristic algorithm),
%R2.M8
as shown in \Cref{tab:tuple_del_real_datasets}}.

%repair completed in 39.8 minutes and removed 22K tuples from the groups 2005-2009 and 2000-2004. The DP algorithm took 16.64 hours to run, and performed a cardinality repair by removing only 16K tuples, all from the 2005-2009 group. \brit{again, rephrase in percentage and in comparison to one another}

\underline{Initial \aod violations.}
We focus on the trend of house prices by the year they were built. The years were binned to 5-year terms: 1985-1989 (inclusive) until 2010-2014 (inclusive). We observed (\Cref{fig:zillow_median}, blue bars) that the median of house prices, grouped by these 5-year terms, is increasing but with an outlier in the years 2005-2009 (\$487K, versus \$565K in 2000-2004). That is, the \aod $\mathsf{year}\nearrow \emedianagg(\mathsf{price})$ almost holds.

\underline{Comparison of removed tuples.}
\revb{\greedyalg %The heuristic algorithm 
(\Cref{fig:zillow_median}, orange bars) terminated in 39.8 minutes, %2392.1504 seconds
removing 22,481 (3.6\%) tuples, most (22,120) from the 2005-2009 group and the rest from 2000-2004.
\dpalg
%The DP algorithm 
(\Cref{fig:zillow_median}, green bars) took 16.64 hours and removed 16,361 tuples (2.6\% of the dataset), all from the 2005-2009 group. %59904.77 seconds

Comparing the tuples removed by \dpalg %the DP algorithm 
with those that remained, we observed notable differences in land use distribution.} Among the removed tuples, 62.7\% were condominium properties and 26.91\% were single-family residences. In contrast, the remaining tuples showed an almost reverse pattern: only 34.36\% were condominium properties, while 57.31\% were single-family residences. This suggests a connection between condominium apartments built between 2005 and 2009 and lower property values, which may be related to the housing bubble in those years. We also analyzed the geographic distribution of the removed versus remaining properties. The most significant difference was observed in Los Angeles, which accounted for 29.2\% of the removed tuples but only 13.17\% of the remaining ones. A possible reason for this is that the Los Angeles housing market was one of the most rapidly deflating markets at the end of the housing bubble in 2008~\cite{baker2008housing}.

\begin{figure}
    \centering
    \includegraphics[width=0.9\linewidth]{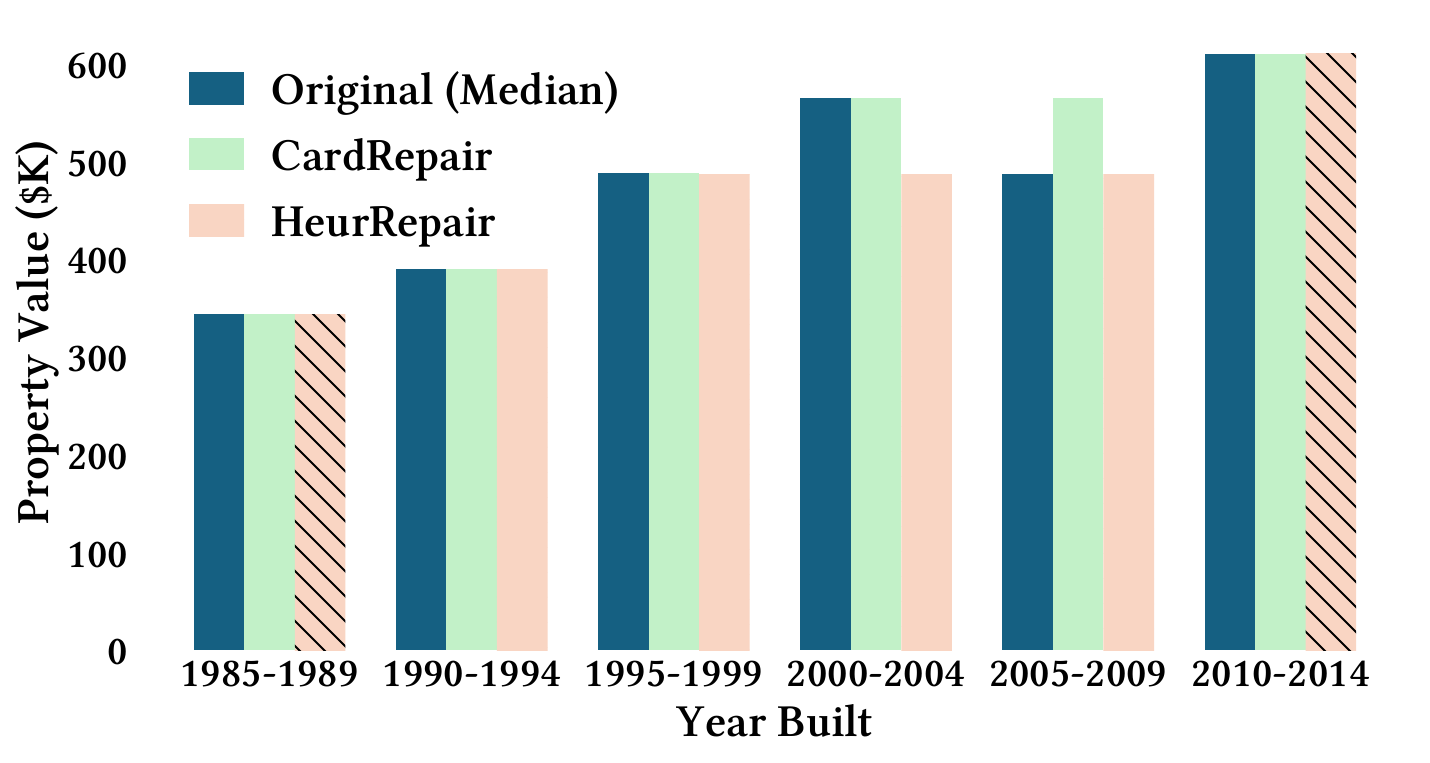}
    %\vskip-1em
    \caption{Property value median vs. year built on Zillow.
    %original Zillow dataset (blue), a cardinality repair (green), and a heuristic repair (peach). 
    %\brit{remove the numbers from the top of the bars. mentioned the relevant numbers in text if needed}
    }
    \label{fig:zillow_median}
\end{figure}

\underline{Distance from trend.}
\revb{As mentioned previously, a \crepair for an \aod can serve as a measure of the distance of a dataset from satisfying a trend.
When applying \greedyalg %the heuristic algorithm 
to find a repair for the opposite trend (a decrease in median house prices with increasing year), it removed 405K %405285 
 tuples (64.3\%), 18.03 times more than
 for the increasing trend. Hence, the distance from the increasing trend is significantly smaller than the distance from the decreasing one. Note that \greedyalg %the heuristic solution 
 does not provide any guarantees on the minimum number of tuples to be removed to satisfy each trend. For such guarantees, we need \dpalg}. %the DP algorithm.

% We also note that the heuristic algorithm took over 17 hours % 63037.1296 seconds
% to run, as the runtime of the heuristic algorithm depends on the number of removed tuples - i.e., the number of iterations in the algorithm). \brit{this cuts the flow of the story here. We also see this behavior in the runtime analysis experiment, right? If so it can be removed from here. }

\revb{When applying \dpalg %the DP algorithm 
to find a repair for the decreasing trend, the repair required removing $146.6K$ %146595 
tuples (23.3\%).} Due to the optimality of the algorithm, it provides a lower bound on the number of tuples to be removed to satisfy the trend in each direction. Therefore, the dataset is 8.96 times closer to satisfying an increasing trend than to satisfying a decreasing trend.

\revcommon{ % R1.O3
%\subsubsection{Note on algorithm behavior}\label{sec:alg_behave}
\smallskip

From \Cref{fig:diabetes_age_query_res,fig:zillow_median}, it may seem that \dpalg %the algorithm 
only equates the aggregation values of consecutive groups.
Note, however, that in \Cref{fig:diabetes_age_query_res}, the values are close but not identical, while in \Cref{fig:zillow_median} (with \medianagg aggregation) they became equal. 
%Whether the repair will result in equal values of consecutive groups depends on the values in the input database, and the choice of aggregation function. 
%
This behavior (making two consecutive groups have close values) is an inherent property of the problem. When each tuple has a small effect on the aggregation value of the group, a \crepair will result in close aggregation values. 
%Consider the case of only two groups, where the values for groups 1 and 2 are $x$ and $y$, s.t. $x>y$. Both \revb{\dpalg and \greedyalg} %the DP and heuristic algorithms 
%will either remove tuples from group 1 to decrease $x$ or from group 2 to increase $y$.  As soon as $x\leq y$, there is no reason to remove additional tuples. 
%If the values of individual tuples are small or close to each other, then $x$ and $y$ would be very close.
%Given this behavior, one may think that the problem is simple to solve: for each violating group (with a larger aggregation value than its predecessor), search for a subset of the group that yields the value of the predecessor. However, this approach is similar to \revb{\greedyalg}, % the heuristic algorithm shown above, 
%and a simplification of this approach at that. 
The heuristic algorithm \greedyalg aims to utilize this property by trying to match the value of a group with its predecessor group. Yet, as we saw in the experiments, this approach may be overly simplistic and may be considerably suboptimal (see \Cref{sec:exp_case_studies} and \Cref{fig:quality_full_data_rel}).
}

\subsection{\revcommon{Comparison to Outlier Detection Methods}}
\label{sec:outlier_exp}
%\shunit{To save space, maybe we should include here a subset of the table, with only the best parameter combinations, and have the full table in the appendix.}
\revcommon{ 
% Why we chose this as a baseline
As an alternative to repairing for an \aod, we compare against methods that remove statistically identified outliers. This choice reflects a common strategy in data cleaning, where violations of expected trends are attributed to noisy or erroneous records~\cite{xiong2006enhancing,borrohou2023data}, rather than to fundamental structure in the data. Like our algorithms, outlier removal methods work by deleting a subset of tuples, but they rely on statistical rarity rather than explicitly targeting monotonicity violations. This comparison allows us to quantify the benefit of using \aod-aware repair algorithms over generic noise-removal approaches.
Therefore, we next aim to answer our 5th research question from the beginning of \Cref{sec:experiments}: %\Cref{sec:experiments}: 
can existing outlier removal methods be used to find a repair for an \aod?

We consider three settings: the SO dataset with the \aod $\mbox{edu}\nearrow\eavgagg(\mbox{salary})$, and H\&M with two AODs: $\mbox{age}\searrow\emaxagg(\mbox{price})$ and $\mbox{age}\searrow\esumagg(\mbox{price})$. Out of the real-world scenarios, these are most likely to be affected by outliers: for Diabetes and German credit, all values of the aggregated attribute are 0 or 1 (there are no outliers to be found); and the \medianagg function used with Zillow is less affected by outliers than \avgagg, \maxagg and \sumagg.

\input{outlier-table}

We experimented with three outlier removal methods using the scikit-learn implementations with hyper-parameter tuning.
\begin{enumerate*}%[leftmargin=*]
\item \emph{Z-score}~\cite{kaliyaperumal2015outlier} models $r\dbr{A}$ as samples from a normal distribution, and removes points whose distance from the mean is over $\tau$ standard deviations. We experimented with $1\leq \tau \leq 6$.
\item \emph{Local outlier factor (LOF)}~\cite{breunig2000lof} detects samples with substantially lower density compared to the $k$ closest neighbors are considered outliers. We tried $k$ values in $\set{3,5,8,10}$.
\item \emph{Isolation forest}~\cite{liu2008isolation} trains random decision trees to isolate observations; easily isolated samples (with shorter path lengths in the decision tree) are considered outliers. A contamination parameter controls the number of expected outliers. We tried contamination values between 0.001 and 0.2.
\end{enumerate*}
For each method, we also examined a variation where the outlier removal is done per group $r_i$ rather than the full dataset. 

For each dataset setting, we applied all combinations of outlier removal method and hyper-parameters as a preliminary step, selected the combinations that were closest to a repair according to the sum of violations after removal, and applied \dpalg after the removal to find a \crepair. %\ag{This should be described in the setup section.}. \shunit{Why? it's part of the experiment}

% Z-score~\cite{kaliyaperumal2015outlier}, Local Outlier Factor~\cite{breunig2000lof}, and Isolation Forest~\cite{liu2008isolation}, with hyper parameter tuning. 

%We performed a hyper parameter search for each method: we experimented with $1\leq \tau \leq 6$ for Z-score, $k$ For each outlier removal method we performed a hyper parameter search: for Z-score, we experimented with a score threshold (as the threshold is larger, less values are considered outliers). For Local Outlier Factor, which is a K-Nearest-Neighbors based method, we varied the number of neighbors K. As for isolation forest, we tried several values for the contamination parameter (the expected fraction of outliers in the dataset). For each method, we also examined a variation where the outlier removal is done per group $r_i$ versus on the full dataset. 

The results of the best parameters for each combination of scenario and outlier removal method are shown in \Cref{tab:outlier_removal_res} (the full tables are available in \Cref{sec:outlier_exp_full}). For each outlier removal method, we report the number of removed outliers and the total number of tuple removals including those removed by \dpalg. 
Note that due to the size of the dataset, the LOF method timed out at 2 hours for the H\&M dataset (for comparison, \dpalg took under 1 minute for this dataset). %\shunit{something about this time being longer than it took the repair alg?}
%\red{TODO edit. also add that LOF for H\&M was not scalable (1M rows), timed out at 1.5 hours which is more than what the DP takes.} 
%%For each set of variations of a method, we highlight (in bold) the variation that was closest to a repair (the minimal $S_{\mathsf{MVI}}$, and minimal number of removed tuples in the case of a tie). For these best variations, we ran the DP algorithm to find a cardinality repair after the outliers were removed.

Most variations of outlier removal methods did not yield a monotonic subset (namely, a valid solution). %\brit{If I understand correctly, a C-Repair is the optimal solution, or a solution not necessarily the optimal one? If the second is correct, so why do you call it a monotonic subset?}. 
For the \sumagg scenario, none of the methods yielded a monotonic subset, all of them removed over 30 times more tuples than necessary, and \dpalg had to remove a similar number of tuples as without the preliminary outlier detection.
For the \avgagg and \maxagg scenarios, some of the outlier removal methods did return a monotonic subset satisfying the \aod, but they required removing 10 times more tuples than the \crepair.

Still, outlier removal can occasionally accelerate \dpalg. For example, for SO with \avgagg, Isolation forest with 0.05 contamination (in both the group-wise and the non group-wise settings), there were only a few remaining tuples to remove. In both cases, \greedyalg
yielded tight bounds, thus greatly shortening the runtimes of \dpalg (60.5 seconds and 115.9 seconds, respectively).

Overall, we conclude that \emph{repairing for an \aod 
%is a different problem from outlier detection, 
typically requires more than just applying outlier detection, 
as the tuples that obscure the trend are not necessarily statistical outliers}. They may appear structurally standard, but still interfere with the expected trend.
}

%% file: outlier-table.tex
\begin{table}[t]
\caption{\revcommon{Outlier removal on 3 settings. For each method, the best variation %with the minimal sum of violations (and minimal number of removed tuples) 
is shown for global/per-group outlier removal.} %$S_\mathsf{MVI}$ is the sum of violations after applying outlier removal, and the rightmost column is the total number of removals (including both outliers and the repair). \shunit{TODO the table is too wide}} %\brit{The two last columns' names are not informative and not explained in the text. add to the caption theri meaning}
}
\label{tab:outlier_removal_res}
\vskip-0.5em
\footnotesize
\begin{tabular}{llllll}
\toprule

%\textbf{Setting} & \makecell{\textbf{Per}\\\textbf{Group}} & \textbf{Method} & \textbf{Param.} &  \textbf{\# Removed} & \textbf{$S_\mathsf{MVI}$} & \makecell{\textbf{\revb{\dpalg}}\\\textbf{removed}} \\
\textbf{Setting} & \textbf{Method} & \makecell{\textbf{Best}\\\textbf{Param}} &  \makecell{\textbf{\# Removed}\\\textbf{outliers}} & \makecell{\textbf{\# Removed}\\\textbf{total}} \\

\midrule

% SO \avgagg & global & Z score & 6 & 369 & 0  & 0\\
% SO \avgagg & global & LOF & 3 & 318 & 25667 & 38\\
% SO \avgagg & global & Iso.Forest & 0.05 & 1519 & 6515 & 7\\
% SO \avgagg & per-group & Z score & 6 & 373 & 0 & 0\\
% SO \avgagg & per-group & LOF & 8 & 972 & 2142 & 18\\
% SO \avgagg & per-group & Iso.Forest & 0.05 & 1556 & 4424 & 5\\
% \hline
% H\&M \maxagg & global & Z score & 6   & 1356  & 0 & 0 \\
% H\&M \maxagg & global & Iso.Forest & 0.005 & 8489  & 0 & 0 \\
% H\&M \maxagg & per-group & Z score & 2    & 83664  & 4 & 1270 \\ % dp time: 2.92
% H\&M \maxagg & per-group & Iso.Forest & 0.01  & 16657 & 13 & 4601 \\ % DP time: 3.05
% \hline 
% H\&M \sumagg & global & Z score & 1  & 462973 & 88112 & 2797 \\ %DP time 3333.54
% H\&M \sumagg & global & Iso.Forest & 0.2 & 367546 & 99800 & 2270 \\ % DP time 2495.144
% H\&M \sumagg & per-group & Z score & 1.5 & 117349 & 89214 & 1581 \\ % DP time 18,790 sec
% H\&M \sumagg & per-group & Iso.Forest & 0.01  & 16657  & 83644 & 1426 \\ % DP time 11680 sec

%global/pergroup - WITH Smvi
% SO \avgagg & Z score & 6/6 & 369/373 & 0/0  & 369/373 \\
% SO \avgagg & LOF & 3/8 & 318/972 & 25667/2142 & \textbf{356}/990\\
% SO \avgagg & Iso.Forest & 0.05/0.05 & 1519/1556 & 6515/4424 & 1526/1561\\

%global/pergroup
SO \avgagg & Z score & 6/6 & 369/373  & 369/373 \\
SO \avgagg & LOF & 3/8 & 318/972 & \textbf{356}/990\\
SO \avgagg & Iso.Forest & 0.05/0.05 & 1519/1556 & 1526/1561\\
% SO \avgagg & per-group & Z score & 6 & 373 & 0 & 0\\
% SO \avgagg & per-group & LOF & 8 & 972 & 2142 & 18\\
% SO \avgagg & per-group & Iso.Forest & 0.05 & 1556 & 4424 & 5\\
\hline
%global/pergroup
\revcommon{H\&M \maxagg} & \revcommon{Z score} & \revcommon{6/2}   & \revcommon{1356/83664}  &  \revcommon{\textbf{1356}/84934} \\
\revcommon{H\&M \maxagg} & \revcommon{Iso.Forest} & \revcommon{0.005/0.01} & \revcommon{8489/16657}  & \revcommon{8489/21258} \\

%global/pergroup - WITH Smvi
% \revcommon{H\&M \maxagg} & \revcommon{Z score} & \revcommon{6/2}   & \revcommon{1356/83664}  & \revcommon{0/4} & \revcommon{\textbf{1356}/84934} \\
% \revcommon{H\&M \maxagg} & \revcommon{Iso.Forest} & \revcommon{0.005/0.01} & \revcommon{8489/16657}  & \revcommon{0/13} & \revcommon{8489/21258} \\

% H\&M \maxagg & per-group & Z score & 2    & 83664  & 4 & 1270 \\ % dp time: 2.92
%H\&M \maxagg & per-group & Iso.Forest & 0.01  & 16657 & 13 & 4601 \\ % DP time: 3.05

\hline 
%global/pergroup - WITH Smvi
% \revcommon{H\&M \sumagg} & \revcommon{Z score} & \revcommon{1/1.5}  & \revcommon{462973/117349} & \revcommon{88112/89214} & \revcommon{465770/118930} \\ %DP time 3333.54
% \revcommon{H\&M \sumagg} & \revcommon{Iso.Forest} & \revcommon{0.2/0.01} & \revcommon{367546/16657} & \revcommon{99800/83644} & \revcommon{369816/\textbf{18083}} \\ % DP time 2495.144

\revcommon{H\&M \sumagg} & \revcommon{Z score} & \revcommon{1/1.5}  & \revcommon{462973/117349} & \revcommon{465770/118930} \\ %DP time 3333.54
\revcommon{H\&M \sumagg} & \revcommon{Iso.Forest} & \revcommon{0.2/0.01} & \revcommon{367546/16657} & \revcommon{369816/\textbf{18083}} \\ % DP time 2495.144

%H\&M \sumagg & per-group & Z score & 1.5 & 117349 & 89214 & 1581 \\ % DP time 18,790 sec
% H\&M \sumagg & per-group & Iso.Forest & 0.01  & 16657  & 83644 & 1426 \\ % DP time 11680 sec

\bottomrule
\end{tabular}
\vskip-1em
\end{table}

\revcommon{
% \red{EDIT this. 
% points: 
% most method+scenario combinations did not yield a repair. 
% Those that did, removed more tuples than necessary. 
% OD sometimes shortens the runtime (the work) of our algorithm by decreasing the number of tuples left to remove. But it's wasteful.
% Outlier detection was least helpful for sum (removed the most tuples and required a lot more removals to repair).}

%with the exception of Z-score with a threshold $\tau\geq2$. But even with $\tau=6$ (values over 6 standard deviations from the mean were considered an outlier), this method still removed 10 times more tuples (369) than the cardinality repair (37). 
%Other outlier removal methods were close to a repair, with the additional removal by DP requiring between 5 and 38 tuples; but the total numbers of tuples removed (including ``outliers'' and the DP algorithm output) were much higher (10-28 times more tuples) than the cardinality repair. 

%The results suggest that not all tuples removed by the DP algorithms align with what standard outlier detection algorithms would identify as outliers. This is supported by the observation that using outlier detection methods as standalone baselines generally did not reveal a repair, except in the one case described above, %(Z-score with a threshold $\tau>2$), 
%where the parameters allowed the algorithm to remove many more tuples %between 10-28 times more 
%than the cardinality repair.
}

%% file: conclusions.tex
\section{Concluding Remarks}\label{sec:conclusions}

\revb{We introduced the notion of an aggregate order dependency (AOD) and investigated the problem of computing a cardinality-based repair for databases that violate an \aod. This framework serves as a principled approach to analyzing violations of expected monotonic trends in data. We analyzed the computational complexity of the repair problem, proposed a general algorithmic framework (\revb{\dpalg}), and provided tailored optimizations for common aggregate functions. We also presented a heuristic alternative (\revb{\greedyalg}) that is often highly effective (that is, fast and close to optimal quality). Our experimental evaluation demonstrates the practical effectiveness of the proposed methods and highlights, through \revb{analysis of the removed tuples,} %case studies,
the utility of our framework in quantifying and interpreting trend violations.}
%\brit{in this paragraph mention the algorithms' names, then you can refer to them in the next paragraph}

As for the next steps, a staggering open challenge lies in handling the \emph{average} aggregate function, which remains computationally demanding and requires further optimization techniques; this difficulty is echoed in related contexts such as Shapley value computation~\cite{abramovich2025advancingfactattributionquery}. \revcommon{%R1.O2
Furthermore, \dpalg has theoretical guarantees but it is computationally heavy, while \greedyalg is usually faster but its quality is not guaranteed. %\shunit{does it make sense to use the algorithm names in the last sentence?} 
To close this gap, it would be helpful to find effective algorithms with nontrivial guaranteed approximation.}

Beyond that, several promising directions remain open for future work. First, our techniques can likely be extended to incorporate constraints on the \emph{magnitude of change} between consecutive groups, such as requiring a minimum absolute or relative difference, thereby capturing stronger forms of monotonicity. Second, a natural generalization is to accommodate \emph{multiple} AODs, enabling the assessment of how closely a dataset conforms to multiple monotonic trends simultaneously. 
Third, we can extend the framework to support a richer class of AODs, such as ones involving aggregation over joins of multiple relations, negation, and so on. \revc{Another important extension is to support intervention models beyond tuple deletion, including updates of cell values and insertion of new tuples.} 

\revc{Finally, as an alternative to the number of removed tuples, we can explore different \emph{measures of intervention}, such as a weighted sum of tuples, assuming that different tuples are associated with different weights (e.g., representing confidence scores). Another measure of intervention is the complexity of the description of removed tuples; in this context,} we are currently exploring a variation of the framework where we delete batches of tuples using tuple templates, thereby seeking populations of deleted tuples that are characterized by easy-to-explain database queries (e.g., people of a certain age group), in the vein of previous work on pattern-based explanations~\cite{DBLP:journals/pvldb/AgmonGYZK24,wu2013scorpion, roy2014formal, roy2015explaining, ibrahim2018}; yet, while previous work has focused on \emph{abductive} explanations (what to keep in the database), our focus is on \emph{contrastive} explanations (what to remove).

%\brit{related work on query result explanation with patterns: \cite{wu2013scorpion, roy2014formal, roy2015explaining, ibrahim2018}}

%\benny{Please go through the Conclusions}
%\ag{Another possible future direction: Support a broader class of queries, including once that include joins, nested queries, negation?}
%\brit{Another possible future work: value updates instead of tuple deletion, restricted variants of deletion, e.g., pattern removal. We can say we have some preliminary theoretical results in this direction (right?)}

%% file: 10-app-proofs.tex
\section{Omitted Proofs}\label{app:proofs}
 Following are proofs of results omitted from the body of the paper.

\thmnphardness*

\begin{proof}
We begin with $\alpha=\esumagg$.
We will show a reduction from the NP-complete \emph{Subset-Sum} problem in the variant of natural numbers. An instance of this problem consists of a set $X=\set{x_1,\dots,x_n}$ of natural numbers, along with a natural number $s$; the goal is to decide whether $X$ has a subset $Y$ with $\sum_{x\in Y}x=s$. Given $X$ and $s$, we construct the relation $r$ that consists of the following tuples over $(G,A)$:
\begin{enumerate}
\item $(i,s)$ for $i=0,\dots,n-1$;
\item $(n,x_j)$ for $j=1,\dots,n$;
\item $(\ell,s)$ for $\ell=n+1,\dots,2n$.
\end{enumerate}
We claim that $(X,s)$ is a ``yes'' instance (i.e., the above $Y$ exists) if and only if $r$ has a monotonic subset (having the aggregate value $s$ for each group) that requires deleting fewer than $n$ tuples: 
\begin{itemize}
    \item If the subset $Y$ exists, then we remove from $r$ all tuples $(n,x_j)$ with $x_j\in X\setminus Y$, and obtain a monotonic subset. 
    \item Conversely, if $r$ has a monotonic subset achieved by deleting fewer than $n$ tuples, then it must leave at least one tuple of each of the above three types. In particular, the tuples $(n,x_j)$ that remain in $r$ must have their $x_j$s summing up to $s$ for the result to satisfy the \aod.
\end{itemize}
In particular, it is NP-hard to find a \crepair deleting a minimum number of tuples from $r$.

For $\alpha=\eavgagg$, we use a similar reduction, with a slightly different construction of $r$, which now consists of the following tuples:
\begin{enumerate}
\item $(i,0)$ for $i=0,\dots,n-1$;
\item $(n,-s)$ and, additionally, $(n,x_j)$ for $j=1,\dots,n$;
\item $(\ell,0)$ for $\ell=n+1,\dots,2n$.
\end{enumerate}
%\jonny{In the problem statement, do all elements have non-negative values? If so, then this example does not fit what we are trying to prove. }
The same reasoning of \sumagg shows that $(X,s)$ is a ``yes'' instance if and only if $r$ has a monotonic subset that requires deleting fewer than $n$ tuples. In this case, the aggregate value of each group is zero.
\end{proof}

\lemmaDmatrixons*
%Before we prove this lemma, we give some intuition as to its meaning. When it is possible to add another $v_j$ element ($D_j[s-v_j] \neq c_j$), the optimal solution for $s$ will either include it ($D_j[s]=D_j[s-v_j]+1$) or not include any $v_j$ elements ($D_j[s] = 0$) \ag{Just move this paragraph to be combined with the text above the lemma.}.

\begin{proof}
%\red{maybe rename OPT to M*?}
Let $\mathsf{OPT}$ denote the optimal solution: $\mathsf{OPT}[s]$ will be the size of the largest subset of $r_i$ with sum $s$.

We first observe that
$\mathsf{OPT}[s] \geq \mathsf{OPT}[s-v_j]+1$, since a solution with sum $s-v_j$ can be completed into a solution for $s$ by adding another $v_j$ tuple to the subset.
We divide into two cases.

Case 1: $\mathsf{OPT}[s] > \mathsf{OPT}[s-v_j]+1$. 
We will now show that $v_j$ items are not used in the optimal solution for $s$.
Assume by way of contradiction that $D[j,s]\neq 0$. We can remove a $v_j$ tuple from the optimal solution for $s$, and get a solution with sum $s-v_j$ with $\mathsf{OPT}[s]-1$ tuples. Then $\mathsf{OPT}[s-v_j]\geq \mathsf{OPT}[s]-1$. However, by the assumption in case 1, this means: $\mathsf{OPT}[s] > \mathsf{OPT}[s-v_j]+1 \geq \mathsf{OPT}[s]$ (contradiction).

Case 2: $\mathsf{OPT}[s] = \mathsf{OPT}[s-v_j]+1$. Meaning that the solution with sum $s$ contains one more tuple than the solution for $s-v_j$. 
Adding a $v_j$ tuple to the optimal solution for $s-v_j$ would yield a solution for $s$ with $D[j,s-v_j]+1$ tuples of value $v_j$. Since the algorithm searches for an optimal solution with a minimal number of instances of the largest value, we get:
\begin{equation}\label{eq:sum_knapsack_1}
D[j,s]\leq D[j,s-v_j]+1
\end{equation}
%$D_j[s]\leq D_j[s-v_j]+1$ (1) \ag{Replace `(1)' with a different clearer notation to refer back to?}. 

Assume by way of contradiction: $D[j,s] \notin \set{0,D[j,s-v_j]+1}$. Then:
\begin{equation}\label{eq:sum_knapsack_2}
0 < D[j,s] < D[j,s-v_j]+1   
\end{equation}
%$0 < D_j[s] < D_j[s-v_j]+1$ (2). 
Then removing a $v_j$ tuple from the optimal solution with sum $s$ would yield a solution with sum $s-v_j$, with $OPT[s]-1$ tuples in total, and with $D_j[s]-1$ instances of $v_j$. 
From \Cref{eq:sum_knapsack_2} we get $D[j,s]-1 < D[j,s-v_j]$, in contradiction to the minimality of $D[j,s-v_j]$.
\end{proof}

%\lemhprune*

% \begin{proof}
% Let $r'$ be a repair of $r$ \ester{again, what is a repair?} w.r.t.~an \aod $G\nearrow \alpha(A)$ that uses $x_1$ for group $i$. 
% Then we have $\alpha(r'_i[A])=x_1$ and $\alpha(r'_j[A])\geq x_1$ for $j>i$. 
% Since $x_1\geq x_2$, $\alpha(r'_j[A])\geq x_2$ for $j>i$. 
% Let $r'_{>i} = r'_i\cup \dots r'_m$. Then $H_{i-1}[x_2] \cup r'_{>i}$ is a repair for the \aod, and
% $|H_{i-1}[x_2]\cup r'_{>i}| \geq |H_{i-1}[x_1] \cup r'_{>i}| = |r'|$.
% \end{proof}

\thmGreedyRatio*

\begin{proof}

For each case we show a database with schema $(G,A)$, where \greedyalg will reach a repair by removing at least $\Omega(n)$ times more tuples than the optimal solution.

\paragraph{\avgagg.}
Consider the following database with schema $(G,A)$, composed of the following tuples: $\set{(1,1),(1,3),(2,1),(2,3),(3,2)}$ and also $n-5$ tuples $(3,1)$. Initially, the averages of the 3 groups are $\eavgagg(r_1) = \eavgagg(r_2)=2$ and $\eavgagg(r_3)=(2+n-5)/n-4$, meaning $1<\eavgagg(r_3)<2$, and the \aod $G\nearrow\eavgagg(A)$ is not satisfied. A possible monotonic subset is achieved by removing the two tuples $(1,3)$ and $(2,3)$, resulting in an average of $1$ for $r_1, r_2$. \greedyalg would not choose to remove these tuples. Removing $(1,3)$ has impact 0 on the $\text{MVI}$, and removing $(2,3)$ has a negative impact. The only tuples with a positive impact are the $n-5$ copies of $(3,1)$, which \greedyalg will remove one by one.

\paragraph{\sumagg.}
Consider the following database with schema $(G,A)$, composed of the following tuples:
$n-5$ copies of $(1,1)$, as well as $(2,2n)$, 2 copies of $(2,n/3)$ and 2 copies of $(3,n)$.
Initially, the sums of the 3 groups are $\esumagg(r_1) = n-5$, $\esumagg(r_2)=8n/3$ and $\esumagg(r_3)=2n$, and the \aod $G\nearrow\esumagg(A)$ is not satisfied since $\esumagg(r_2)>\esumagg(r_3)$.
Removing the two tuples $(2,n/3)$ will yield a monotonic subset, since $\esumagg(r_2)$ will be $2n$. In the first iteration of \greedyalg, tuples from $r_1$ have impact $0$, and tuples from $r_3$ will have a negative impact. Therefore the algorithm will remove a tuple from $r_2$: specifically, $(2, 2n)$ has the maximal impact of $2n/3 - (n-5-2n/3)=5+n/3$ (while the other two tuples in $r_2$ have impact $n/3$). The new sums are $n-5, 2n/3, 2n$, and the \aod is still not satisfied. Next, \greedyalg will remove only tuples from $r_1$, as the others have impact $0$. Since each tuple contributes $1$ to $\esumagg(r_1)$, $n/3-5$ additional tuples will be removed.

\paragraph{\medianagg.}
Consider the following database with schema $(A,G)$, composed of the following tuples:
$(1,1)$, $3$ copies of $(1,n)$, and $(2,1),(2,2),\dots,(2,n-4)$.
Initially, the medians are $\emedianagg(r_1) = n$, $\emedianagg(r_2)\in \set{1,\dots,n-4}$, and the \aod $G\nearrow\emedianagg(A)$ is not satisfied since $\emedianagg(r_1)>\emedianagg(r_2)$.
Removing the three tuples of the form $(1,n)$ yields a monotonic subset, since $\emedianagg(r_1)$ will be $1$.
However, the impact of the $(1,n)$ tuples is $0$, since it does not change $\emedianagg(r_1)$. Instead, \greedyalg will remove all $n-4$ tuples from $r_2$, starting with the lower $A$ values, since each one will increase $\emedianagg(r_2)$ and have a positive impact.
\end{proof}

%% file: 10-app-sum_opt_example.tex
\section{Missing Details from \Cref{sec:optimizations}}
In this section, we give additional details that were omitted from \Cref{sec:optimizations}.

\subsection{Pruning the Intermediate Result Dictionary}\label{sec:dp_pruning}

A notable drawback of the \dpalg outlined in \Cref{alg:dp} is that the size of the intermediate result dictionary, $H_i$, can grow to be very large. We next describe two methods of pruning that dictionary, followed by an evaluation of their effects on the size of the dictionary and on the runtime.

\subsubsection{Pruning the Intermediate Result Dictionary using Dominated Values}\label{sec:dp_pruning_dom}
Let $F_i = \set{x \mid O_{\alpha}(r_i,x)\neq \emptyset}$ denote the set of feasible aggregate values of $r_i$. Each feasible value $x\in F_i$ is added to $H_i$ in line~8 of \Cref{alg:dp}. The size of $F_i$ may be exponential in the size of $r_i$ since, potentially, every subset of $r_i$ may lead to a different aggregate value. However, after computing $H_i$ (at the end of the loop of line~3), we can prune $H_i$, so that the number of keys saved in this dictionary is linear in the size of the dataset.

We say that a monotonic subset $r'\subseteq r$ %\ester{what is the meaning of repair here? Is it just a subset that satisfies the AOD? I don't think that we have a definition of repair in the paper, only that of a cardinality repair.} \shunit{yes, a subset that satisfies the AOD. we used to have such a definition. Should it make a come back?}
\emph{uses} a value $x$ for group $i$ %an entry $x$ of $H_i$ 
if it holds that $\alpha(r'_i\dbr{A})=x$.
An entry $x_1$ of $H_i$ is \emph{dominated} by entry $x_2$ of $H_i$ if $x_1 > x_2$ and for every monotonic subset $r'$ that uses $x_1$ for $r_i$, there exists another monotonic subset $r''$ such that $|r''|\geq|r'|$, 
%\ag{The word `dominated' suggests that there is a total order, i.e., $|r''| > |r'|$. This terminology makes it a bit unclear. In other words, if $|r''| = |r'|$ can it be that $x_1$ and $x_2$ dominate each other?}\jonny{I am not sure 'dominate' implies a total order (see game theory, where strategies can dominate each other and the ordering is partial)},
$r''$ uses $x_2$ for $r_i$, %\ester{for the same group $r_i$ or can it be a different group?}, 
and $\forall j\neq i: \alpha(r'_j\dbr{A})=\alpha(r'_j\dbr{A})$. That is, it is possible to use $x_2$ instead of $x_1$ as the aggregate value for group $i$, without compromising the optimality of the solution. 
%For each entry $x$ of $H_i$, if any repair that uses it can use another entry $x'$ of $H_i$ instead, and have the same number of tuples or more - then we can prune $x$ from $H_i$.
%\ag{It's worth recalling here that $H[x]$ is a set.}
Recall that $H_i[x]$ holds the largest subset $r'$ of $r_1\cup \dots \cup r_i$ such that $\alpha(r'_i\dbr{A})=x$.

\begin{restatable}{lemma}{lemhprune}
\label{lem:hprune}
Let $x_1$ and $x_2$ be two values in $H_i$ such that $x_1 > x_2$ and $|H_{i-1}[x_1]| \leq |H_{i-1}[x_2]|$. Then, $x_1$ is dominated by $x_2$.%\jonny{in the last paragraph $x_1$ is dominated by $x_2$ and in the lemma it is flipped.}
\end{restatable}

\begin{proof}
Let $r'$ be a monotonic subset of $r$ %\ester{again, what is a repair?} 
w.r.t.~an \aod $G\nearrow \alpha(A)$, that uses $x_1$ for group $i$. 
Then we have $\alpha(r'_i[A])=x_1$ and $\alpha(r'_j[A])\geq x_1$ for $j>i$. 
Since $x_1 > x_2$, $\alpha(r'_j[A])\geq x_2$ for $j>i$. 
Let $r'_{>i} = r'_i\cup \dots r'_m$. Then $H_{i-1}[x_2] \cup r'_{>i}$ is a monotonic subset for the \aod, and
$|H_{i-1}[x_2]\cup r'_{>i}| \geq |H_{i-1}[x_1] \cup r'_{>i}| = |r'|$.
\end{proof}

Given \Cref{lem:hprune}, we can prune $H_i$ in the following manner: we iterate over the keys $x$ in ascending order and keep track of the largest solution size $k$ seen so far. When considering a value of $x$, keep it only if $|H[x]|>k$.
Since the solution sizes $|H_i[x]|$ must be strictly increasing, this means that the size of $H$ after this pruning is bounded by the size of the dataset.
%\ag{It would be nice to demonstrate the pruning process if you can do it succinctly.}
Consider $r_1=\set{t\in r | t.\mbox{edu} = 1}$ from the example in \Cref{tab:DB_example}, with $r_1\dbr{A}=[1,2]$. Then $|H_1[1]| = 1, |H_1[2]| = 1, |H_1[3]| = 2$. In this case, 1 dominates 2, and we prune 2 from $H_1$.

\subsubsection{Pruning the Intermediate Result Dictionary using a Heuristic Upper Bound}\label{sec:dp_pruning_greedy}
%\shunit{move this to before greedy and use a general heuristic algorithm}
To accelerate the search, we can use a lower %\brit{upper? The heuristic gives a solution of size k, so the size of the optimal solution is at most k}\shunit{No, lower. note the size of the repair is the number of tuples we keep.} 
bound on the size of the \crepair (i.e., an upper bound on the number of removed tuples). Such a bound can come from a heuristic %an inexact \ag{Replace with `heuristic'} 
algorithm that finds a monotonic subset, and we show an example of such an algorithm (\greedyalg) in \Cref{sec:greedy_tuple_del}. 
%Let $m$ denote the number of tuples removed in the given (non-optimal) repair. 
We first run \greedyalg and reach a (possibly non-maximal) monotonic subset, that requires removing $h$ tuples. The number of tuple removals for the \crepair is at most $h$.
%\benny{The greedy algorithm will come only later. We should discuss about a general inexact algorithm that gives a one-sided error (more than the minimum). }
Therefore, we prune the search so as not to consider partial solutions that remove more tuples than \greedyalg. %given solution \ag{What is the `given solution'? The output of the heuristic algorithm? If so, do you need to compute it first before running the DP algo.? This needs to be clarified.}.

In the intermediate result dictionary $H$, we only keep solutions that remove up to $m$ tuples from the entire relation $r$. Specifically, in line~8 of \Cref{alg:dp}, if $|H_{i-1}[x]\cup F_i[x]|\leq |r|-m$, then $H_{i-1}[x]\cup F_i[x]$ cannot be completed into an optimal solution, and therefore we can avoid saving $x$ in $H_i$.

\subsubsection{Experimental Evaluation}\label{sec:exp_dp_pruning}
We next compare the runtimes of variations of \dpalg, over various aggregation functions (\sumagg, \medianagg, \avgagg, \maxagg).
The compared variations are: (1) Naive DP - the naive version described in \Cref{alg:dp}. (2) Pruning the intermediate result dictionary by dominating values (\Cref{sec:dp_pruning_dom}). (3) Pruning the intermediate result dictionary by a heuristic bound (\Cref{sec:dp_pruning_greedy}). (4) Both pruning methods combined. For each aggregation function we use the setting that achieved the fastest runtimes according to \Cref{sec:exp_optimizations}: for \procagg{sum}, we apply the knapsack optimization (\Cref{sec:opt_knapsack_aggpack}). For \procagg{median} and \procagg{avg}, we used the combination of \procaggalpha pruning by a heuristic bound and optimizing \procaggalpha with value-based iteration. 

\begin{figure}[tbp]
    \centering
    \subfloat[\sumagg.\label{fig:agnostic_opt_sum}]{
        \includegraphics[width=0.48\linewidth]{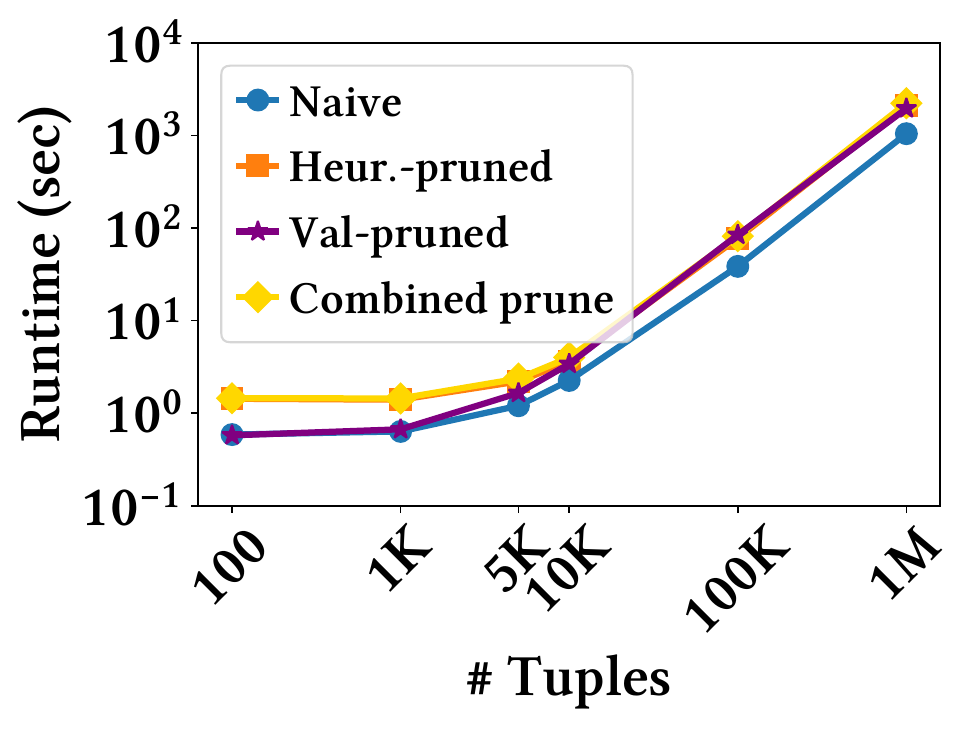}
    }
    \hfill
    \subfloat[\medianagg.\label{fig:agnostic_opt_median}]{
        \includegraphics[width=0.48\linewidth]{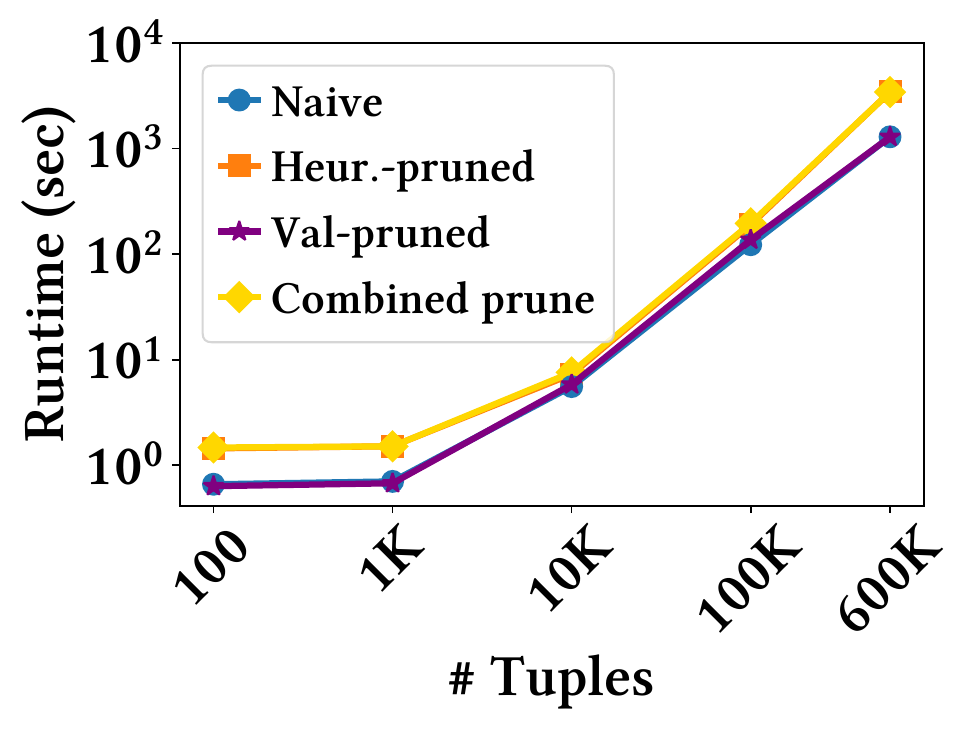}
    }
    \\
    % \hfill
    \subfloat[\avgagg.\label{fig:agnostic_opt_avg}]{
        \includegraphics[width=0.48\linewidth]{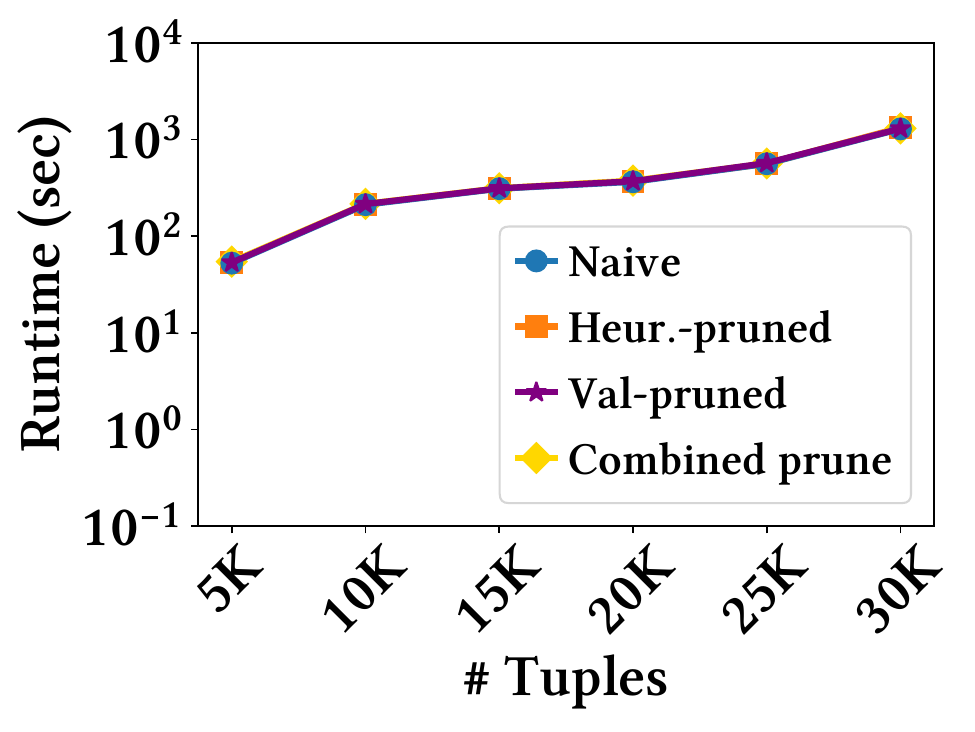}
    }
    \hfill
    \subfloat[\maxagg.\label{fig:agnostic_opt_max}]{
        \includegraphics[width=0.48\linewidth]{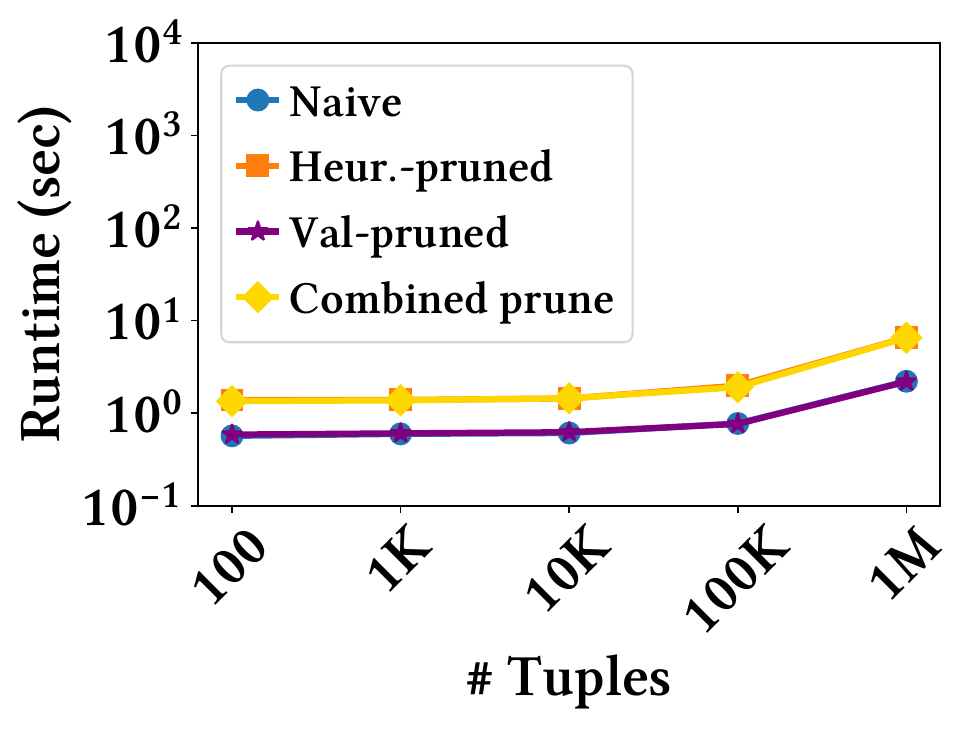}
    }
    \caption{Comparison of runtime performance of different DP aggregation-agnostic optimizations, over \sumagg, \medianagg, \avgagg, and \maxagg across increasing number of tuples, using samples from real datasets.}
    \label{fig:agnostic_opt_runtime}
\end{figure}

\begin{figure}[b]
    \centering
    \subfloat[SO \avgagg.\label{fig:hkeys_prune_so_avg}]{
        \includegraphics[width=0.47\linewidth]{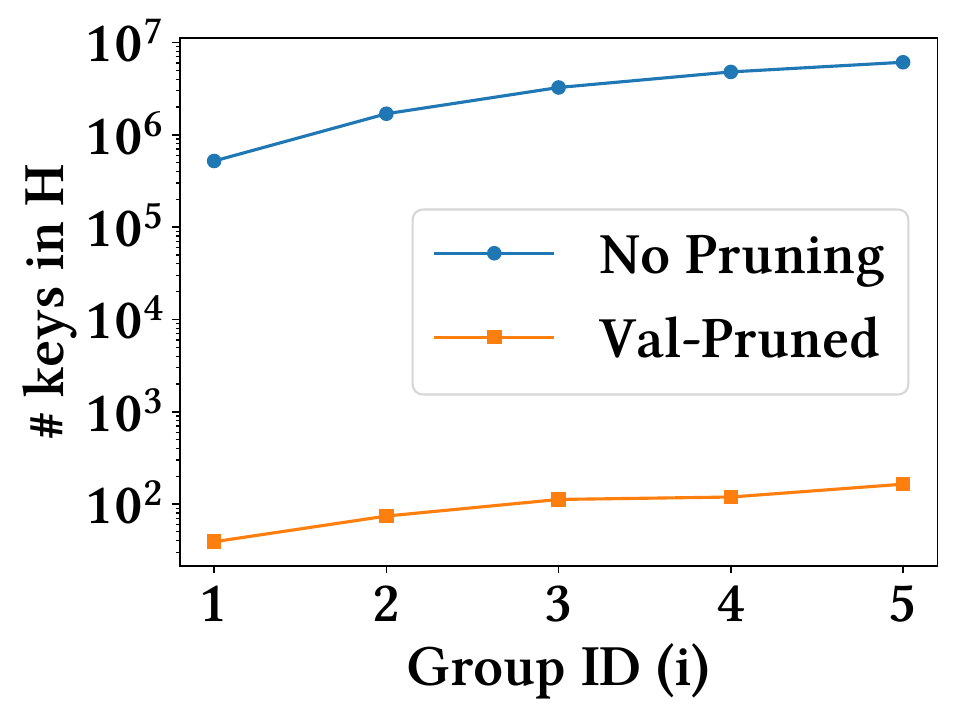}
    }
    \hfill
    \subfloat[H\&M \sumagg.\label{fig:hkeys_prune_hm_sum}]{
        \includegraphics[width=0.47\linewidth]{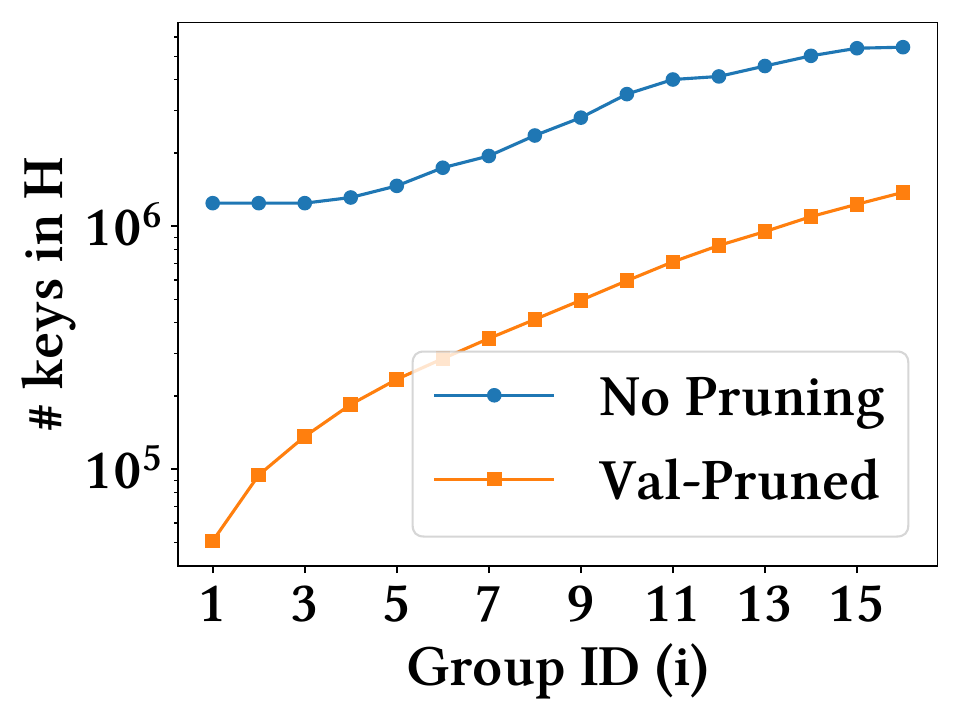}
    }
    \caption{Number of keys stored in the intermediate result dictionary $H$ with and without pruning.}
    \label{fig:hkeys_prune}
\end{figure}

From \Cref{fig:agnostic_opt_runtime}, we can see that both methods of pruning the intermediate result dictionary caused an increase in runtimes. This is due to the overhead of this optimization: pruning requires an additional iteration over the keys in the dictionary. However, these optimizations are useful for reducing memory overhead. \Cref{fig:hkeys_prune} shows the number of keys stored in $H$ over the number of processed groups, for \avgagg and \sumagg, the two harder cases of the problem. For \avgagg, pruning dominated values reduced the number of stored keys by 4 orders of magnitude (6,112,712 versus 164 for the last group). % 6112712 to 164 for group 5.
For \sumagg, pruning dominated values resulted in 4 times less stored keys for the last group (1.4M vs 5.4M). %5,449,245 vs 1,376,805
The difference in the size of improvement is related to the number of possible distinct aggregation values of subsets of $r\dbr{A}$. For \avgagg, every different sum and size of a subset possibly lead to a new average, while for \sumagg, the size of the subset does not matter.

We conclude that the pruning optimizations for the intermediate result dictionary substantially reduce $H$, their current overhead diminishes a runtime improvement. An interesting future work direction is to develop the algorithm further to reduce the pruning overhead.

\subsection{Pruning Holistic Packing using a Heuristic Bound}\label{sec:heuristic_pruning}

\subsubsection{\sumagg}
%how to change the inner DP for sum to support removing subsets instead of keeping subsets
Recall that the dynamic program in \Cref{sec:aggpack_instances_sum_avg} used a data structure $M[j,s]$, that holds the maximal subset of the first $j$ tuples $t_1,\dots,t_j$ with sum $s$ of $A$ values. To be able to prune the computation based on a bound $\bound$ on the number of tuples to remove, we first replace the structure $M$ with a similar structure $P[j,s]$ with a different meaning. $P[j,s]$ holds the minimal subset $r'$ of $t_1,\dots,t_j$ that when removed from $r$, the sum of $A$ values is $s$. I.e., $\esumagg((r\setminus r')\dbr{A}) = s$.

Let $T=\esumagg(r\dbr{A})$.
To populate the table $P[j,s]$, we start by setting $P[j,s]=0$ for $j=0,s=T$ and $\infty$ otherwise (when removing $0$ tuples, the sum of the remaining tuples is $T$).

For $j>0$ and any $s\in V_{\alpha}(r)$, if $s+t_j[A] \in V_{\alpha}(r)$, 
$$P[j,s]=\min\left(P[j-1,s],1+P[j-1,s+t_j[A]]\right)\,.$$  
If $s-t_j[A]\notin V_\alpha(r)$, then we set $P[j,s]= P[j-1,s]$.

To prune the computation based on a bound $\bound$ on the number of tuples to remove, we add another condition to the update of $P[j,s]$: if $1+P[j-1,s+t_j[A]] > \bound$ (there are already too many tuples in the removed subset), we set $P[j,s]=P[j-1,s]$.

\subsubsection{\avgagg}
As for \sumagg, to be able to prune the search for a maximum size subset with an $\eavgagg$ value of $x$, we modify the dynamic programming table described in \Cref{sec:aggpack_instances_sum_avg} %\Cref{alg:aggpack_avg_opt} 
as follows. Recall that the data structure $M'[j,\ell,s]$ was used to hold a boolean value, being true when $\set{t_1,\dots, t_j}$ has a subset with precisely $\ell$ tuples, where the $A$ values sum up to $s$. 

The modified table $P'[j,\ell,s]$ holds true when $\set{t_1,\dots, t_j}$ has a subset $r'$ of $\ell$ tuples such that $\esumagg((r\setminus r')\dbr{A}) = s$.

To populate $P'$, we initialize $P'[j,\ell,s]$ to true for $\ell=0, s=T$, and every $j$, and false otherwise.

For $j>0$ and $\ell>0$, we check whether $s+t_j[A]\in V_\alpha(r)$; if so:
\[
P'[j,\ell,s]= 
P'[j-1,\ell,s] \lor
P'[j-1,\ell-1,s+t_j[A]]\,,
\]
otherwise, if $s+t_j[A]\notin V_\alpha(r)$, then:
\[
P'[j,\ell,s]=
P'[j-1,\ell,s]\,.
\]

Again, as for \sumagg, if a bound $\bound$ on the number of tuples to remove is given, we can limit the update only to cells of $P'$ that have $\ell\leq \bound$.

%\shunit{how to change the inner DP for avg to support removing subsets instead of keeping subsets}

\subsection{Holistic Packing for Median: Pseudo-code}
\label{sec:app:holistic-median}

\Cref{alg:aggpack_median_opt_app} gives the pseudo-code of the algorithm of holistic packing for \medianagg.

\begin{algorithm}[t]
\small
\DontPrintSemicolon
\SetKwInOut{Input}{Input}\SetKwInOut{Output}{Output}
\LinesNumbered
\newcommand\mycommfont[1]{\footnotesize\ttfamily\textcolor{blue}{#1}}
\SetCommentSty{mycommfont}
% \Input{\emph{(1)} A sorted array $T$ of $r_i[A]$, \emph{(2)} a histogram $\mathsf{Hist}$ of ordered unique $A$ values $v_j$, where $\mathsf{Hist}[v_j]$ is the number of appearances of $v_j$ in the input $r_i$, and \emph{(3)} an upper bound $\bound$ on the number of tuples to remove.}
\Input{A bag of values $r_i\dbr{A}$, and an upper bound $\bound$ on the number of tuples to remove.}
\Output{A mapping $M$ from a $\mathsf{median}$ value $x$ to maximum subset size, and an array $D$ mapping $x$ to the $A$ values composing it and the number of tuples on each side to include.}

$T \gets $ a sorted array of $r_i\dbr{A}$ values \;
$\mathsf{Hist} \gets$ a histogram of $T$, in ascending order by values $v_j$ \;
$M \gets$ a mapping from each \medianagg $x$ to $-1$ \; %\benny{map is what we call mapping or function? Why change? We have function, mapping, dictionary, and map so far} \;
$D \gets$ a mapping from $\mathsf{median}$ to (empty list, 0) \;
$n \gets |T|$ \;

\tcp{Define the range of pivot candidates with pruning.}
%$\mathsf{start} \gets 0$, $\mathsf{end} \gets n$ \;
$\mathsf{start} \gets \max(0, \lfloor n/2 - \bound/2 - 1 \rfloor)$ \;
$\mathsf{end} \gets \min(n, \lfloor n/2 + \bound/2 + 1 \rfloor)$ \;

\tcp{Single pivot medians.}
\For{$i \gets \mathsf{start}$ \KwTo $\mathsf{end} - 1$}{
    $k \gets \min(i, n - i - 1)$ \;
    $\mathsf{count} \gets 2 \cdot k + 1$ \;
    \If{$M[T[i]] < \mathsf{count}$}{
        $M[T[i]] \gets \mathsf{count}$ \;
        $D[T[i]] \gets ([T[i]], k)$ \;
    }
}

\tcp{Two adjacent pivot medians.}
\For{$i \gets \mathsf{start}$ \KwTo $\min(\mathsf{end}, n - 1) - 1$}{
    $k \gets \min(i, n - i - 2)$ \;
    $\mathsf{count} \gets 2 \cdot k + 2$ \;
    $x \gets (T[i] + T[i+1]) / 2$ \;
    \If{$M[x] < \mathsf{count}$}{
        $M[x] \gets \mathsf{count}$ \;
        $D[x] \gets ([T[i], T[i+1]], k)$ \;
    }
}

\tcp{Two nonadjacent pivot medians.}
$k_{\text{left}} \gets 0$ \;
\For{$i \gets 0$ \KwTo $|\mathsf{Hist}| - 1$}{
    $k_{\text{left}} \gets k_{\text{left}} + \mathsf{Hist}[i]$ \;
    $k_{\text{right}} \gets n - k_{\text{left}}$ \;

    \tcp{Pruning - too far to the left}
    \If{$2 \cdot k_{\text{left}} \leq n - \bound$}{
        \textbf{continue} \;
    }
    \For{$j \gets i + 1$ \KwTo $|\mathsf{Hist}| - 1$}{
        \tcp{Pruning - too far to the right}
        \If{$2\cdot k_{\text{right}}\leq n-\bound$}{
            \textbf{continue} \;
        }

        $x \gets (\mathsf{Hist}[i] + \mathsf{Hist}[j]) / 2$ \;
        $\mathsf{count} \gets 2 \cdot \min(k_{\text{left}}, k_{\text{right}})$ \;
        \If{$M[x] < \mathsf{count}$}{
            $M[x] \gets \mathsf{count}$ \;
            $D[x] \gets ([\mathsf{Hist}[i], \mathsf{Hist}[j]], \min(k_{\text{left}}, k_{\text{right}} - 1)$ \; 
        }

        $k_{\text{right}} \gets k_{\text{right}} - \mathsf{Hist}[j]$ \;
    }
}

\Return $(M, D)$ \;

\caption{\procagg{median}} %\benny{Did you see the line breaking of this code?}
\label{alg:aggpack_median_opt_app}
\end{algorithm}

\subsection{Example for Holistic Sum Packing} \label{sec:sum_opt_example_full}

\begin{figure*}[t]
\definecolor{mygreen}{RGB}{209, 238, 205}
\definecolor{myblue}{RGB}{203, 233, 247}
\definecolor{mypeach}{RGB}{251, 222, 201}

\begin{center}
\renewcommand{\arraystretch}{1.3}
%\scriptsize
\begin{tabular}{|c|*{21}{c|}}
\hline
$s$ & \textbf{0} & \textbf{1} & \textbf{2} & \textbf{3} & \textbf{4} & \textbf{5} & \textbf{6} & \textbf{7} & \textbf{8} & \textbf{9} & \textbf{10} & \textbf{11} & \textbf{12} & \textbf{13} & \textbf{14} & \textbf{15} & \textbf{16} & \textbf{17} & \textbf{18} & \textbf{19} & \textbf{20} \\

\hline\hline
$M$ &
\cellcolor{mygreen}0 & -1 & \cellcolor{mygreen}1 & -1 & \cellcolor{mygreen}2 &
\cellcolor{mypeach}1 & -1 & \cellcolor{mypeach}2 & -1 &
\cellcolor{mypeach}3 & \cellcolor{mypeach}2 & -1 &
\cellcolor{mypeach}3 & -1 &
\cellcolor{mypeach}4 & -1 &
-1 & -1 & -1 & -1 & -1 \\
\hline
% \begin{tabular}[c]{@{}c@{}}temp $M$\end{tabular} &
% \cellcolor{mygreen}0 & -1 & \cellcolor{mygreen}1 & -1 & \cellcolor{mygreen}2 &
% \cellcolor{mypeach}1 & \cellcolor{myblue}1 & \cellcolor{mypeach}2 & \cellcolor{myblue}2 &
% \cellcolor{mypeach}3 & \cellcolor{myblue}2 & \cellcolor{myblue}2 &
% \cellcolor{mypeach}3 & \cellcolor{myblue}3 &
% \cellcolor{mypeach}4 & \cellcolor{myblue}4 &
% \cellcolor{myblue}3 & -1 &
% \cellcolor{myblue}4 & -1 & \cellcolor{myblue}5 \\
\begin{tabular}[c]{@{}c@{}}temp $M$\end{tabular} &
0 & -1 & 1 & -1 & 2 &
1 & \cellcolor{myblue}1 & 2 & \cellcolor{myblue}2 &
3 & \cellcolor{myblue}\st{2} 3 & \cellcolor{myblue}2 &
3 & \cellcolor{myblue}3 &
4 & \cellcolor{myblue}4 &
\cellcolor{myblue}3 & -1 &
\cellcolor{myblue}4 & -1 & \cellcolor{myblue}5 \\
\hline \hline
\multirow{1}{*}{$D[1]:v_1=2$} &  &  & \cellcolor{mygreen}1 &  & \cellcolor{mygreen}2 &  &  &  &  &  &  &  &  &  &  &  &  &  &  &  &  \\
\hline
\multirow{1}{*}{$D[2]: v_2=5$} &  &  &  &  &  & \cellcolor{mypeach}1 &  & \cellcolor{mypeach}1 &  & \cellcolor{mypeach}1 &   \cellcolor{mypeach}2 &  & \cellcolor{mypeach}2 &  & \cellcolor{mypeach}2 &  &  &  &  &  &\\
\hline
\multirow{1}{*}{$D[3]: v_3=6$} &  &  &  &  &  &  & \cellcolor{myblue}1 &  & \cellcolor{myblue}1 &  & \cellcolor{myblue}1 &  \cellcolor{myblue}1 &  & \cellcolor{myblue}1 &  & \cellcolor{myblue}1 &   \cellcolor{myblue}1 &  & \cellcolor{myblue}1 &  & \cellcolor{myblue}1\\
\hline
\end{tabular}
\end{center}

 \caption{Example for \Cref{alg:aggpack_sum_opt}. Empty cells in $D_j$ contain the default value $0$.}\label{fig:sum_opt_example}
\end{figure*}

%\begin{example}
Consider the group $r_2$ defined by $\mbox{edu}{=}2$ from \Cref{tab:DB_example}. The values of the aggregate attribute are $2,2,5,5,6$, creating the input histogram $\mathsf{Hist}=[(2,2), (5,2), (6,1)]$. The keys are sorted in ascending $\mbox{edu}$ order, so $v_1=2, v_2=5, v_3=6$. The sum of all $r_2[A]$ values is $T=20$; hence, the size of $M, M'$, and each row of $D$ is $21$. The cell $M[0]$ is initialized to $0$, and all other cells to $-1$.

\Cref{fig:sum_opt_example} shows the point in the algorithm where we have already performed the first two iterations of the loop in line~5 of \Cref{alg:aggpack_sum_opt}. That is, we have processed $v_1=2$ (green cells) and $v_2=5$ (orange cells). 
For example, using two instances of $v_2=5$ we reached $s=10$, so $M[10]=2$ and $D[2,10]=2$. Using one instance of $5$ and two instances of $2$ we have reached $s=9$, so $M[9]=3$, $D[2,9]=1$, and $D[1,9-5]$ has already been set to $2$ in the first iteration.
At this point, $Z=2\cdot 2+5\cdot 2=14$.
The last pair in the histogram is $(v_j,c_j)=(v_3,c_3)=(6,1)$. We increase $Z$ to $20$ (line 6). Then, $M$ is copied into $M'$, which is shown in \Cref{fig:sum_opt_example} in the non-colored cells of $M'$. 

We iterate over $s$ values from $v_j=6$ to $Z=20$ (line~8). For example, consider $s=10$. In the previous iteration, we had $M[10]=2$ using two instances of $5$, and no instances of $6$, so $u=D[3, 10-6]=0$. We now consider $M[s_{\text{prev}}]=M[s-(u+1)\cdot v_j]=M[10-6]=M[4]=2$, update $M'[10]=2+0+1$ (line~18), and $D[3,10]=1$ (line~19). 

As another example, consider $s=16$. We have $u=D[j,s-v_j]=D[3,10]=1$ since the instance of $6$ was already used to construct a maximum-sized subset with sum 10. The condition in line 12 holds, and we cannot use $u+1$ instances of $6$.
Therefore we iterate over possible usages of 6 (combined with the previous values) for a way to construct $16$. In line 15 we have $M[s-u'\cdot v_j]=M[16-6]=M[10]=2$. Note that $M[10]$ still holds the solution made up of two instances of 5 (the best solution without using the current $v_j$), since the $M$ data structure only changes in the end of the iteration on $v_j$. In lines 16-17, we update $M'[16]=2+1=3$ and $D[3,16]=1$.
%\qed
%\end{example}

%% file: 10-app-greedy_optimizations.tex
\section{Optimizations of the Heuristic Algorithm}\label{sec:greedy_optimize}

We next describe optimizations to \greedyalg, the heuristic algorithm described in \Cref{sec:greedy_tuple_del}. As explained before, for every aggregation function $\alpha$, the impact of a tuple $t$ (computed in lines~6-8 of \Cref{alg:greedy}) can be calculated by running the query 
``\textsf{SELECT $\alpha($A$)$ GROUP BY G}'' twice: once with $t$ and once without it. Then, $S_{\textsf{MVI}}$ is calculated for each of the queries, and the impact is the difference between them.
We next describe several optimizations of this process. We begin with general optimizations that apply to all aggregations functions $\alpha$, and then describe aggregation-specific optimizations.

%\subsubsection{General Optimizations}

%We can make this process more efficient using some observations (described next) and the properties of specific aggregation functions.

%\ag{Explain the motivation for this section, why we need to consider each aggregate function separately, and refer the reader to the relevant line of ComputeImpact.}

\paragraph{Iteration on a reduced set of tuples}
We optimize the algorithm by calculating tuple impacts only for those within violating groups, i.e., any group $r_i$ where $\textsf{MVI}(r_i, r_{i+1})>0$ or $\textsf{MVI}(r_{i-1}, r_{i})>0$. %\benny{I don't understand the grammar of this sentence; after semicolon there should be a legal sentence.}
Note that this optimization may alter the algorithm's outcome. While the original algorithm could select for removal a zero-impact tuple from anywhere in the relation (if zero is the highest impact, as shown in Example~\ref{example:non-pos-impact}), the optimization will only consider such a tuple if it belongs to a violating group. This heuristic logically focuses the repair effort on tuples that are active participants in a violation.
%\red{First, we note that removing a tuple $t\in r_i$ can only change the aggregation value of $r_i$, and therefore can only have an effect on $\textsf{MVI}(r_i, r_{i+1})$ and $\textsf{MVI}(r_{i-1}, r_{i})$. An immediate optimization is to replace the iteration over all tuples (\Cref{alg:greedy}, line~5) with an iteration only on the tuples of violating groups, i.e., $r_i$ such that $\textsf{MVI}(r_i, r_{i+1}) > 0$ or $\textsf{MVI}(r_{i-1}, r_{i})>0$.
%}

We can further optimize the process for \maxagg, \countagg, \countdagg and \sumagg. These functions are \emph{monotonic} in the sense that a tuple removal from group $r_i$ can only reduce (or preserve) the value $\alpha(r_i\dbr{A})$. Therefore, if $\textsf{MVI}(r_i, r_{i+1}) > 0$, every removal of a tuple from $r_{i+1}$ would only increase $\textsf{MVI}(r_i, r_{i+1})$. Hence, when selecting a tuple to remove, we iterate over (and compute the impact for) tuples in $r_i$, not those in $r_{i+1}$. In other words, $r_i$ is considered a violating group only if $\textsf{MVI}(r_i, r_{i+1})>0$.
%there is a violation caused by groups $r_i$ and $r_{i+1}$, 
%we can iterate only on $r_i$ tuples (excluding $r_{i+1}$ tuples).
%, and line 5 in \Cref{alg:greedy} can be changed to:$V\gets \set{r_i | \mathsf{MVI}[i]>0}$.}

% \paragraph{Avoiding repeating computations}
% \red{Storing the aggregate values $\alpha(r_i\dbr{A})$ for all groups $r_i$ can accelerate the calculation of each tuple's impact. Since removing a tuple $t\in r_i$ only changes the value of $\alpha(r_i\dbr{A})$, we can use the stored aggregate values of all other groups to calculate the impact, avoiding the need to recalculate these values. That is, to calculate the impact of a tuple $t\in r_i$, we only need to calculate the values $\mathsf{MVI}(r_{i-1}, r_i\setminus\{t\})$ and $\mathsf{MVI}(r_i\setminus\{t\}, r_{i+1})$, as the rest of the values needed for the computation of the sum of violations are already stored.
% \benny{Isn't it too obvious to be described in a full paragraph? I suggest deleting it altogether.}
% }

%\subsubsection{Aggregation-Specific Optimizations}
The next optimizations are based on specific properties of each aggregation function.
\paragraph{Optimization for \maxagg and \minagg}
%To avoid unnecessary recalculations, we maintain a dictionary of the maximum value $\alpha(r_{i}[A])$ for each group $r_{i}$.

%First, we note that for $\maxagg$,

A key observation for \maxagg
%\benny{Wait - how did $\agg{max}$ become $\agg{Max}$? Why did we change the notation to uppercase, as opposed to, say, Theorem 4.1? Please replace all of the $\agg{???}$ with macros so that we are consistent with the casing.}
is that when there are multiple tuples with the maximum $A$ value in a group $r_i$, removing any single one of them will have no impact on $\alpha(r_{i}\dbr{A})$ (and consequently on the sum of violations). Similarly, removing any other tuple from $r_i$ will have zero impact on $\alpha(r_{i}\dbr{A})$. The only operation that reduces the value of $\alpha(r_{i}\dbr{A})$ is the removal of all the maximum-value tuples together.
Therefore, we optimize the algorithm by computing the impact of removing all maximum-value tuples in each violating group together, instead of removing tuples one at a time. 
A similar observation applies to \minagg.%\brit{Similarly, we use this observation for MIN as well?} \shunit{added a sentence about that, but actually we don't have an implementation for min. Not that it's hard, it was just less interesting after we had max, so we never implemented it.} \brit{that's fine. In the experiment part when you talk about the implementation we can say in a sentence that we haven't implemented MIN, but it is also ok to omit that.  }

\paragraph{Optimization for \countagg and \countdagg} %\benny{Is this the same as \countdagg in Theorem 4.1???}
For \countagg, all tuples of a given group have the same impact. If $r_i$ is a violating group, any tuple removed from $r_i$ will reduce $\alpha(r_i\dbr{A}])$ by one; therefore, all tuples from $r_i$ have the same impact on the sum of violations, and it is sufficient to calculate the impact of one of these tuples.

For \countdagg, similarly to \maxagg, removing any proper subset of tuples with the same $A$ value from $r_i$ has no impact on $\ecountdagg(r_i\dbr{A})$. Therefore, instead of removing a single tuple at each iteration, we remove all the tuples of $r_i$ that share the least common $A$ value in this group together.

\paragraph{Optimizations for \sumagg and \avgagg}
In the case of \sumagg, to compute the impact of a tuple $t\in r_i$, we compute the modified aggregate value based on the stored value $\esumagg(r_i\dbr{A})$. That is, $\esumagg((r_i\setminus\{t\})\dbr{A}) = \esumagg(r_i\dbr{A})-t[A]$. %The sum of each group $r_i$ is computed once and stored to avoid recalculations.
% We then compute the only violations that may change as a result of a change in $r_i$: $\mathsf{MVI}(r_{i-1}, r_i\setminus{t})$ and $\mathsf{MVI}(r_i\setminus{t}, r_{i+1})$.

For \avgagg,
we track the sum and the count of tuples for each group. Then instead of computing $\eavgagg((r_i \setminus \set{t})\dbr{A})$ from scratch, we can compute it as $(\esumagg(r_i\dbr{A})-t[A])/(\ecountagg(r_i\dbr{A})-1)$.
%These two quantities are updated easily upon tuple removal and used to compute the group average efficiently.\benny{Something is missing here. Are you saying that you do not need to compute $\alpha(r_i)$ from scratch but only update it? Please say it more explicitly.}

%For \avgagg, if $r_i$ is a violating group (a group where the MVI is greater than zero), the algorithm also evaluates tuples in the next group ($r_{i+1}$). This is done by changing line 5 in \Cref{alg:greedy} to $V\gets \set{r_i | \mathsf{MVI}[i] > 0 \vee \mathsf{MVI}[i-1] > 0}$). This allows the algorithm to consider decreasing $\alpha(r_i[A])$ or increasing $\alpha(r_{i+1}[A])$, which are both ways to reduce the violation for $r_i$.  

\paragraph{Optimization for \medianagg}
We maintain the count of the tuples in each group $r_i$, and a sorted list $L$ of $r_i\dbr{A}$ values and their corresponding indices. The list is sorted by $A$ values in ascending order. 
Instead of removing each tuple and calculating $\emedianagg((r_i \setminus \set{t_j})\dbr{A})$, we can compute it in three batches: for all tuples such that $j<m$, all tuples where $j>m$, and the median tuple $j=m$. For each of these batches, every tuple removed will result in the same median.
%\red{we can compute only up to 3 values.} \shunit{TODO rephrase: (use the term "batch"} we notice that the new median is the same for all tuples $t_j$ where $j<m$ (where $m$ is the index of the median). The new median is also the same value for all tuples such that $j>m$. The new medians are as follows.
%For each tuple considered for removal, the new median is calculated as follows.    
%Let $m$ denote the median index, $j$ denote the index of the removed tuple, and $\mathrm{SL}$ denote the sorted list.
If the count of tuples after removal ($|L|-1$) is odd, the new median is given by:
        \[
        \emedianagg(r_i\setminus \set{t_j}) = 
        \begin{cases} 
        \mathrm{L}[m] & \text{if } j < m, \\
        \mathrm{L}[m-1] & \text{if } j \geq m
        %\mathrm{SL}[m-1] & \text{if } j > m.
        \end{cases}
        \]        
If $|L|-1$ is even, the new median is the average of two values:
        \[
        \emedianagg(r_i\setminus \set{t_j}) = 
        \begin{cases} 
        (\mathrm{L}[m] + \mathrm{L}[m+1])/2 & \text{if } j < m, \\
        (\mathrm{L}[m-1] + \mathrm{L}[m+1])/2 & \text{if } j = m, \\
        (\mathrm{L}[m-1] + \mathrm{L}[m])/2 & \text{if } j > m.
        \end{cases}
        \]

%% file: 10-app-outlier_removal.tex
%\usepackage[margin=1in]{geometry}
\section{Omitted Results from Comparison to Outlier Detection Methods}\label{sec:outlier_exp_full}

We bring here the full set of results from three combinations of dataset and \aod. For each such combination, we performed hyper parameter search for the three outlier removal methods described in \Cref{sec:outlier_exp}. 
%We next aim to answer our 5th research question from the beginning of \Cref{sec:experiments}: can outlier removal methods be used to find a repair for an \aod?
%We consider the SO dataset with the \aod $\mbox{edu}\nearrow\eavgagg(\mbox{salary})$. 

For each method, we performed hyper parameter tuning on the parameters described next.
%We experimented with 3 outlier removal methods of different types, as follows. For each method, we used its scikit-learn implementation and performed hyper-parameter tuning.
\begin{enumerate}[leftmargin=*]
\item \emph{Z-score}~\cite{kaliyaperumal2015outlier} with threshold parameter $\tau$: models $r\dbr{A}$ as samples from a normal distribution and removes points whose distance from the mean is over $\tau$ standard deviations. We experimented with $1\leq \tau \leq 6$.
\item \emph{Local outlier factor}~\cite{breunig2000lof}: samples with substantially lower density compared to the $k$ closest neighbors are considered outliers. We tried $k$ values in $\set{3,5,8,10}$.
\item \emph{Isolation forest}~\cite{liu2008isolation}: random decision trees are trained to isolate observations; easily isolated samples (with shorter path lengths in the decision tree) are considered outliers. A contamination parameter controls the number of expected outliers in the data. We tried contamination values between 0.001 and 0.2.
\end{enumerate}
For each method, we also examined a variation where the outlier removal is done per group $r_i$ versus on the full dataset.

All methods (with the various parameters) were applied on the datasets as a preliminary step. We then selected the ones that were closest to a repair - i.e., the ones with the minimal sum of violations after removal. On the resulting subset of the dataset for those method variations, we applied \dpalg to find a cardinality repair. 

The results are shown in \Cref{tab:outlier_removal_res_so_avg,tab:outlier_removal_res_hm_sum,tab:outlier_removal_res_hm_max}.
For each outlier removal method, we report the number of removed tuples, and the sum of violations after removal. For each set of variations of a method, we highlight (in bold) the variation that was closest to a repair (the minimal $S_{\mathsf{MVI}}$, and minimal number of removed tuples in the case of a tie). For the variations that were closest to satisfying the \aod, we also report the number of additional tuple removals required to reach a monotonic subset.

% Z-score~\cite{kaliyaperumal2015outlier}, Local Outlier Factor~\cite{breunig2000lof}, and Isolation Forest~\cite{liu2008isolation}, with hyper parameter tuning. 

%We performed a hyper parameter search for each method: we experimented with $1\leq \tau \leq 6$ for Z-score, $k$ For each outlier removal method we performed a hyper parameter search: for Z-score, we experimented with a score threshold (as the threshold is larger, less values are considered outliers). For Local Outlier Factor, which is a K-Nearest-Neighbors based method, we varied the number of neighbors K. As for isolation forest, we tried several values for the contamination parameter (the expected fraction of outliers in the dataset). For each method, we also examined a variation where the outlier removal is done per group $r_i$ versus on the full dataset. 

%The results are shown in \Cref{tab:outlier_removal_res}. For each outlier removal method, we report the number of removed tuples, and the sum of violations after removal. For each set of variations of a method, we highlight (in bold) the variation that was closest to a repair (the minimal $S_{\mathsf{MVI}}$, and minimal number of removed tuples in the case of a tie). For these best variations, we ran the DP algorithm to find a cardinality repair after the outliers were removed.

\begin{table}
\caption{Outlier removal results for stack overflow with $\alpha=\eavgagg$. For each method, the variation with the minimal sum of violations (and minimal number of removed tuples) is in bold. Removed+ is the number of additional tuple removals required to reach a monotonic subset after outlier removal.}
\label{tab:outlier_removal_res_so_avg}
\footnotesize
\begin{tabular}{llllll}
\toprule
\textbf{Method} & \textbf{Param.} & \textbf{Groupwise} & \textbf{\# Removed} & \textbf{$S_\mathsf{MVI}$} & \textbf{Removed+} \\
\midrule

Z score & 1 & No & 1480 & 6304 & \\
Z score & 1.5 & No & 1235 & 3680 & \\
Z score & 2 & No & 1037 & 0 & \\
Z score & 2.5 & No & 865 & 0 & \\
Z score & 3 & No & 724 & 0 & \\
Z score & 3.5 & No & 630 & 0 &\\
Z score & 4 & No & 544 & 0 &\\
Z score & 5 & No & 425 & 0 &\\
Z score & 6 & No & \textbf{369} & \textbf{0}  & 0\\
\hline
LOF & 3 & No & 318 & \textbf{25667} & 38\\
LOF & 5 & No & 412 & 32736  & \\
LOF & 8 & No & 417 & 35204 &\\
LOF & 10 & No & 465 & 31542 & \\
\hline
Iso.Forest & 0.001 & No & 23 & 17182 & \\
Iso.Forest & 0.005 & No & 23 & 17182 & \\
Iso.Forest & 0.01 & No & 23 & 17182 & \\
Iso.Forest & 0.05 & No & 1519 & \textbf{6515} & 7\\
Iso.Forest & 0.1 & No & 3125 & 8542 & \\
Iso.Forest & 0.2 & No & 6252 & 10569 & \\
\hline
Z score & 1 & Yes & 1509 & 11200 & \\
Z score & 1.5 & Yes & 1244 & 7073 & \\
Z score & 2 & Yes & 1062 & 2787 & \\
Z score & 2.5 & Yes & 888 & 6252 & \\
Z score & 3 & Yes & 733 & 2166 & \\
Z score & 3.5 & Yes & 638 & 4229 & \\
Z score & 4 & Yes & 558 & 2119 & \\
Z score & 5 & Yes & 441 & 3764 & \\
Z score & 6 & Yes & 373 & \textbf{0} & 0\\
\hline
LOF & 3 & Yes & 600 & 29850 & \\
LOF & 5 & Yes & 766 & 23598 & \\
LOF & 8 & Yes & 972 & \textbf{2142} & 18\\
LOF & 10 & Yes & 1055 & 3609 & \\
\hline
Iso.Forest & 0.001 & Yes & 23 & 17182 & \\
Iso.Forest & 0.005 & Yes & 83 & 15653 & \\
Iso.Forest & 0.01 & Yes & 179 & 36211 & \\
Iso.Forest & 0.05 & Yes & 1556 & \textbf{4424} & 5\\
Iso.Forest & 0.1 & Yes & 3112 & 5053 & \\
Iso.Forest & 0.2 & Yes & 6126 & 9060 & \\
\bottomrule
\end{tabular}
\end{table}

\begin{table}
\caption{Outlier removal results for H\&M with $\alpha=\emaxagg$. For each method, the variation with the minimal sum of violations (and minimal number of removed tuples) is in bold.}
\label{tab:outlier_removal_res_hm_max}
\footnotesize
\begin{tabular}{llllll}
\toprule
\textbf{Method} & \textbf{Param.} & \textbf{Groupwise} & \textbf{\# Removed} & \textbf{$S_\mathsf{MVI}$} & \textbf{Removed+} \\
\midrule
Z score & 1   & No & 462973 & 0 \\
Z score & 1.5 & No & 179157 & 0 \\
Z score & 2   & No & 84041 & 0 \\
Z score & 2.5 & No & 54412 & 0 \\
Z score & 3   & No & 21512 & 0 \\
Z score & 3.5 & No & 18258 & 1 \\
Z score & 4   & No & 8378  & 4 \\
Z score & 5   & No & 3158  & 4 \\
Z score & 6   & No & \textbf{1356}  & \textbf{0} & 0\\
\hline
Iso.Forest & 0.001 & No & 1262  & 4 \\
Iso.Forest & 0.005 & No & \textbf{8489}  & \textbf{0} & 0 \\
Iso.Forest & 0.01  & No & 18540 & 0 \\
Iso.Forest & 0.05  & No & 92774 & 0 \\
Iso.Forest & 0.1   & No & 181400 & 0 \\
Iso.Forest & 0.2   & No & 367546 & 0 \\
\hline
Z score & 1   & Yes  & 117349 & 458 \\
Z score & 1.5 & Yes  & 117349 & 577 \\
Z score & 2   & Yes  & 83664  & \textbf{4} & 1270 \\
Z score & 2.5 & Yes  & 54489  & 5 \\
Z score & 3   & Yes  & 21335  & 6 \\
Z score & 3.5 & Yes  & 18105  & 7 \\
Z score & 4   & Yes  & 9019   & 10 \\
Z score & 5   & Yes  & 3970   & 17 \\
Z score & 6   & Yes  & 1712   & 13 \\
\hline
Iso.Forest & 0.001 & Yes & 1385  & 47 \\
Iso.Forest & 0.005 & Yes & 8348  & 26 \\
Iso.Forest & 0.01  & Yes & 16657 & \textbf{13} & 4601\\
Iso.Forest & 0.05  & Yes & 83907 & 16 \\
Iso.Forest & 0.1   & Yes & 117349 & 577 \\
Iso.Forest & 0.2   & Yes & 117349 & 458 \\
\bottomrule
\end{tabular}
\end{table}

\begin{table}
\caption{Outlier removal results for H\&M with $\alpha=\esumagg$. For each method, the variation with the minimal sum of violations (and minimal number of removed tuples) is in bold.}
\label{tab:outlier_removal_res_hm_sum}
\footnotesize
\begin{tabular}{llllll}
\toprule
\textbf{Method} & \textbf{Param.} & \textbf{Groupwise} & \textbf{\# Removed} & \textbf{$S_\mathsf{MVI}$} & \textbf{Removed+} \\
\midrule

Z score & 1   & No  & 462973 & \textbf{88112} & 2797\\
Z score & 1.5 & No  & 179157 & 94748 \\
Z score & 2   & No  & 84041  & 106035 \\
Z score & 2.5 & No  & 54412  & 106576 \\
Z score & 3   & No  & 21512  & 106172 \\
Z score & 3.5 & No  & 18258  & 105808 \\
Z score & 4   & No  & 8378   & 106379 \\
Z score & 5   & No  & 3158   & 105009 \\
Z score & 6   & No  & 1356   & 105355 \\
\hline
Iso.Forest & 0.001 & No & 1262   & 105841 \\
Iso.Forest & 0.005 & No & 8489   & 105950 \\
Iso.Forest & 0.01  & No & 18540  & 106113 \\
Iso.Forest & 0.05  & No & 92774  & 107011 \\
Iso.Forest & 0.1   & No & 181400 & 107873 \\
Iso.Forest & 0.2   & No & 367546 & \textbf{99800} & 2270 \\
\midrule
Z score & 1   & Yes & 117349 & 94262 \\
Z score & 1.5 & Yes & 117349 & \textbf{89214} & 236279 \\
Z score & 2   & Yes & 83664  & 105083 \\
Z score & 2.5 & Yes & 54489  & 104848 \\
Z score & 3   & Yes & 21335  & 105236 \\
Z score & 3.5 & Yes & 18105  & 105568 \\
Z score & 4   & Yes & 9019   & 113748 \\
Z score & 5   & Yes & 3970   & 105007 \\
Z score & 6   & Yes & 1712   & 105355 \\
\hline
Iso.Forest & 0.001 & Yes & 1385   & 105355 \\
Iso.Forest & 0.005 & Yes & 8348   & 105950 \\
Iso.Forest & 0.01  & Yes & 16657  & \textbf{83644} & 1426 \\
Iso.Forest & 0.05  & Yes & 83907  & 111784 \\
Iso.Forest & 0.1   & Yes & 117349 & 95054 \\
Iso.Forest & 0.2   & Yes & 117349 & 151932 \\

\bottomrule
\end{tabular}
\end{table}

Most variations of outlier removal methods did not yield a monotonic subset (namely, a valid solution). 
For the \sumagg scenario, none of the methods yielded a monotonic subset, all of them removed over 30 times more tuples than necessary, and \dpalg had to remove a similar number of tuples as without the preliminary outlier detection.
For the \avgagg and \maxagg scenarios, some of the outlier removal methods did return a monotonic subset satisfying the \aod, but they required removing 10 times more tuples than the \crepair.

Still, outlier removal can occasionally accelerate \dpalg. For example, for SO with \avgagg, Isolation forest with 0.05 contamination (in both the group-wise and the non group-wise settings), there were only a few remaining tuples to remove. In both cases, \greedyalg
yielded tight bounds, thus greatly shortening the runtimes of \dpalg (60.5 seconds and 115.9 seconds, respectively, versus 1082 seconds).

Overall, we conclude that \emph{repairing for an \aod 
typically requires more than just applying outlier detection, 
as the tuples that obscure the trend are not necessarily statistical outliers}. They may appear structurally standard, but still interfere with the expected trend.

%% file: 10-app-case_studies.tex
\section{Additional Case Studies}\label{sec:case_studies_plus}
\paragraph{H\&M}
For H\&M we considered the sum of purchases grouped by age between the ages of 25 and 40, which exhibited an almost perfect decrease. Both the heuristic and the exact algorithm found a \crepair for the \aod $age \searrow \esumagg(\mbox{price})$ by removing 1464 tuples (0.07\%). A \crepair for the reversed \aod $age \nearrow \esumagg(\mbox{price})$ would require removing 741,302 tuples (39\%), meaning that indeed the dataset is much closer to supporting $age \searrow \esumagg(\mbox{price})$.

\paragraph{German credit}
We focused on the average credit risk score over time that a person has been employed in their place of work. We would expect that as the employment time is longer, the person is more likely to be labeled as a good credit risk. However, the rate for unemployed (63\%) was higher than for people employed under 1 year (59\%), and the rate for people with over 7 years of employment in their place of work (75\%) was lower than for people with 4-7 years of employment (78\%). Both \greedyalg and \dpalg found a \crepair by removing 16 (1.6\%) tuples, removing people with over 7 years of employment, and unemployed individuals.%\jonny{the "and unemployed" is mildly confusing.} 
A \crepair for the opposite direction \aod required removing 90 (9\%) tuples.

\paragraph{Stack Overflow}
\begin{figure}[b]
    \centering

    \includegraphics[width=\linewidth]{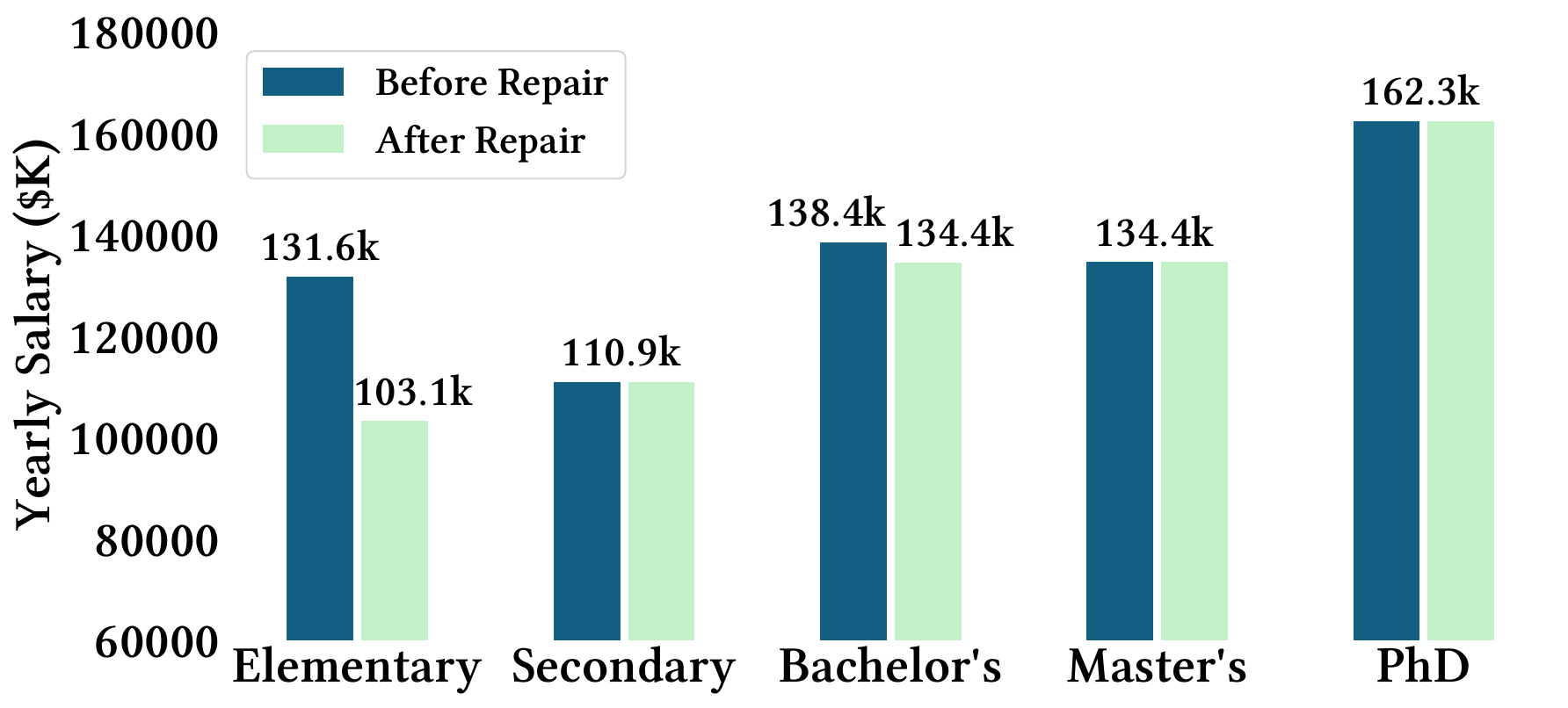}
    \caption{Average yearly salary over education level, according to the Stack Overflow dataset.}
    \label{fig:so_avg_case_study}
\end{figure}

For the \aod $\mbox{edu}\nearrow\eavgagg(\mbox{income})$ (\Cref{fig:so_avg_case_study}), we considered the categories primary school, secondary school, Bachelor's degree, Master's degree, and PhD (excluding professional degrees like MD, JD). The original dataset does not satisfy the \aod, as the average salary for primary school is too high (\$131.6K versus \$110.9K for secondary school) and the salary for Bachelor's degree (\$138.4K) is higher than for Master's (\$134.4K). \dpalg and \greedyalg removed the exact same subset of 37 tuples (0.12\% of the data), but \greedyalg's runtime was much faster (0.61 seconds versus 1082.42 seconds for the most optimized version of \dpalg). For the opposite \aod, the \greedyalg removed 403 tuples (1.3\% of the dataset), but we were unable to run \dpalg, since even with pruning to keep only repairs that remove up to 403 tuples, the algorithm would take days to run.

%would require removing 403 tuples \ag{($1.2\%$ of the dataset)} (10 times more than in the other direction). This is the result of the greedy algorithm, %4.41 seconds

\hide{
\subsubsection{H\&M age}
Consider the sum of H\&M purchase prices grouped by age. As mentioned previously, we notice an interesting phenomenon (\Cref{fig:hm_gb_result}) where there is a peak at the age of 25, and then a clear decline until the age of 40, followed by another, lower peak at 50. 
\begin{figure}
    \centering
    \includegraphics[width=0.5\linewidth]{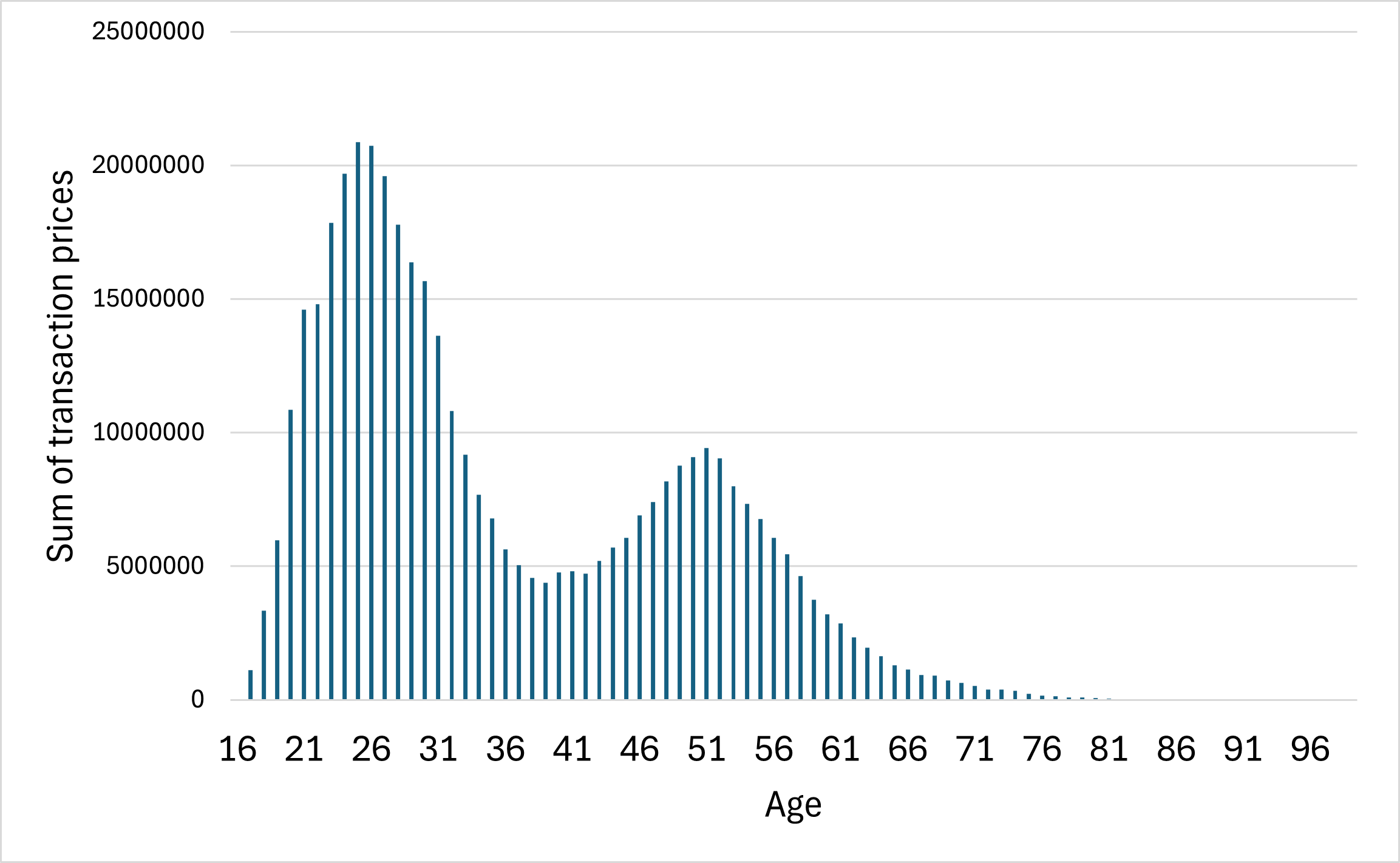}
    \caption{H\&M sum of transaction prices grouped by age.}
    \label{fig:hm_gb_result}
\end{figure}

However, in the time period of May to July that we focused on, the trend is slightly violated by an early increase before 40 (at 38 \red{TODO verify}). We can use the AOD repair to quantify the distance from a database where the AOD is satisfied. We ran the DP algorithm and found that removing 1464 tuples (which are 0.07\% of the dataset) is enough to satisfy the AOD. In comparison, to repair for a reversed AOD that requires that the sum of purchases increases with age between ages 25 and 40, would require to remove 741,302 tuples (39\%). \red{TODO what is special about the removed tuples?}

\subsubsection{H\&M price}
Question: where does most of the revenue come from? Cheap or expensive items? Does revenue increase as the price increases or decreases? Consider the total revenue (sum of transaction prices) from each price bucket \Cref{fig:hm_price_bin_gb_result}. It seems that none of the hypothesis is correct, but which one is closer to reality?

\begin{figure}
    \centering
    \includegraphics[width=0.5\linewidth]{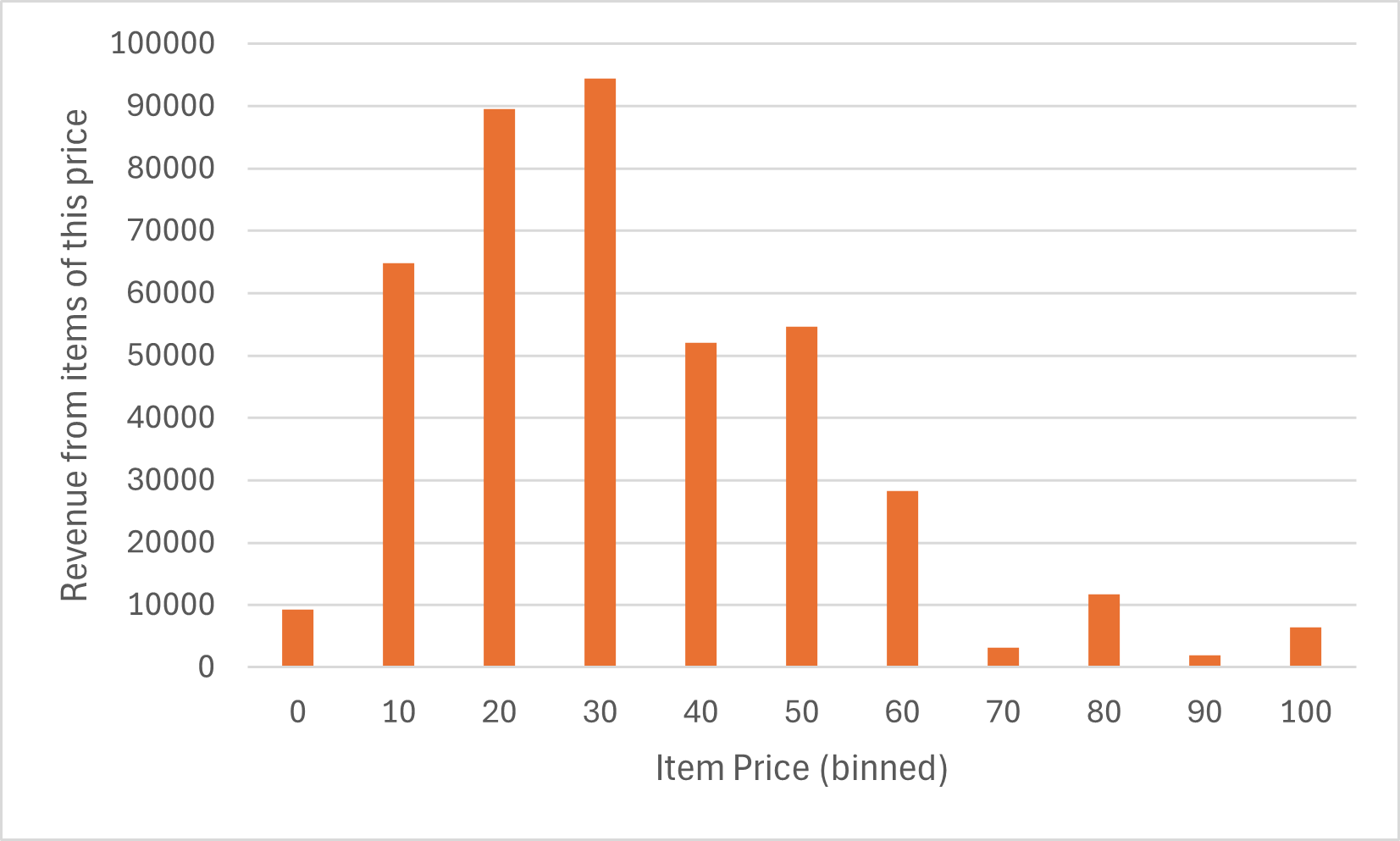}
    \caption{H\&M sum of transaction prices grouped by the item price.}
    \label{fig:hm_price_bin_gb_result}
\end{figure}

When considering the \aod requiring an increase in revenue as the item price grows, the cardinality repair size is 2,869,924 (out of 15,039,894). % No h -pruning. agg pack knapsack optimization. - took 3720.56 seconds to run.
The removed tuples were the full groups 40-100. 
}

%% file: 10-app-greedy_median.tex
\section{Pitfalls of the Heuristic Algorithm for Median}\label{sec:median_difficulty}
%\benny{Why do we have a label here? There is no number here... If you reference this label later, then you have the wrong reference. This part can be in a remark, I guess, with the example written inline.}

%\benny{By the ``greedy approach'' you mean the greedy algorithm? \Cref{alg:greedy}? Why is it an "approach" here? Is it true for every greedy algorithm?}

%\benny{Why is this a part of "6.2 Optimizations"?}

One of the key challenges with \greedyalg for \medianagg is its inefficiency in removing tuples. %\brit{tuples?}
The main issue is that for \medianagg, the impact of the removed tuple depends on the other tuples in the group (specifically those in proximity of the current median). In these cases, the greedy approach that only sees a single step ahead is likely to make suboptimal choices.
    
    % % WITH IMPACTS in the table
    % \begin{table}[H]
    % \caption{Greedy Algorithm Example for Median Aggregation}
    % \label{tab:example_greedy_median}
    % \centering
    % \begin{tabular}{|c|c|c|c|}
    %     \hline
    %     \rowcolor{gray!30} \textbf{Impact (Iteration)} & \textbf{$r_1$} & \textbf{$r_2$} & \textbf{Impact} \\ \hline
    %     0.5 (0) & \cellcolor{red!30} 11 & 10 & 0 \\ \hline
    %     0.5 (1) & \cellcolor{red!30} 10 & 5  & 0 \\ \hline
    %     0.5 (2) & \cellcolor{red!30} 9  & \cellcolor{green!30} 5  & 0 \\ \hline
    %     0.5 (3) & \cellcolor{red!30} 8  & \cellcolor{green!30} 5  & 0 \\ \hline
    %     0.5 (4) & \cellcolor{red!30} 7  & \cellcolor{green!30} 1  & 0 \\ \hline
    % \end{tabular}
    % \end{table}

    % WITHOUT IMPACTS in the table

\begin{figure}
  \begin{center}
    \centering
    \begin{tabular}{|c|c|}
        \hline
        \rowcolor{gray!30} \textbf{$r_1$} & \textbf{$r_2$} \\ \hline
        \cellcolor{red!30} 11 & 10 \\ \hline
        \cellcolor{red!30} 10 & 5 \\ \hline
        \cellcolor{red!30} \textbf{9}  & \cellcolor{green!30} \textbf{5} \\ \hline
        \cellcolor{red!30} 8  & \cellcolor{green!30} 5 \\ \hline
        \cellcolor{red!30} 7  & \cellcolor{green!30} 1 \\ \hline
    \end{tabular}
    \end{center}
     \caption{ \label{tab:example_greedy_median}Heuristic algorithm example for \medianagg aggregation.}
\end{figure}

    \begin{example}
    Consider a relation $r$ consisting of two groups, $r_1$ with $G$ value $1$ and $r_2$ with $G$ value $2$. The two columns $r_1,r_2$ in Table~\ref{tab:example_greedy_median} represent the bag of values in each group (e.g.,~the value $11$ in $r_1$ represents the tuple $(1,11)$). We consider the \aod: $G\nearrow \emedianagg(A)$. %\ag{What is the query and constraint?}.
    Initially, $\emedianagg(r_1\dbr{A})=9$ and $\emedianagg(r_2\dbr{A})=5$. The impact of removing each single tuple from $r_2$ is $0$, since there are multiple instances of the median value. 
    The algorithm then chooses to remove the tuple $(1,11)$ that has impact $0.5$ on $S_{\textsf{MVI}}$, as after the removal, $\emedianagg(r_1\dbr{A})=9$.
    The situation is repeated in each iteration of the algorithm until all five tuples in $r_1$ are removed (marked in red), where an optimal solution is to remove the three lowest-value tuples from $r_2$ (marked in green). \qed 
    \end{example}

%\benny{Too much text for the pitfall of median. Also - this is a simple heuristic... how do we explain that we make no attempt to solve the pitfall? We just remark on it and go on? The reader would expect us to give a solution, especially given that we do not provide any guarantees.}

A possible improvement to \greedyalg could consider one group $r_i$ at a time (instead of one tuple at a time) and try to reduce the violations of this group to 0. For median, given $\mathrm{LB} = \emedianagg(r_{i-1}\dbr{A})$ and $\mathrm{UB}=\emedianagg(r_{i+1}\dbr{A})$, it is possible to devise a closed formula to find the maximum-sized subset $r'_i\subseteq r_i$ such that $\mathrm{LB} \leq \emedianagg(r_i\dbr{A})\leq \mathrm{UB}$.

%% file: 10-graphs_with_CIs.tex
\section{Graphs with Additional Information}\label{sec:graphs_with_CIs}
In \Cref{sec:experiments} we have omitted the confidence intervals from \Cref{fig:runtime_per_aggregation_synth}, to make them more accessible. We include them here in \Cref{fig:runtime_per_aggregation_synth_plus_CI}.

For the comparison between aggregations (\Cref{fig:runtime_per_aggregation_synth}) we have included \Cref{fig:runtime_per_aggregation_synth_plus_CI}.

\begin{figure}[b]
    \centering
    \subfloat[Increasing \# tuples.\label{fig:scale_rows_CI}]{
        \includegraphics[width=0.54\linewidth]{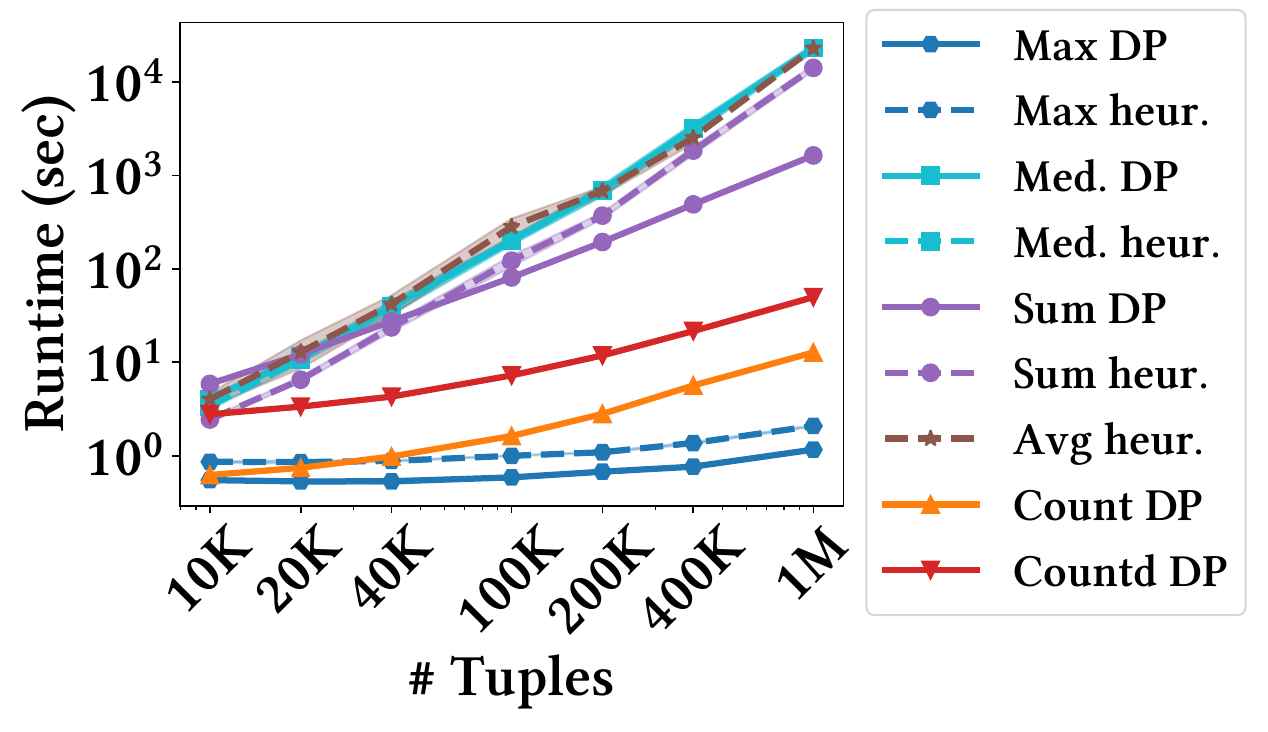}
    }
    % \hfill
    \subfloat[Increasing \# \addtuples.\label{fig:scale_violations_CI}]{
        \includegraphics[width=0.42\linewidth]{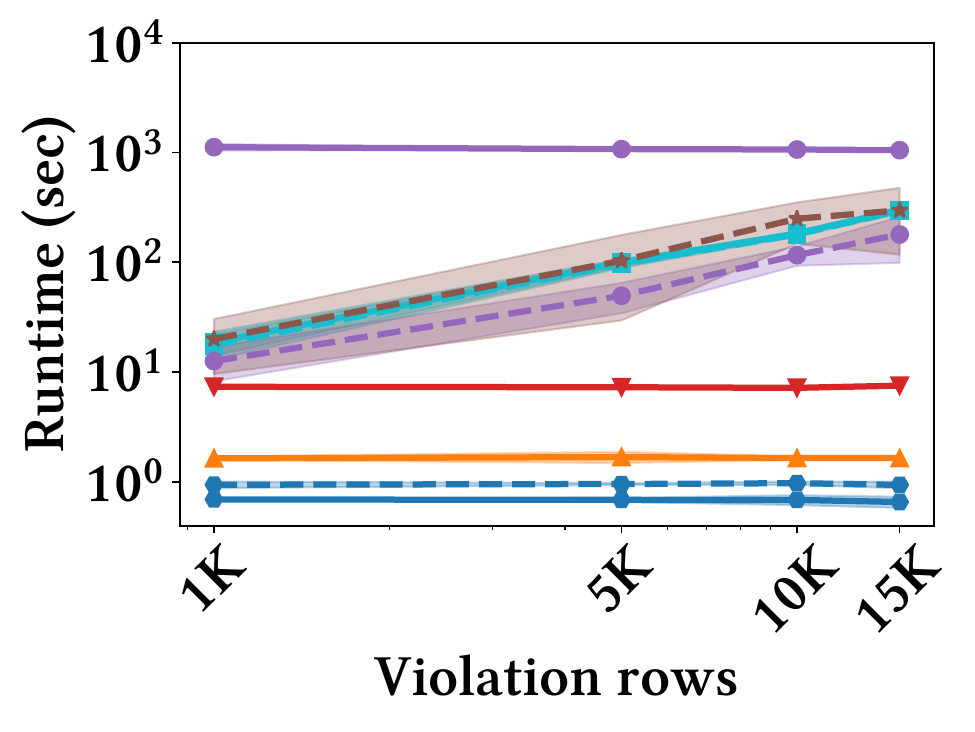}
    }
    \caption{Run times of \dpalg and \greedyalg with various aggregations over synthetic data, with 95\% confidence intervals. This is a fuller version of \Cref{fig:runtime_per_aggregation_synth}.}
    \label{fig:runtime_per_aggregation_synth_plus_CI}
\end{figure}

%% file: trends.bib
@String{Computing = "Computing" }

@String{Computer = "{IEEE} Computer" }

@String{Springer = "Springer-Verlag" }

@article{lin2021detecting,
  title={On detecting cherry-picked generalizations},
  author={Lin, Yin and Youngmann, Brit and Moskovitch, Yuval and Jagadish, HV and Milo, Tova},
  journal={Proceedings of the VLDB Endowment},
  volume={15},
  number={1},
  pages={59--71},
  year={2021},
  publisher={VLDB Endowment}
}

@article{patton2010monotonicity,
  title={Monotonicity in asset returns: New tests with applications to the term structure, the CAPM, and portfolio sorts},
  author={Patton, Andrew J and Timmermann, Allan},
  journal={Journal of Financial Economics},
  volume={98},
  number={3},
  pages={605--625},
  year={2010},
  publisher={Elsevier}}

@inproceedings{bhattacharyya2011testing,
  title={Testing monotonicity of distributions over general partial orders.},
  author={Bhattacharyya, Arnab and Fischer, Eldar and Rubinfeld, Ronitt and Valiant, Paul},
  booktitle={ICS},
  pages={239--252},
  year={2011},
  organization={Citeseer}
}

@article{hussian2005monotonic,
  title={Monotonic regression for the detection of temporal trends in environmental quality data},
  author={Hussian, Mohamed and Grimvall, Anders and Burdakov, Oleg and Sysoev, Oleg},
  journal={MATCH Commun. Math. Comput. Chem},
  volume={54},
  pages={535--550},
  year={2005}
}

@article{lee2009testing,
  title={Testing for stochastic monotonicity},
  author={Lee, Sokbae and Linton, Oliver and Whang, Yoon-Jae},
  journal={Econometrica},
  volume={77},
  number={2},
  pages={585--602},
  year={2009},
  publisher={Wiley Online Library}
}

@article{ghosal2000testing,
  title={Testing monotonicity of regression},
  author={Ghosal, Subhashis and Sen, Arusharka and Van Der Vaart, Aad W},
  journal={Annals of statistics},
  pages={1054--1082},
  year={2000},
  publisher={JSTOR}
}

@article{hall2000testing,
  title={Testing for monotonicity of a regression mean by calibrating for linear functions},
  author={Hall, Peter and Heckman, Nancy E},
  journal={Annals of Statistics},
  pages={20--39},
  year={2000},
  publisher={JSTOR}
}

@inproceedings{roy2014formal,
  title={A formal approach to finding explanations for database queries},
  author={Roy, Sudeepa and Suciu, Dan},
  booktitle={Proceedings of the 2014 ACM SIGMOD international conference on Management of data},
  pages={1579--1590},
  year={2014}
}

@article{wu2013scorpion,
  title={Scorpion: Explaining Away Outliers in Aggregate Queries},
  author={Wu, Eugene and Madden, Samuel},
  journal={Proceedings of the VLDB Endowment},
  volume={6},
  number={8},
  year={2013}
}

@article{roy2015explaining,
  title={Explaining query answers with explanation-ready databases},
  author={Roy, Sudeepa and Orr, Laurel and Suciu, Dan},
  journal={Proceedings of the VLDB Endowment},
  volume={9},
  number={4},
  pages={348--359},
  year={2015},
  publisher={VLDB Endowment}
}

@article{langer2016order_deps,
  title={Efficient order dependency detection},
  author={Langer, Philipp and Naumann, Felix},
  journal={The VLDB Journal},
  volume={25},
  pages={223--241},
  year={2016},
  publisher={Springer}
}

@inproceedings{consonni2019order_deps,
  title={Discovering Order Dependencies through Order Compatibility.},
  author={Consonni, Cristian and Sottovia, Paolo and Montresor, Alberto and Velegrakis, Yannis and others},
  booktitle={EDBT},
  pages={409--420},
  year={2019}
}

@article{holoclean,
author = {Rekatsinas, Theodoros and Chu, Xu and Ilyas, Ihab F. and R\'{e}, Christopher},
title = {HoloClean: holistic data repairs with probabilistic inference},
year = {2017},
issue_date = {August 2017},
publisher = {VLDB Endowment},
volume = {10},
number = {11},
issn = {2150-8097},
url = {https://doi.org/10.14778/3137628.3137631},
doi = {10.14778/3137628.3137631},
journal = {Proc. VLDB Endow.},
month = {aug},
pages = {1190–1201},
numpages = {12}
}

@inproceedings{DBLP:conf/icdt/CarmeliGKLT21,
  author       = {Nofar Carmeli and
                  Martin Grohe and
                  Benny Kimelfeld and
                  Ester Livshits and
                  Muhammad Tibi},
  editor       = {Ke Yi and
                  Zhewei Wei},
  title        = {Database Repairing with Soft Functional Dependencies},
  booktitle    = {24th International Conference on Database Theory, {ICDT} 2021, March
                  23-26, 2021, Nicosia, Cyprus},
  series       = {LIPIcs},
  volume       = {186},
  pages        = {16:1--16:17},
  publisher    = {Schloss Dagstuhl - Leibniz-Zentrum f{\"{u}}r Informatik},
  year         = {2021}
}

@article{geerts2013llunatic,
  title={The LLUNATIC data-cleaning framework},
  author={Geerts, Floris and Mecca, Giansalvatore and Papotti, Paolo and Santoro, Donatello},
  journal={Proceedings of the VLDB Endowment},
  volume={6},
  number={9},
  pages={625--636},
  year={2013},
  publisher={VLDB Endowment}
}

@inproceedings{bohannon2006conditional,
  title={Conditional functional dependencies for data cleaning},
  author={Bohannon, Philip and Fan, Wenfei and Geerts, Floris and Jia, Xibei and Kementsietsidis, Anastasios},
  booktitle={2007 IEEE 23rd international conference on data engineering},
  pages={746--755},
  year={2006},
  organization={IEEE}
}

@article{wijsen2001trends,
  title={Trends in databases: Reasoning and mining},
  author={Wijsen, Jef},
  journal={IEEE Transactions on Knowledge and Data Engineering},
  volume={13},
  number={3},
  pages={426--438},
  year={2001},
  publisher={IEEE}
}

@misc{german_credit,
  author       = {Hofmann, Hans},
  title        = {{Statlog (German Credit Data)}},
  year         = {1994},
  howpublished = {UCI Machine Learning Repository},
  note         = {{DOI}: https://doi.org/10.24432/C5NC77}
}

@article{baker2008housing,
  title={The housing bubble and the financial crisis},
  author={Baker, Dean},
  journal={Real-world economics review},
  volume={46},
  number={20},
  pages={73--81},
  year={2008},
  publisher={Center for Economic and Policy Research}
}

@InProceedings{ibrahim2018,
author="Ibrahim, Ibrahim A.
and Li, Xue
and Zhao, Xin
and Maskari, Sanad Al
and Albarrak, Abdullah M.
and Zhang, Yanjun",
editor="Phung, Dinh
and Tseng, Vincent S.
and Webb, Geoffrey I.
and Ho, Bao
and Ganji, Mohadeseh
and Rashidi, Lida",
title="Automated Explanations of User-Expected Trends for Aggregate Queries",
booktitle="Advances in Knowledge Discovery and Data Mining",
year="2018",
publisher="Springer International Publishing",
address="Cham",
pages="602--614",
isbn="978-3-319-93034-3"
}

@InProceedings{OD_repair2018,
author="Qiu, Yu
and Tan, Zijing
and Yang, Kejia
and Yang, Weidong
and Zhou, Xiangdong
and Guo, Naiwang",
editor="Pei, Jian
and Manolopoulos, Yannis
and Sadiq, Shazia
and Li, Jianxin",
title="Repairing Data Violations with Order Dependencies",
booktitle="Database Systems for Advanced Applications",
year="2018",
publisher="Springer International Publishing",
address="Cham",
pages="283--300",
isbn="978-3-319-91458-9"
}

@article{chomicki2005minimal,
  title={Minimal-change integrity maintenance using tuple deletions},
  author={Chomicki, Jan and Marcinkowski, Jerzy},
  journal={Information and Computation},
  volume={197},
  number={1-2},
  pages={90--121},
  year={2005},
  publisher={Elsevier}
}

@inproceedings{chu2013holistic,
  author={Chu, Xu and Ilyas, Ihab F. and Papotti, Paolo},
  booktitle={2013 IEEE 29th International Conference on Data Engineering (ICDE)}, 
  title={Holistic data cleaning: Putting violations into context}, 
  year={2013},
  volume={},
  number={},
  pages={458-469},
  doi={10.1109/ICDE.2013.6544847}
}

@article{flesca2010querying,
  title={Querying and repairing inconsistent numerical databases},
  author={Flesca, Sergio and Furfaro, Filippo and Parisi, Francesco},
  journal={ACM Transactions on Database Systems (TODS)},
  volume={35},
  number={2},
  pages={1--50},
  year={2010},
  publisher={ACM New York, NY, USA}
}

@inproceedings{flesca2007preferred,
  title={Preferred database repairs under aggregate constraints},
  author={Flesca, Sergio and Furfaro, Filippo and Parisi, Francesco},
  booktitle={Scalable Uncertainty Management: First International Conference, SUM 2007, Washington, DC, USA, October 10-12, 2007. Proceedings 1},
  pages={215--229},
  year={2007},
  organization={Springer}
}

@book{flesca2011repairing,
  title={Repairing and querying databases under aggregate constraints},
  author={Flesca, Sergio and Furfaro, Filippo and Parisi, Francesco},
  year={2011},
  publisher={Springer Science \& Business Media}
}

@article{DBLP:journals/tcs/GinsburgH83,
  author       = {Seymour Ginsburg and
                  Richard Hull},
  title        = {Order Dependency in the Relational Model},
  journal      = {Theor. Comput. Sci.},
  volume       = {26},
  pages        = {149--195},
  year         = {1983}
}

@book{DBLP:series/synthesis/2012Fan,
  author       = {Wenfei Fan and
                  Floris Geerts},
  title        = {Foundations of Data Quality Management},
  series       = {Synthesis Lectures on Data Management},
  publisher    = {Morgan {\&} Claypool Publishers},
  year         = {2012}
}

@inproceedings{DBLP:conf/sigmod/LivshitsKTIKR21,
  author       = {Ester Livshits and
                  Rina Kochirgan and
                  Segev Tsur and
                  Ihab F. Ilyas and
                  Benny Kimelfeld and
                  Sudeepa Roy},
  title        = {Properties of Inconsistency Measures for Databases},
  booktitle    = {{SIGMOD} Conference},
  pages        = {1182--1194},
  publisher    = {{ACM}},
  year         = {2021}
}

@inproceedings{DBLP:conf/lpnmr/Bertossi19,
  author       = {Leopoldo E. Bertossi},
  title        = {Repair-Based Degrees of Database Inconsistency},
  booktitle    = {{LPNMR}},
  series       = {Lecture Notes in Computer Science},
  volume       = {11481},
  pages        = {195--209},
  publisher    = {Springer},
  year         = {2019}
}

@article{10.1145/3725397,
author = {Mohapatra, Shubhankar and Gilad, Amir and He, Xi and Kimelfeld, Benny},
title = {Computing Inconsistency Measures Under Differential Privacy},
year = {2025},
issue_date = {June 2025},
publisher = {Association for Computing Machinery},
address = {New York, NY, USA},
volume = {3},
number = {3},
url = {https://doi.org/10.1145/3725397},
journal = {Proc. ACM Manag. Data},
month = jun,
articleno = {140},
numpages = {27}
}

@inproceedings{DBLP:conf/icdt/KolahiL09,
  author       = {Solmaz Kolahi and
                  Laks V. S. Lakshmanan},
  title        = {On approximating optimum repairs for functional dependency violations},
  booktitle    = {{ICDT}},
  series       = {{ACM} International Conference Proceeding Series},
  volume       = {361},
  pages        = {53--62},
  publisher    = {{ACM}},
  year         = {2009}
}

@inproceedings{salimi2019interventional,
  title={Interventional fairness: Causal database repair for algorithmic fairness},
  author={Salimi, Babak and Rodriguez, Luke and Howe, Bill and Suciu, Dan},
  booktitle={Proceedings of the 2019 International Conference on Management of Data},
  pages={793--810},
  year={2019}
}

@article{wu2017computational,
  title={Computational fact checking through query perturbations},
  author={Wu, You and Agarwal, Pankaj K and Li, Chengkai and Yang, Jun and Yu, Cong},
  journal={ACM Transactions on Database Systems (TODS)},
  volume={42},
  number={1},
  pages={1--41},
  year={2017},
  publisher={ACM New York, NY, USA}
}

@article{wu2014toward,
  title={Toward computational fact-checking},
  author={Wu, You and Agarwal, Pankaj K and Li, Chengkai and Yang, Jun and Yu, Cong},
  journal={Proceedings of the VLDB Endowment},
  volume={7},
  number={7},
  pages={589--600},
  year={2014},
  publisher={VLDB Endowment}
}

@article{jo2019aggchecker,
  title={Aggchecker: A fact-checking system for text summaries of relational data sets},
  author={Jo, Saehan and Trummer, Immanuel and Yu, Weicheng and Wang, Xuezhi and Yu, Cong and Liu, Daniel and Mehta, Niyati},
  journal={VLDB},
  volume={12},
  number={12},
  pages={1938--1941},
  year={2019},
  publisher={VLDB Endowment}
}

@article{asudeh2020detecting,
  title={On detecting cherry-picked trendlines},
  author={Asudeh, Abolfazl and Jagadish, Hosagrahar Visvesvaraya and Wu, You and Yu, Cong},
  journal={VLDB},
  volume={13},
  number={6},
  pages={939--952},
  year={2020},
  publisher={VLDB Endowment}
}

@article{DBLP:journals/pvldb/AgmonGYZK24,
  author       = {Shunit Agmon and
                  Amir Gilad and
                  Brit Youngmann and
                  Shahar Zoarets and
                  Benny Kimelfeld},
  title        = {Finding Convincing Views to Endorse a Claim},
  journal      = {Proc. {VLDB} Endow.},
  volume       = {18},
  number       = {2},
  pages        = {439--452},
  year         = {2024},
  url          = {https://www.vldb.org/pvldb/vol18/p439-agmon.pdf},
  bibsource    = {dblp computer science bibliography, https://dblp.org}
}

@article{asudeh2021perturbation,
  title={Perturbation-based Detection and Resolution of Cherry-picking},
  author={Asudeh, Abolfazl and Wu, You Will and Yu, Cong and Jagadish, HV},
  journal={A Quarterly bulletin of the Computer Soc.~ of the IEEE Technical Committee on Data Engineering},
  volume={45},
  number={3},
  pages={39-51},
  year={2021}
}

@misc{relbench,
      title={RelBench: A Benchmark for Deep Learning on Relational Databases},
      author={Joshua Robinson and Rishabh Ranjan and Weihua Hu and Kexin Huang and Jiaqi Han and Alejandro Dobles and Matthias Fey and Jan E. Lenssen and Yiwen Yuan and Zecheng Zhang and Xinwei He and Jure Leskovec},
      year={2024},
      eprint={2407.20060},
      archivePrefix={arXiv},
      primaryClass={cs.LG},
      url={https://arxiv.org/abs/2407.20060},
}

@article{puri1990recursive,
  title={On recursive formulas for isotonic regression useful for statistical inference under order restrictions},
  author={Puri, Prem S and Singh, Harshinder},
  journal={Journal of statistical planning and inference},
  volume={24},
  number={1},
  pages={1--11},
  year={1990},
  publisher={Elsevier}
}

@article{ramsay1998estimating,
  title={Estimating smooth monotone functions},
  author={Ramsay, James O},
  journal={Journal of the Royal Statistical Society: Series B (Statistical Methodology)},
  volume={60},
  number={2},
  pages={365--375},
  year={1998},
  publisher={Wiley Online Library}
}

@article{brunk1955maximum,
  title={Maximum likelihood estimates of monotone parameters},
  author={Brunk, Hugh D},
  journal={The Annals of Mathematical Statistics},
  pages={607--616},
  year={1955},
  publisher={JSTOR}
}

@misc{abramovich2025advancingfactattributionquery,
      title={Advancing Fact Attribution for Query Answering: Aggregate Queries and Novel Algorithms}, 
      author={Omer Abramovich and Daniel Deutch and Nave Frost and Ahmet Kara and Dan Olteanu},
      year={2025},
      eprint={2506.16923},
      archivePrefix={arXiv},
      primaryClass={cs.DB},
      url={https://arxiv.org/abs/2506.16923}, 
}

@inproceedings{dong1982,
author = {Dong, Jirun and Hull, Richard},
title = {Applying approximate order dependency to reduce indexing space},
year = {1982},
isbn = {0897910737},
publisher = {Association for Computing Machinery},
address = {New York, NY, USA},
url = {https://doi.org/10.1145/582353.582375},
doi = {10.1145/582353.582375},
booktitle = {Proceedings of the 1982 ACM SIGMOD International Conference on Management of Data},
pages = {119–127},
numpages = {9},
location = {Orlando, Florida},
series = {SIGMOD '82}
}

@article{ginsburg1986,
author = {Ginsburg, Seymour and Hull, Richard},
title = {Sort sets in the relational model},
year = {1986},
issue_date = {July 1986},
publisher = {Association for Computing Machinery},
address = {New York, NY, USA},
volume = {33},
number = {3},
issn = {0004-5411},
url = {https://doi.org/10.1145/5925.5929},
doi = {10.1145/5925.5929},
journal = {J. ACM},
month = may,
pages = {465–488},
numpages = {24}
}

@article{ng1999,
title = {Ordered functional dependencies in relational databases},
journal = {Information Systems},
volume = {24},
number = {7},
pages = {535-554},
year = {1999},
issn = {0306-4379},
doi = {https://doi.org/10.1016/S0306-4379(99)00031-9},
url = {https://www.sciencedirect.com/science/article/pii/S0306437999000319},
author = {Wilfred Ng},
}

@article{ng2001,
author = {Ng, Wilfred},
title = {An extension of the relational data model to incorporate ordered domains},
year = {2001},
issue_date = {September 2001},
publisher = {Association for Computing Machinery},
address = {New York, NY, USA},
volume = {26},
number = {3},
issn = {0362-5915},
url = {https://doi.org/10.1145/502030.502033},
doi = {10.1145/502030.502033},
journal = {ACM Trans. Database Syst.},
month = sep,
pages = {344–383},
numpages = {40},
}

@INPROCEEDINGS{jin2020,
  author={Jin, Yifeng and Zhu, Lin and Tan, Zijing},
  booktitle={2020 IEEE 36th International Conference on Data Engineering (ICDE)}, 
  title={Efficient Bidirectional Order Dependency Discovery}, 
  year={2020},
  volume={},
  number={},
  pages={61-72},
  doi={10.1109/ICDE48307.2020.00013}}

@article{DBLP:journals/tods/LivshitsKR20,
  author       = {Ester Livshits and
                  Benny Kimelfeld and
                  Sudeepa Roy},
  title        = {Computing Optimal Repairs for Functional Dependencies},
  journal      = {{ACM} Trans. Database Syst.},
  volume       = {45},
  number       = {1},
  pages        = {4:1--4:46},
  year         = {2020}
}

@article{DBLP:journals/vldb/MiaoZLWC23,
  author       = {Dongjing Miao and
                  Pengfei Zhang and
                  Jianzhong Li and
                  Ye Wang and
                  Zhipeng Cai},
  title        = {Approximation and inapproximability results on computing optimal repairs},
  journal      = {{VLDB} J.},
  volume       = {32},
  number       = {1},
  pages        = {173--197},
  year         = {2023}
}

@inproceedings{DBLP:conf/icdt/GiladIK23,
  author       = {Amir Gilad and
                  Aviram Imber and
                  Benny Kimelfeld},
  title        = {The Consistency of Probabilistic Databases with Independent Cells},
  booktitle    = {{ICDT}},
  series       = {LIPIcs},
  volume       = {255},
  pages        = {22:1--22:19},
  publisher    = {Schloss Dagstuhl - Leibniz-Zentrum f{\"{u}}r Informatik},
  year         = {2023}
}

@inproceedings{DBLP:conf/sigmod/BohannonFFR05,
  author       = {Philip Bohannon and
                  Michael Flaster and
                  Wenfei Fan and
                  Rajeev Rastogi},
  title        = {A Cost-Based Model and Effective Heuristic for Repairing Constraints
                  by Value Modification},
  booktitle    = {{SIGMOD} Conference},
  pages        = {143--154},
  publisher    = {{ACM}},
  year         = {2005}
}

@inproceedings{DBLP:conf/icdt/KaminskyKLNW25,
  author       = {Youri Kaminsky and
                  Benny Kimelfeld and
                  Ester Livshits and
                  Felix Naumann and
                  David Wajc},
  title        = {Repairing Databases over Metric Spaces with Coincidence Constraints},
  booktitle    = {{ICDT}},
  series       = {LIPIcs},
  volume       = {328},
  pages        = {14:1--14:18},
  publisher    = {Schloss Dagstuhl - Leibniz-Zentrum f{\"{u}}r Informatik},
  year         = {2025}
}

@inproceedings{DBLP:conf/foiks/MahmoodVBN24,
  author       = {Yasir Mahmood and
                  Jonni Virtema and
                  Timon Barlag and
                  Axel{-}Cyrille Ngonga Ngomo},
  title        = {Computing Repairs Under Functional and Inclusion Dependencies via
                  Argumentation},
  booktitle    = {FoIKS},
  series       = {Lecture Notes in Computer Science},
  volume       = {14589},
  pages        = {23--42},
  publisher    = {Springer},
  year         = {2024}
}

@article{DBLP:journals/iandc/ChomickiM05,
  author       = {Jan Chomicki and
                  Jerzy Marcinkowski},
  title        = {Minimal-change integrity maintenance using tuple deletions},
  journal      = {Inf. Comput.},
  volume       = {197},
  number       = {1-2},
  pages        = {90--121},
  year         = {2005}
}

@book{DBLP:series/synthesis/2011Bertossi,
  author       = {Leopoldo E. Bertossi},
  title        = {Database Repairing and Consistent Query Answering},
  series       = {Synthesis Lectures on Data Management},
  publisher    = {Morgan {\&} Claypool Publishers},
  year         = {2011}
}

@book{pearl2009causality,
  title={Causality},
  author={Pearl, Judea},
  year={2009},
  publisher={Cambridge university press}
}

@inproceedings{liu2008isolation,
  title={Isolation forest},
  author={Liu, Fei Tony and Ting, Kai Ming and Zhou, Zhi-Hua},
  booktitle={2008 eighth ieee international conference on data mining},
  pages={413--422},
  year={2008},
  organization={IEEE}
}

@inproceedings{breunig2000lof,
  title={LOF: identifying density-based local outliers},
  author={Breunig, Markus M and Kriegel, Hans-Peter and Ng, Raymond T and Sander, J{\"o}rg},
  booktitle={Proceedings of the 2000 ACM SIGMOD international conference on Management of data},
  pages={93--104},
  year={2000}
}

@article{kaliyaperumal2015outlier,
  title={Outlier detection and missing value in time series ozone data},
  author={Kaliyaperumal, Senthamarai Kannan and Kuppusamy, Manoj and Gounder, Arumugam Subbanna},
  journal={International Journal of Scientific Research in Knowledge},
  volume={3},
  number={9},
  pages={220--226},
  year={2015}
}

@inproceedings{DBLP:conf/sigmod/GiladDR20,
  author       = {Amir Gilad and
                  Daniel Deutch and
                  Sudeepa Roy},
  editor       = {David Maier and
                  Rachel Pottinger and
                  AnHai Doan and
                  Wang{-}Chiew Tan and
                  Abdussalam Alawini and
                  Hung Q. Ngo},
  title        = {On Multiple Semantics for Declarative Database Repairs},
  booktitle    = {Proceedings of the 2020 International Conference on Management of
                  Data, {SIGMOD} Conference 2020, online conference [Portland, OR, USA],
                  June 14-19, 2020},
  pages        = {817--831},
  publisher    = {{ACM}},
  year         = {2020},
  url          = {https://doi.org/10.1145/3318464.3389721},
  doi          = {10.1145/3318464.3389721},
  bibsource    = {dblp computer science bibliography, https://dblp.org}
}

@article{xiong2006enhancing,
  title={Enhancing data analysis with noise removal},
  author={Xiong, Hui and Pandey, Gaurav and Steinbach, Michael and Kumar, Vipin},
  journal={IEEE transactions on knowledge and data engineering},
  volume={18},
  number={3},
  pages={304--319},
  year={2006},
  publisher={IEEE}
}

@article{borrohou2023data,
  title={Data cleaning survey and challenges--improving outlier detection algorithm in machine learning},
  author={Borrohou, Sanae and Fissoune, Rachida and Badir, Hassan},
  journal={Journal of Smart Cities and Society},
  volume={2},
  number={3},
  pages={125--140},
  year={2023},
  publisher={SAGE Publications Sage UK: London, England}
}

@InProceedings{ovchinnik2019,
    author="Ovchinnik, Sergey
        and Otero, Fernando E. B.
        and Freitas, Alex A.",
    editor="Bramer, Max
        and Petridis, Miltos",
    title="Monotonicity Detection and Enforcement in Longitudinal Classification",
    booktitle="Artificial Intelligence XXXVI",
    year="2019",
    publisher="Springer International Publishing",
    address="Cham",
    pages="63--77",
    isbn="978-3-030-34885-4"
}

@article{zhou2016,
    title = {An empirical study of Bayesian network parameter learning with monotonic influence constraints},
    journal = {Decision Support Systems},
    volume = {87},
    pages = {69-79},
    year = {2016},
    issn = {0167-9236},
    doi = {https://doi.org/10.1016/j.dss.2016.05.001},
    url = {https://www.sciencedirect.com/science/article/pii/S0167923616300744},
    author = {Yun Zhou and Norman Fenton and Cheng Zhu},
}

@article{gonzalez2024,
    author = {Gonz\'{a}lez-Almagro, Germ\'{a}n and S\'{a}nchez-Bermejo, Pablo and Suarez, Juan Luis and Cano, Jos\'{e}-Ram\'{o}n and Garc\'{\i}a, Salvador},
    title = {Semi-supervised clustering with two types of background knowledge: Fusing pairwise constraints and monotonicity constraints},
    year = {2024},
    issue_date = {Feb 2024},
    publisher = {Elsevier Science Publishers B. V.},
    address = {NLD},
    volume = {102},
    number = {C},
    issn = {1566-2535},
    url = {https://doi.org/10.1016/j.inffus.2023.102064},
    doi = {10.1016/j.inffus.2023.102064},
    journal = {Inf. Fusion},
    month = feb,
    numpages = {15},
}

@article{martello1984mixture,
  title={A mixture of dynamic programming and branch-and-bound for the subset-sum problem},
  author={Martello, Silvano and Toth, Paolo},
  journal={Management Science},
  volume={30},
  number={6},
  pages={765--771},
  year={1984},
  publisher={INFORMS}
}

@article{poirriez2009,
title = {A hybrid algorithm for the unbounded knapsack problem},
journal = {Discrete Optimization},
volume = {6},
number = {1},
pages = {110-124},
year = {2009},
issn = {1572-5286},
doi = {https://doi.org/10.1016/j.disopt.2008.09.004},
url = {https://www.sciencedirect.com/science/article/pii/S1572528608000686},
author = {Vincent Poirriez and Nicola Yanev and Rumen Andonov},
}
